\documentclass[prodmode,acmtocl]{acmsmall}
\let\orgsetcounter\setcounter 

\usepackage{xspace}
\usepackage{epsfig}
\usepackage{latexsym,amssymb,stmaryrd}
\usepackage{fancybox,epic}
\usepackage{amsmath}
\usepackage{color} 
\usepackage{graphics}
\usepackage{epsfig}
\usepackage{graphicx}
\usepackage[draft]{fixme}
\usepackage{tikz}
\usepackage{wrapfig}
\usepackage{supertabular}
\usepackage{textcomp}
\usepackage{gensymb}
\usepackage{pxfonts}
\usepackage{cancel}

\newcommand{\citet}[1]{\cite{#1}}
\newcommand{\citep}[1]{\cite{#1}}
\newcommand{\citeauthor}[1]{\cite{#1}}

\usepackage{listings}
\usepackage{fancyhdr} 
  
\newcommand{\cF}[0]{{\cal F}}

\newcommand{\cI}[0]{{\cal I}}

\newcommand{\cK}[0]{{\cal K}}
\newcommand{\cL}[0]{{\cal L}}

\newcommand{\cN}[0]{{\cal N}}
\newcommand{\cO}[0]{{\cal O}}
\newcommand{\cP}[0]{{\cal P}}

\newcommand{\cR}[0]{{\cal R}}

\newcommand{\cV}[0]{{\cal V}}

\newcommand{\cX}[0]{{\cal X}}
\newcommand{\cY}[0]{{\cal Y}}

\def\SHARINGBASED{\ensuremath{\mathcal{Y}_{sh}}\xspace}
\def\REACHBASED{\ensuremath{\mathcal{Y}_{rc}}\xspace}
\def\FIELDBASED{\ensuremath{\mathcal{Y}_{\field}}\xspace}

\newcommand{\body}[1]{\ensuremath{{#1}^b}}
\newcommand{\inp}[1]{\ensuremath{{#1}^i}}

\newcommand{\locals}[1]{\ensuremath{{#1}^l}}
\newcommand{\scope}[1]{\ensuremath{{#1}^s}}
\newcommand{\nil}{\mbox{\lstinline!null!}\xspace}
\newcommand{\xx}{\ensuremath{\mathsf{x}}\xspace} 

\newcommand{\xtwo}{\ensuremath{\mathsf{x2}}\xspace}
\newcommand{\yy}{\ensuremath{\mathsf{y}}\xspace}
\newcommand{\zz}{\ensuremath{\mathsf{z}}\xspace}
\newcommand{\mm}{\ensuremath{\mathsf{m}}\xspace}
\newcommand{\nn}{\ensuremath{\mathsf{n}}\xspace}
\newcommand{\pp}{\ensuremath{\mathsf{p}}\xspace}

\newcommand{\tmp}{\ensuremath{\mathsf{tmp}}\xspace}
\newcommand{\varl}{\ensuremath{\mathsf{l}}\xspace}
\newcommand{\varr}{\ensuremath{\mathsf{r}}\xspace}
\newcommand{\vart}{\ensuremath{\mathsf{t}}\xspace}

\newcommand{\integer}{\mbox{\lstinline!int!}\xspace}
\newcommand{\newk}[1]{\ensuremath{\mbox{\lstinline!new!}~#1}}
\newcommand{\skipc}{\mbox{\lstinline!skip!}\xspace}
\newcommand{\assign}{\mbox{\lstinline!:=!}}
\newcommand{\ifte}{\mbox{\lstinline!if!}\xspace}
\newcommand{\iftethen}{\mbox{\lstinline!then!}\xspace}
\newcommand{\ifteelse}{\mbox{\lstinline!else!}\xspace}
\newcommand{\while}{\mbox{\lstinline!while!}\xspace}
\newcommand{\whilebody}{\mbox{\lstinline!do!}\xspace}
\newcommand{\return}{\mbox{\lstinline!return!}\xspace}
\newcommand{\this}{\ensuremath{{\mathord{\mathit{this}}}}}
\newcommand{\out}{\ensuremath{{\mathord{\mathit{out}}}}}
\newcommand{\true}{\ensuremath{{\mathord{\mathit{true}}}}}
\newcommand{\false}{\ensuremath{{\mathord{\mathit{false}}}}}
\newcommand{\variables}{\ensuremath{\cX}}
\newcommand{\identifiers}{\ensuremath{\cX}}
\newcommand{\classes}{\ensuremath{\cK}}
\newcommand{\class}{\ensuremath{\kappa}}
\newcommand{\objects}{\ensuremath{\cO}}
\newcommand{\subclass}{\prec}
\newcommand{\subclasseq}{\preceq}
\newcommand{\fields}{\ensuremath{\cF}}

\def\field{\ensuremath{\mathit{fld}\xspace}}
\def\methoddef{\ensuremath{\mathit{M}\xspace}}
\def\method{\ensuremath{\mathit{mth}\xspace}}
\def\methodsig{\ensuremath{\mathsf{mth}\xspace}}

\newcommand{\values}{\ensuremath{\cV}}
\newcommand{\locations}{\ensuremath{\cL}}
\def\semnil{\ensuremath{\mathit{null}}}
\newcommand{\heap}{\ensuremath{\mu}}
\newcommand{\objtag}[1]{\ensuremath{{#1}.\mathsf{tag}}}
\newcommand{\objframe}[1]{\ensuremath{{#1}.\mathsf{frm}}}
\newcommand{\newobj}[1]{\ensuremath{\mathit{newobj}(#1)}}
\newcommand{\newloc}[1]{\ensuremath{\mathit{newloc}(#1)}}
\newcommand{\statesym}{\ensuremath{\sigma}}
\newcommand{\state}[2]{\ensuremath{\tuple{#1,#2}}}
\newcommand{\states}[1]{\ensuremath{\Sigma_{#1}}}
\newcommand{\frm}{\ensuremath{\phi}}
\newcommand{\statef}[1]{\ensuremath{{#1}^f}}
\newcommand{\statem}[1]{\ensuremath{{#1}^h}}

\newcommand{\den}{\ensuremath{\delta}}
\newcommand{\denset}[2]{\ensuremath{\Delta(#1{,}#2)}}
\newcommand{\einter}[3]{\ensuremath{{E_{#1} ^{#2}}{\left\llbracket{#3}\right\rrbracket}}}
\newcommand{\cinter}[3]{\ensuremath{{C_{#1}^{#2}}{\left\llbracket{#3}\right\rrbracket}}}

\newcommand{\seq}{\ensuremath{\mapsto}}
\newcommand{\lookup}{\ensuremath{\mathit{lookup}}}

\newcommand{\tp}[1]{\ensuremath{T_{#1}}}

\newcommand{\project}[2]{\ensuremath{\exists{#2}.{#1}}}
\newcommand{\extend}[2]{\ensuremath{ext({#1},{#2})}}
\newcommand{\res}{\ensuremath{\rho}}
\newcommand{\abstp}[1]{\ensuremath{\mathcal{T}_{#1}}}
\newcommand{\typenv}{\ensuremath{\tau}}
\newcommand{\interpretations}{\ensuremath{\Gamma}}
\def\hpath{\ensuremath{\pi\xspace}}
\def\hpconcat{\ensuremath{\cdot\xspace}}

\usetikzlibrary{calc}
\newcommand{\hcancel}[1]{%
    \tikz[baseline=(tocancel.base)]{
        \node[inner sep=0pt,outer sep=0pt] (tocancel) {#1};
        \draw[->] ($(tocancel.south west)+(-1pt,1pt)$) -- ($(tocancel.north east)+(3pt,-1pt)$);
    }%
}%

\def\bool{\ensuremath{\mathit{Bool}}}
\def\finterpretation{\ensuremath{\omega\xspace}}
\newcommand{\fmodels}[1]{\ensuremath{\mathit{models}_{\fpropositions}}(#1)}
\newcommand{\fproposition}[1]{\mbox{\hcancel{\ensuremath{#1}}}}
\def\fpropositions{\ensuremath{\cP}}
\def\pformulae{\ensuremath{\cP\!\cF}}
\def\pformulaeq{\ensuremath{\pformulae_\equiv}}

\def\pformula{\ensuremath{F\xspace}}
\def\pformulag{\ensuremath{G\xspace}}

\newcommand{\onlyfields}[1]{\langle #1\rangle}
\newcommand{\onlyfieldsp}[1]{\langle\fproposition{#1}\rangle}
\def\pformulaempty{\onlyfields{\emptyset}}

\newcommand{\SHARE}[2]{\tuple{#1{\curlyveedownarrow}#2}}

\newcommand{\ALIAS}[2]{\tuple{#1{\cdot}#2}}
\newcommand{\abselema}{\ensuremath{I_a}}
\newcommand{\abselemsp}{\ensuremath{I_{sp}}}
\newcommand{\PURE}[1]{\ensuremath{\dot{#1}}}
\newcommand{\shden}[1]{\ensuremath{\mathsf{SP}_{#1}}}
\def\domalias{\ensuremath{\cI}_a}

\newcommand{\reachable}[2]{\ensuremath{R(#1,#2)}}
\newcommand{\reachablei}[3]{\ensuremath{R^{#1}(#2,#3)}}

\newcommand{\reachableplus}[2]{\ensuremath{R^{+}(#1,#2)}}

\newcommand{\REACHES}[2]{#1{\rightsquigarrow}#2}

\newcommand{\condom}{\ensuremath{\cI_\flat^\typenv}}
\newcommand{\domr}{\ensuremath{\cI_r^\typenv}}

\newcommand{\domc}{\ensuremath{\cI_c^\typenv}}

\newcommand{\domrc}[1]{\ensuremath{\cI_{rc}^{#1}}}

\newcommand{\concelem}{\ensuremath{I_\flat}}
\newcommand{\abselemr}{\ensuremath{I_r}}

\newcommand{\abselemrc}{\ensuremath{I_{rc}}}
\newcommand{\traverses}[2]{#1\!\twoheadrightarrow\!#2}
\newcommand{\ntraverses}[2]{#1\!\ntwoheadrightarrow\!#2}

\newcommand{\UPDATE}[3]{#1\left[#2{\leftarrow}#3\right]}

\def\fmark#1{\bar{#1}}

\def\sqcupf{\ensuremath{\fmark{\sqcup}}}

\def\topfr{\ensuremath{\fmark{\top}_r}}
\def\botfr{\ensuremath{\fmark{\bot}_r}}
\def\sqcapfr{\ensuremath{\fmark{\sqcap}_r}}
\def\sqcupfr{\ensuremath{\fmark{\sqcup}_r}}
\def\sqsubseteqfr{\ensuremath{\fmark{\sqsubseteq}_r}}
\def\sqsupseteqfr{\ensuremath{\fmark{\sqsupseteq}_r}}
\def\topfc{\ensuremath{\fmark{\top}_c}}
\def\botfc{\ensuremath{\fmark{\bot}_c}}
\def\sqcapfc{\ensuremath{\fmark{\sqcap}_c}}
\def\sqcupfc{\ensuremath{\fmark{\sqcup}_c}}
\def\sqsubseteqfc{\ensuremath{\fmark{\sqsubseteq}_c}}

\def\domfr{\ensuremath{\fmark{\cI}_r}}
\def\gammafr{\ensuremath{\fmark{\gamma}_{\mathit r}}}
\def\alphafr{\ensuremath{\fmark{\alpha}_{\mathit r}}}
\def\domfc{\ensuremath{\fmark{\cI}_c}}
\def\gammafc{\ensuremath{\fmark{\gamma}_{\mathit c}}}
\def\alphafc{\ensuremath{\fmark{\alpha}_{\mathit c}}}
\def\domfrc{\ensuremath{\fmark{\cI}_{rc}}}
\newcommand{\domfrctau}[1]{\ensuremath{\fmark{\cI}_{rc}^{#1}}}
\def\gammafrc{\ensuremath{\fmark{\gamma}_{\mathit rc}}}

\def\abselemfr{\ensuremath{\fmark{I}_r}}
\def\abselemfc{\ensuremath{\fmark{I}_c}}
\def\abselemfrc{\ensuremath{\fmark{I}_{rc}}}
\def\sqsubseteqfrc{\ensuremath{\fmark{\sqsubseteq}_{rc}}}

\newcommand{\reachsetf}[1]{\ensuremath{\fmark{\cR}}}
\newcommand{\cycsetf}[1]{\ensuremath{\fmark{\cY}}}
\newcommand{\normalize}[1]{\ensuremath{\cN\left(#1\right)}}

\newcommand{\domkr}{{\ensuremath{^{\class}\!{\cI_r}}}}
\newcommand{\domkc}{{\ensuremath{^{\class}\!{\cI_c}}}}
\newcommand{\abselemkr}{{\ensuremath{^{\class}\!{I}_r}}}
\def\dectype{\delta}

\newcommand{\domxr}[1]{\ensuremath{{^{\textsc{#1}}\!{\cI_r^\typenv}}}}

\newcommand{\reachsetxr}[1]{\ensuremath{{^{\textsc{#1}}\!{\cR^{\typenv}}}}}
\newcommand{\sqsubseteqxr}[1]{\ensuremath{{^{\textsc{#1}}\!{\sqsubseteq}_r}}}

\newcommand{\botxr}[1]{\ensuremath{{^{\textsc{#1}}}\!\!{\bot}_r}}
\newcommand{\topxr}[1]{\ensuremath{{^{\textsc{#1}}}\!{\top}_r}}
\newcommand{\sqcapxr}[1]{\ensuremath{{^{\textsc{#1}}}\!{\sqcap}_r}}
\newcommand{\sqcupxr}[1]{\ensuremath{{^{\textsc{#1}}}\!{\sqcup}_r}}

\newcommand{\abselemxr}[1]{\ensuremath{{^{\textsc{#1}}\!{I}_r}}}

\def\dommr{\domxr{m}}

\def\reachsetmr{\reachsetxr{m}}
\def\sqsubseteqmr{\sqsubseteqxr{m}}

\def\botmr{\botxr{m}}
\def\topmr{\topxr{m}}
\def\sqcapmr{\sqcapxr{m}}
\def\sqcupmr{\sqcupxr{m}}

\def\abselemmr{\abselemxr{m}}

\def\domdr{\domxr{d}}

\def\dompr{\domxr{p}}

\newcommand{\domnr}{\ensuremath{\mathsf{A}_\typenv}}
\newcommand{\reachsetnr}{\ensuremath{\mathsf{UR}_\typenv}}
\def\sqsubseteqnr{\ensuremath{\sqsubseteq}_A}

\def\sqcapnr{\ensuremath{\sqcap}_A}
\def\sqcupnr{\ensuremath{\sqcup}_A}

\newcommand{\abselemnr}{\ensuremath{I_A}}

\newcommand{\domqc}{\ensuremath{\mathsf{Q}_\typenv}}
\newcommand{\cycsetqc}{\ensuremath{\mathsf{UC}_\typenv}}
\def\sqsubseteqqc{\ensuremath{\sqsubseteq}_Q}

\def\sqcapqc{\ensuremath{\sqcap}_Q}
\def\sqcupqc{\ensuremath{\sqcup}_Q}

\newcommand{\abselemqc}{\ensuremath{I_Q}}

\newcommand{\ASEMANTICS}[3]{\ensuremath{\mathcal{C}_{#2}}{\left\llbracket{#3}\right\rrbracket}}
\newcommand{\EXPASEMANTICS}[3]{\ensuremath{\mathcal{E}_{#2}}{\left\llbracket{#3}\right\rrbracket}}
\newcommand{\absinterp}{\ensuremath{\zeta}}
\newcommand{\absden}{\ensuremath{\xi}}
\newcommand{\absdenset}[2]{\ensuremath{\Xi({#1},{#2})}}
\newcommand{\absinterpretations}{\ensuremath{\Psi}}

\DeclareMathOperator{\interp}{\iota}
\DeclareMathOperator{\dom}{dom}
\DeclareMathOperator{\codom}{rng}

\newcommand{\fname}[1]{\ensuremath{\mathsf{#1}}\xspace}
\def\ffield{\fname{f}}
\def\gfield{\fname{g}}
\def\hfield{\fname{h}}

\def\employee{\fname{Emp}}
\def\levelone{\fname{L1}}
\def\leveltwo{\fname{L2}}
\def\device{\fname{Dev}}
\def\laptop{\fname{LP}}
\def\tablet{\fname{TB}}
\def\md{\fname{mD}}
\def\ad{\fname{aD}}
\def\lnk{\fname{lnk}}
\def\owner{\fname{owner}}

\def\pnext{\ensuremath{\fproposition{\mathsf{n}}}\xspace}
\def\pprev{\ensuremath{\fproposition{\mathsf{p}}}\xspace}
\def\pleft{\ensuremath{\mathring{\mbox{\sc{l}}}}\xspace}
\def\pright{\ensuremath{\mathring{\mbox{\sc{r}}}}\xspace}
\def\pparent{\ensuremath{\mathring{\mbox{\sc{p}}}}\xspace}

\def\any{\ensuremath{\mathit{any}}\xspace}

\newcommand{\tuple}[1]{\left\langle #1 \right\rangle}

\def\lfp{\ensuremath{\mathit{lfp}}\xspace}
\newcommand{\CUT}[1]{}

\newcommand{\ouracks}{This work was funded partially by the European
  research project FP7-ICT-610582 ENVISAGE: Engineering Virtualized
  Services (website: \texttt{http://www.envisage-project.eu}), and by
  the Spanish projects TIN2008-05624 and TIN2012-38137.}

\title{Inference of Field-Sensitive\\Reachability and Cyclicity}

\author{
  DAMIANO ZANARDINI \affil{Technical University of Madrid (UPM), Spain}
  and SAMIR GENAIM \affil{Complutense University of Madrid (UCM), Spain}
  }

\begin{abstract}
  In heap-based languages, knowing that a variable $\xx$ points to an
  acyclic data structure is useful for analyzing termination: this
  information guarantees that the depth of the data structure to which
  $\xx$ points is greater than the depth of the structure pointed to
  by $\xx.\field$, and allows bounding the number of iterations of a
  loop which traverses the data structure on $\field$.

  In general, proving termination needs acyclicity, unless
  program-specific or non-automated reasoning is performed. However,
  recent work could prove that certain loops terminate even without
  inferring acyclicity, because they traverse data structures
  ``acyclically''. Consider a double-linked list: if it is possible to
  demonstrate that every cycle involves both the ``next'' and the
  ``prev'' field, then a traversal on ``next'' terminates since no
  cycle will be traversed completely.

  This paper develops a static analysis inferring field-sensitive
  reachability and cyclicity information, which is more general than
  existing approaches. Propositional formul\ae~are computed, which
  describe which fields may or may not be traversed by paths in the
  heap. Consider a tree with edges ``left'' and ``right'' to the left
  and right sub-trees, and ``parent'' to the parent node: termination
  of a loop traversing leaf-up cannot be guaranteed by
  state-of-the-art analyses. Instead, propositional formul\ae~computed
  by this analysis indicate that cycles must traverse ``parent'' and
  at least one between ``left'' and ``right'': termination is
  guaranteed as no cycle is traversed completely.

  This paper defines the necessary abstract domains and builds an
  abstract semantics on them. A prototypical implementation provides
  the expected result on relevant examples.
\end{abstract}




\acmformat{Damiano Zanardini and Samir Genaim. 2014. Inference of
  Field-Sensitive Reachability and Cyclicity.}

\begin{document}
{\let\setcounter\orgsetcounter
\begin{bottomstuff}
  Authors' addresses: Damiano Zanardini, Departamento de Inteligencia
  Artificial, Escuela T\'ecnica Superior de Ingenieros Inform\'aticos,
  Campus de Montegancedo, Boadilla del Monte, 28660 Madrid, Spain;
  Samir Genaim, Departamento de Sistemas Inform\'aticos y
  Computaci\'on, Facultad de Inform\'atica, Universidad Complutense de
  Madrid, C/ Profesor Jos\'e Garc\'ia Santesmases s/n, 28040, Madrid,
  Spain.
\end{bottomstuff}
} 
\maketitle


\section{Introduction}
\label{sec:introduction}

Programming languages with dynamic memory allocation, such as Java,
allow creating and manipulating linked data structures in the
\emph{heap}.  The presence of cyclic data structures in the heap is a
challenging issue in the context of termination
analysis~\cite{DBLP:conf/cav/BerdineCDO06,DBLP:conf/pldi/CookPR06,AlbertACGPZ08,SpotoMP2010},
resource usage
analysis~\cite{Wegbreit75,caslog-short,AlbertAGPZ12,AlbertGM13},
garbage collection~\cite{JonesLins1996}, etc.
Consider the loop
``\lstinline^while (x!=null) do x:=x.f^'': if \xx points to an acyclic
data structure before the loop, then the depth of the data structure
to which \xx points strictly decreases after each iteration;
therefore, the number of iterations is bounded by the initial depth of
the structure.  On the other hand, in general, nothing can be said
about such a decrement if acyclicity cannot be demonstrated, unless
more complex, program-specific or non-automated reasoning is
performed.  This makes acyclicity information essential in order to
bound loop iterations and, by extension, prove termination.

In mainstream Object-Oriented programming languages, data structures
are usually modified by means of \emph{field updates}.  Consider
\lstinline!x.f:=y!: if \xx and \yy are guaranteed to point to
\emph{disjoint} parts of the heap \emph{before} the command, then
there is no possibility to create a cycle.  On the other hand, if they
are not disjoint, i.e., they \emph{share} a common part of the heap,
then a cyclic structure might be created.  This simple mechanism,
denoted in the following as \SHARINGBASED, has been used in previous
work~\cite{RossignoliS06-short}: \xx and \yy are declared as possibly
cyclic whenever they share before the update.  Refinements of
\SHARINGBASED have been proposed
\cite{DBLP:conf/popl/GhiyaH96,GenaimZ10,GenaimZ13,NikolicS14}, which
also consider the \emph{reachability} between program variables.  In
this example, the acyclicity information can be more precise if it is
possible to know \emph{how} \xx and \yy share: in general, it can be
the case that (1) \xx and \yy \emph{alias}, i.e., point directly to
the same location; (2) \xx \emph{reaches} the location pointed to by
\yy; (3) \yy \emph{reaches} the location pointed to by \xx; or (4)
they both indirectly reach a common location (here, this case is
referred to as \emph{deep sharing}, see Section
\ref{sec:auxiliaryAnalyses}).  The field update \lstinline!x.f:=y!
might create a cycle only in cases (1) or (3).  The latter approach is
able to prove acyclicity in cases like
``\lstinline!y:=x.next.next;x.next:=y;!''  (which typically removes an
element from a linked list), where the former fails.  For simplicity,
this technique will be denoted by \REACHBASED, ignoring discrepancies
between the different works implementing such a reachability-based
analysis.  \REACHBASED improves on \SHARINGBASED in that the class of
data structures which can be proved to be acyclic is larger.

However, recent research
\cite{ScapinSpotoMScThesis,DBLP:conf/cav/BrockschmidtMOG12} went one
step ahead by proving, in some cases, the termination of programs
\emph{even if} the data structures they traverse\footnote{The idea of
  \emph{traversing} fields will be defined precisely later, but can be
  understood as dereference.} are cyclic.  In fact, cycles often enjoy
certain properties which allow to guarantee that loops never traverse
them \emph{completely}.  Suppose that \xx points to a cyclic data
structure, and the loop
\lstinline^while (x!=null) do x:=x.f^
is supposed to traverse it.  Recent works were able to prove
termination if either (a) no cycle can involve \ffield
\cite{ScapinSpotoMScThesis}; or (b) cycles have to involve a set $X$
of fields which contains fields different from \ffield
\cite{DBLP:conf/cav/BrockschmidtMOG12}\footnote{This paper also deals
  with other cases of algorithms on cyclic data structures, which are
  beyond the scope of this discussion.}.

The present cyclicity analysis, denoted in the following by
\FIELDBASED, is more general and more precise than the above
approaches, and allows inferring \emph{field-sensitive} reachability
and cyclicity information which can be used to prove termination of a
wider class of programs.  The information inferred by \FIELDBASED
takes the form of \emph{propositional formul\ae} which indicate which
are the fields involved (1) in paths between two variables; or (2) in
cycles reachable from a variable.  A propositional formula can tell
that a field never occurs in cycle, or that it always occur, or that
its presence in cycles is conditional.  Consider the case of a
\lstinline!Tree! class implementing trees where each node has a
\lstinline!left! and a \lstinline!right! field pointing to its left
and right sub-trees, respectively, and a \lstinline!parent! field
pointing to the parent node.  Suppose also that two loops traverse the
tree (1) from the root to a leaf, by following a certain path; and (2)
from this leaf, back to the root.  The first loop traverses
\lstinline!left! and \lstinline!right! a certain number of times,
while the second only traverses \lstinline!parent!.  This kind of tree
is a cyclic data structure; however, it enjoys the property that every
cycle has to traverse \lstinline!parent! and at least one between
\lstinline!left! and \lstinline!right!.  Condition (a) above does not
hold for any of the loops, since they traverse fields which are
actually involved in cycles.  Condition (b) does not hold either,
since the only field which is involved in all possible loops is
\lstinline!parent!, but the second loop actually traverses it, so that
termination cannot be proved.  On the other hand, the propositional
formul\ae~computed by the present analysis represent the desired
cyclicity information which allows proving termination of both loops,
since it is possible to prove that they will never traverse a cycle
completely.  Another example of cyclic structure where cycles can
traverse several different sets of fields is a \emph{cyclic grid},
i.e., some kind of bidimensional double-linked list where each node
has \lstinline!left!, \lstinline!right!, \lstinline!up!, and
\lstinline!down! links to neighbour nodes.  This data structure has
cycles which traverse \lstinline!left! and \lstinline!right!, or
\lstinline!up! and \lstinline!down!, but also longer cycles traversing
all fields.

Following the well-known theory of \emph{Abstract Interpretation}
\cite{Cousot77}, the paper introduces \emph{abstract domains}
representing the properties of interest, and discusses their relation
with existing work.  A sound \emph{abstract semantics} is built on
these domains, which computes the desired reachability/cyclicity
information.  An intra-procedural subset of the abstract semantics has
been implemented, and gives the expected result on the examples
discussed in this introduction.

\subsubsection*{Main contributions}

The main contributions of the present paper are as follows:

\begin{itemize}
\item The paper defines abstract domains which capture field-sensitive
  reachability and cyclicity information in form of propositional
  formul\ae.
\item The domains are compared to related work and proved to be more
  precise.
\item A sound abstract semantics is built on the abstract domains.
\item The approach is partially implemented (only a subset of the
  intra-procedural component), and the expected result is obtained on
  relevant examples.
\end{itemize}


\subsection{Related work}
\label{sec:relatedWork}

The present paper is very related to research in the area of
\emph{Pointer analysis} \cite{Hind01}, which considers properties of
the heap and builds \emph{static analyses} to enforce them.  Clearly,
techniques which directly deal with the reachability and cyclicity
originated by paths in the heap represent the closest work in this
area.  Apart from that, \emph{Aliasing}, \emph{Sharing},
\emph{Points-to} and \emph{Shape} analysis are the most related
pointer analyses which can be found in the literature.

\emph{Termination analysis} is a well-established research area which
overlaps with Pointer analysis when heap-manipulating programming
languages are considered; it also has to be discussed as related work.
Finally, \emph{Resource-usage analysis} is also related because the
same results which are useful in order to prove termination can also
help in estimating the resource consumption of a program.

\subsubsection*{Pointer Analysis}

A well-known technique in Pointer analysis, \emph{Aliasing analysis}
\cite{Hind01} investigates the program variables which might point to
the same heap location at runtime.  \emph{Sharing analysis}
\cite{spoto:pair_sharing} is more general in that it determines if two
variables $v_1$ and $v_2$ can reach a common location in the heap,
i.e., if the portions of the heap which are reachable from $v_1$ and
$v_2$ are not disjoint.  Aliasing between two variables implies that
they also share.  \emph{Points-to analysis} computes the set of
objects which might be referred to by a pointer variable.

Research on \emph{Shape Analysis} \cite{WilhelmSR00-long} basically
reasons about heap-manipulating programs in order to prove program
properties.  In most cases, \emph{safety} properties are dealt with
\cite{BardinFN04,SagivRW02,RinetzkyBRSW05}.  On the other hand,
termination is a \emph{liveness} property, and is, typically, the
final property to be proved when analyzing cyclicity; therefore, work
on liveness
\cite{Reynolds02,BalabanPZ05,DBLP:conf/cav/BerdineCDO06,DBLP:conf/pldi/CookPR06,BrotherstonBC08-short}
is closer to the present approach.  Most papers use techniques based
on \emph{Model Checking} \cite{Muller-OlmSS99}, \emph{Predicate
  Abstraction} \cite{GrafS97}, \emph{Separation Logic}
\cite{Reynolds02} or \emph{Cyclic proofs} \cite{BrotherstonBC08-short}
in order to prove properties of programs manipulating the heap.
Typically, shape analyses capture aliasing and points-to information,
and build a representation of the heap from which reachability
information can be obtained.  Such analyses are very precise,
sometimes at the cost of (i) limiting the shape of the data structures
which can be analyzed; (ii) simplifying the programming language to be
dealt with; or (iii) reducing scalability.

\subsubsection*{Reachability and Cyclicity analysis}

The oldest notion of reachability dates back to \cite{Nelson83}: his
\emph{reachability predicate} is supposed to tell if a heap location
reaches another one in a linear list.  A reachability-based acyclicity
analysis for C programs was developed by
\cite{DBLP:conf/popl/GhiyaH96}.  That analysis was presented as a
\emph{data-flow} analysis, and the terms ``direction'' and
``interference'', were used for, respectively, reachability and
sharing.  Analyses which compute basically the same information were
presented in more recent work.  \cite{GenaimZ10,GenaimZ13} describe a
formalization of the analysis proposed by
\cite{DBLP:conf/popl/GhiyaH96} in the framework of \emph{Abstract
  Interpretation}, based on a Java-like Object-Oriented language and
provided with soundness proofs.  The same analysis has been also
formalized by means of Abstract Interpretation by \cite{NikolicS14},
which efficiently implement it in the Julia analyzer for Java
(bytecode) and Android\footnote{\texttt{http://www.juliasoft.com}}.
As already discussed in the introduction, the analysis proposed by
\cite{RossignoliS06-short} is less precise since it does not consider
reachability in order to detect cycles.  The present work also builds
upon the results presented in
\cite{ScapinSpotoMScThesis,DBLP:conf/cav/BrockschmidtMOG12}.  The
relation with such works was explained in the introduction, and will
be made even more clear in the rest of the paper, especially in
Section \ref{sec:domainComparison}.

\subsubsection*{Termination and Resource-usage Analysis}

The main goal of most approaches to reachability and cyclicity
analysis is to help \emph{Termination analysis} proving the
termination of loops traversing data structures in the heap.  This is
the case of practically all the papers discussed in the previous
paragraph.  In particular, \cite{ScapinSpotoMScThesis} and
\cite{DBLP:conf/cav/BrockschmidtMOG12} are able to prove termination
even when some kinds of cyclic data structures are traversed.  Proving
termination of a given loop is typically done by finding a
\emph{ranking function} that decreases in every iteration.  For loops
traversing acyclic data structures, the bound is interpreted in terms
of the \emph{depth} of the data structure~\cite{SpotoMP2010} (e.g, the
length of a list, the depth of a tree, etc.).  On the other hand, for
cyclic data structures which are traversed in an acyclic way, the
bound can be interpreted in terms of the acyclic depth, i.e., the
maximal length of acyclic paths.

The abstract domains defined by \cite{ScapinSpotoMScThesis} can assess
that a data structure, although possibly cyclic, might only contain
cycles with certain characteristics; namely, that the fields traversed
by the cycle do not belong to a given set.  This way, it is possible
to prove that traversing a cyclic data structure will terminate,
provided the traversal only concerns fields which are guaranteed not
to appear in cycles.  Importantly, this abstract domain is not able to
deal with the examples of the cyclic tree and the double-linked list
(Section \ref{sec:anExample1} and \ref{sec:anExample2}), since the
field traversed by the loop are involved in cycles.  As a matter of
fact, the abstract domain $\domfrc$ used by \FIELDBASED and presented
in Section \ref{sec:abstractDomains} is strictly more expressive than
the one used by \cite{ScapinSpotoMScThesis}, as proved in Section
\ref{sec:negation}.

\cite{DBLP:conf/cav/BrockschmidtMOG12} address a similar problem from
a similar point of view: termination can be proved in cases where it
is guaranteed that any cycle \emph{must} traverse some set of fields.
Their work can prove the termination of a loop traversing a
double-linked list by building a \emph{Termination Graph} and proving
properties which entail program termination.  Such an analysis has
been implemented in the AProVE tool \cite{DBLP:conf/cade/GieslST06}.
It is important to point out that the property they manage is strictly
less expressive than the one represented by $\domfrc$, as discussed in
Section \ref{sec:requirement}, and that the example of the cyclic tree
cannot be dealt with.

Being closely related to Termination analysis, \emph{Resource-usage
  analysis} \cite{Wegbreit75,caslog-short,AlbertAGPZ12,AlbertGM13}
also benefits from precise cyclicity results, since the mechanisms
which are used to compute upper or lower bounds on the resource
consumption of a program are similar to those used to prove its
termination.  In fact, ranking functions can be used to provide bounds
on the number of loop iterations.

\subsection{Example: tree with edges to parent nodes}
\label{sec:anExample1}

This code fragment works on a class \lstinline!Tree! with fields
\lstinline!left!, \lstinline!right! (pointing to the left and right
sub-tree, respectively), and \lstinline!parent! (the link to the
parent node).  The procedure \lstinline!join! takes two trees
\lstinline!l! and \lstinline!r!, and builds a new tree whose left and
right branch are \lstinline!l! and \lstinline!r!, respectively.

\begin{minipage}{45mm} 
  \begin{lstlisting}
  class Tree {
    Tree left;
    Tree right;
    Tree parent;
  }
  \end{lstlisting}
\end{minipage}
\begin{minipage}{75mm}
  \begin{lstlisting}[firstnumber=6]
  Tree join(Tree l, Tree r) {
    Tree t;  t := new Tree;
    t.left := l;
    t.right := r;
    if (l!=null) then l.parent := t;
    if (r!=null) then r.parent := t;
    return t;
  }
  \end{lstlisting}
\end{minipage}

\noindent
In general, the \lstinline!parent! link makes this data structure
cyclic.  However, a loop traversing a tree either root-down (e.g.,
\lstinline^while (x!=null) x:=x.left^) or leaf-up (e.g.,
\lstinline^while (x!=null) x:=x.parent^) will certainly terminate.  In
order to gather the necessary information to prove termination of such
loops, it is not enough to study which fields are never involved in
cycles, as \cite{ScapinSpotoMScThesis} does, since all fields of
\lstinline!Tree! can be involved in some cycle.  Moreover, to know
that some set of fields must be traversed by all cycles, as done by
\cite{DBLP:conf/cav/BrockschmidtMOG12}, is also not enough, since the
only field which must be obligatorily traversed is \lstinline!parent!,
so that, for example, a leaf-up traversal would be imprecisely taken
as potentially non-terminating since it actually traverses all
mandatory fields.  However, termination could be proved by detecting
that every cycle must involve \lstinline!parent! together with at
least one between \lstinline!left! and \lstinline!right!.  This
example is further discussed in Section \ref{sec:backExample1}.

\subsection{Example: double-linked list}
\label{sec:anExample2}

Consider this code fragment, working on a class \lstinline!Node! with
fields \lstinline!n! and \lstinline!p!, pointing to the next and
previous element of the list, respectively.

\begin{minipage}{6cm}
  \begin{lstlisting}
  i := 1;
  tmp := new Node;
  while (i<10) {
    x := new Node;
    x.n := tmp;
    tmp.p := x;
    tmp := x;
    i := i+1;
  }
  \end{lstlisting}
\end{minipage}
\begin{minipage}{6cm}
  \begin{lstlisting}[firstnumber=10]
  while (x!=null) {
    x := x.n;
  }
  \end{lstlisting}
\end{minipage}

\noindent
The code on the left-hand side creates a double-linked list with ten
elements.  It is clear that such a list is a cyclic data structure;
however, any cycle will certainly involve at least once both
\lstinline!n! and \lstinline!p!.  Consequently, the loop on the
right-hand side will terminate because it only traverses
\lstinline!n!; in other words, it will never entirely traverse a
cycle.  Most standard termination analyzers reject the second loop as
potentially diverging, since acyclicity of the data structure pointed
to by \xx cannot be proved (indeed, it is cyclic).  On the other hand,
the presented approach analyzes the loop in lines 3--9 and infers that
any cyclic path must traverse both \lstinline!n!  and \lstinline!p!,
thus making possible to prove that the loop in lines 10--12
terminates.  A similar piece of information is obtained by
\cite{DBLP:conf/cav/BrockschmidtMOG12}.  This example will be further
discussed in Section \ref{sec:backExample2}.


\section{A simple object-oriented language}
\label{sec:language}

This section defines the syntax and the denotational semantics of a
simplified version of Java.  Class, method, field, and variable names
are taken from a set $\identifiers$ of valid \emph{identifiers}.
A \emph{program} consists of a set of classes
$\classes\subseteq\identifiers$ partially ordered by the
\emph{subclass} relation $\prec$.
Following Java, a \emph{class declaration} takes the form
``\lstinline!class $\class_1$ [extends $\class_2$] { $t_1~\field_1$;$\ldots$ $t_n~\field_n$; $\methoddef_1$ $\ldots$ $\methoddef_k$}!''
where
each ``$t_i~\field_i$'' declares the field $\field_i$ to have type
$t_i\in\classes\cup\{\integer\}$,
and each $\methoddef_j$ is a method definition.  The optional
statement
``\lstinline!extends $\class_2$!'' declares $\class_1$ to be a
subclass of $\class_2$.
A \emph{method definition} takes the form
``\lstinline!$t$ $\method$($t_1~w_1$,$\ldots$,$t_n~w_n$) {$t_{n+1}~w_{n+1}$;$\ldots$$t_{n+p}~w_{n+p}$; $\mathit{com}$}!''
where: $\method\in\identifiers$ is the method name;
$t\in\classes\cup\{\integer\}$ is the type of the return value;
$w_1,\ldots,w_n \in\identifiers$ are the formal parameters;
$w_{n+1},\ldots,w_{n+p}\in\identifiers$ are local variables; $t_{n+k}$
is the \emph{declared type} of $w_{n+k}$, hereafter denoted by
$\dectype(w_{n+k})$; and the command $\mathit{com}$ follows this
grammar:
\[
\begin{array}{rl}
  \mathit{exp}~{:}{:}{=} & n \mid \nil \mid  v \mid v.\field \mid \mathit{exp}_1 \oplus \mathit{exp}_2 \mid \newk{\class} \mid v.\method\mbox{\lstinline!(!}\bar{v}\mbox{\lstinline!)!} \\
  \mathit{com}~{:}{:}{=} & \skipc \mid
  v\assign \mathit{exp} \mid 
  v.\field\assign \mathit{exp} \mid 
  \mathit{com_1}\mbox{\lstinline!;!}\mathit{com_2} \mid\\
  & \ifte~\mathit{exp}~\iftethen~\mathit{com_1}~[\ifteelse~\mathit{com_2}] \mid 
  \while~\mathit{exp}~\whilebody~\mathit{com} \mid
  \return~\mathit{exp}
\end{array} 
\]
where $v,\method,\field\in\identifiers$; $\bar{v} \in \identifiers^*$;
$n\in\mathbb{Z}$; $\class\in\classes$; and $\oplus$ is a binary
operator on $\integer$.
For simplicity, and without loss of generality, \emph{conditions} in
\ifte and \while statements are assumed not to have side effects.
A \emph{method signature} $\class.\method(t_1,\ldots,t_{n}){:}t$
refers to a method $\method$ defined in class $\class$, taking $n$
parameters of type $t_1,\ldots,t_{n} \in \classes\cup\{\integer\}$,
and returning a value of type $t$.
Given a signature $\methodsig$, let $\body{\methodsig}$ be its code
$\mathit{com}$ (i.e., the command appearing in its definition);
$\inp{\methodsig}$ be its set of input variables
$\{\this,w_1,\ldots,w_n\}$, where $\this$ refers to the object
receiving the call; $\locals{\methodsig}$ be its set of local
variables $\{w_{n+1},\ldots,w_{n+m}\}$; and
$\scope{\methodsig}=\inp{\methodsig}\cup\locals{\methodsig}$.  Given a
program, $\fields$ denotes the set of fields declared in
it\footnote{For simplicity, \lstinline!int! fields will be often
  ignored since they have no impact on the heap.}.

A \emph{type environment} $\typenv$ is a partial map from $\variables$
to $\classes\cup\{\integer\}$ which associates types to variables at a
given program point.  Abusing notation, when it is clear from the
context, type environments will be confused with sets of variables
when types are not important; i.e., $v \in \typenv$ will stand for $v
\in \dom(\typenv)$.
A \emph{state} over $\typenv$ is a pair consisting of a frame and a
heap.
A \emph{heap} $\heap$ is a partial mapping from an infinite and
totally ordered set $\locations$ of memory locations to objects;
$\heap(\ell)$ is the object bound to $\ell\in\locations$ in the heap
$\heap$.
An \emph{object} $o \in \objects$ is a pair consisting of a class tag
$\objtag{o}\in\classes$, and a frame $\objframe{o}$ which maps its
fields into $\values=\mathbb{Z}\cup\locations\cup\{\semnil\}$.  For
simplicity, it is assumed that no two fields $\class.\field$ and
$\class'.\field$ with the same field name can be declared in a
program, so that $\field$ will be usually a shorthand for
$\class.\field$;
this is not a significant restriction w.r.t.~Java since the actual
field to which a Java expression $v.\field$ may refer to can be (and
actually is) known statically.
Shorthands are used: $o.\field$ for $\objframe{o}(\field)$;
$\heap[\ell\mapsto o]$ to modify the heap $\heap$ such that a location
$\ell$ contains the object $o$; and $\heap[\ell.\field \mapsto
  \mathit{val}]$ to modify the value of the field $\field$ of the
object $\heap(\ell)$ to $\mathit{val}\in\values$.
A \emph{frame} $\frm$ maps variables in $\dom(\typenv)$ to $\values$.
For $v \in \dom(\typenv)$, $\frm(v)$ refers to the value of $v$, and
$\frm[v\mapsto \mathit{val}]$ is the frame where $v$ has been set to
$\mathit{val}$, or defined to be $\mathit{val}$ if
$v\not\in\dom(\frm)$.
The set of states over $\typenv$ is
\[
\states{\typenv}=\left\{\state{\frm}{\heap} \left|
    \begin{array}{rl}
      1. & \frm\text{ is a frame over } \typenv \text{, }
      \heap\text{ is a heap, and both are well-typed}  \\
      2. & \codom(\frm)\cap\locations\subseteq\dom(\heap)\\
      3. & \forall\ell\in\dom(\heap).~\codom(\objframe{\heap(\ell)})\cap
      \locations\subseteq\dom(\heap) \\
    \end{array}
  \right.\!\!\right\}
\]
Given $\statesym \in \states{\typenv}$, $\statef{\statesym}$ and
$\statem{\statesym}$ refer to its frame and its heap,
respectively. 
The lattice
$\condom=\tuple{\wp(\states{\typenv}),\top,\bot,\cap,\cup}$ defines
the \emph{concrete domain}, where $\top{=}\states{\typenv}$ and
$\bot{=}\emptyset$.

A \emph{denotation} $\den$ over type environments $\typenv_1$ and
$\typenv_2$ is a partial map from $\states{\typenv_1}$ to
$\states{\typenv_2}$: it describes how the state changes when some
code is executed.
The set of denotations from $\typenv_1$ to $\typenv_2$ is
$\denset{\typenv_1}{\typenv_2}$.
An \emph{interpretation} $\interp$ is a special denotation which gives
a meaning to methods in terms of their input and output variables: it
maps methods to denotations, such that $\interp(\methodsig) \in
\denset{\inp{\methodsig}}{\{\out\}}$ for each $\methodsig$.  The
variable $\out$ is a special variable denoting the return value of
methods.  Let $\interpretations$ be the set of all interpretations.


\begin{figure}[t]

\begin{center}
\begin{minipage}{12cm}
\(
\small
\begin{array}{@{}r@{~}l@{}}
  \einter{\typenv}{\interp}{n}(\statesym) =& 
     \state{\statef{\statesym}[\res \seq n]}{\statem{\statesym}} \\
  \einter{\typenv}{\interp}{\nil}(\statesym) =& 
     \state{\statef{\statesym}[\res\seq \semnil]}{\statem{\statesym}} \\
  \einter{\typenv}{\interp}{\newk{\class}}(\statesym) =& 
     \state{\statef{\statesym}[\res \seq \ell]}
           {\statem{\statesym}[\ell \seq \mathit{newobj}(\class)]}
      \mbox{ where } \ell \notin \dom(\statem{\statesym}) \\
  \einter{\typenv}{\interp}{v}(\statesym) =& 
     \state{\statef{\statesym}[\res\seq \statef{\statesym}(v)]}{\statem{\statesym}}\\
  \einter{\typenv}{\interp}{v.\field}(\statesym) =& 
     \state{\statef{\statesym}[\res\seq \statem{\statesym}(\statef{\statesym}(v)).\field]}{\statem{\statesym}}\\
  \einter{\typenv}{\interp}{\mathit{exp}_1{\oplus}\mathit{exp}_2}(\statesym) =&  
     \state{\statef{\statesym}\left[\res\seq \statef{\statesym}_1(\res)  \oplus  \statef{\statesym}_2(\res)\right]}
           {\statem{\statesym}_2}\mbox{ where } \\
  & ~~~~\statesym_1=\einter{\typenv}{\interp}{\mathit{exp}_1}(\statesym) \mbox{ and } \statesym_2=\einter{\typenv}{\interp}{\mathit{exp}_2}\left(\state{\statef{\statesym}}{\statem{\statesym}_1}\right)\\
  \einter{\typenv}{\interp}{v_0.\method(v_1,\ldots,v_n)}(\statesym) =&
   \state{\statef{\statesym}\left[\res \seq \statef{\statesym}_2(out)\right]}{\statem{\statesym}_2}
   \mbox{ where }  \statesym_2=\interp(\methodsig)(\statesym_1) \mbox { s.t. $\statesym_1$ is }\\ 
      &
       \statem{\statesym}_1 = \statem{\statesym}; \statef{\statesym}_1(this)=\statef{\statesym}(v_0);
       \forall 1{\le}i{\le}n.~\statef{\statesym}_1(w_i)=\statef{\statesym}(v_i); \\
      & \mbox{and } \methodsig = \lookup(\statesym,v_0.\method(v_1,\ldots,v_n)); 
  \\
\hline
  \cinter{\typenv}{\interp}{\skipc}(\statesym) =& \statesym
  \\
  \cinter{\typenv}{\interp}{v\assign\mathit{exp}}(\statesym) =&
     \state{\statef{\statesym}\left[v \seq \statef{\statesym}_e(\res)\right]}
           {\statem{\statesym}_e} 
  \\
  \cinter{\typenv}{\interp}{v.\field\assign\mathit{exp}}(\statesym) =&
     \state{\statef{\statesym}}
           {\statem{\statesym}\left[\ell.\field \seq \statef{\statesym}_e(\res)\right]}
           \mbox{ where } \ell=\statef{\statesym}(v)
  \\
  \cinter{\typenv}{\interp}{
    \begin{array}{@{}rl@{}}
      \ifte~\mathit{exp} & \iftethen~\mathit{com}_1 \\
                         & \ifteelse~\mathit{com}_2 
    \end{array}}(\statesym)
  =& 
   \mbox{if $\statef{\statesym}_e(\res) \neq 0$
               then $\cinter{\typenv}{\interp}{\mathit{com}_1}(\statesym)$ 
               else $\cinter{\typenv}{\interp}{\mathit{com}_2}(\statesym)$ }
  \\
  \cinter{\typenv}{\interp}{\while~\mathit{exp}~\whilebody~\mathit{com}}(\statesym)
  =& \den(\statesym) \mbox{ where } \den \mbox{ is the least fixpoint
    of }\\
  &   \mbox{$\lambda w.\lambda\statesym.$ 
    if $\statef{\statesym}_e(\res) \neq 0$ 
    then $w(\cinter{\typenv}{\interp}{\mathit{com}}(\statesym))$ else $\statesym$}
  \\
  \cinter{\typenv}{\interp}{\return~\mathit{exp}}(\statesym) =& 
    \state{\statef{\statesym}\left[\out \seq \statef{\statesym}_e(\res)\right]}
           {\statem{\statesym}_e} 
  \\
  \cinter{\typenv}{\interp}{\mathit{com}_1;\mathit{com}_2}(\statesym) =&
     \cinter{\typenv}{\interp}{\mathit{com}_2}(\cinter{\typenv}{\interp}{\mathit{com}_1}(\statesym))
  \\
\end{array}
\)
\end{minipage}
\end{center}
\caption{Denotations for expressions and commands.  The state
  $\statesym_e$ is $\einter{\typenv}{\interp}{\mathit{exp}}(\statesym)$.}
\label{fig:den-semantics}
\end{figure}

Denotations for expressions and commands are depicted in
Figure~\ref{fig:den-semantics}.
An expression denotation $\einter{\typenv}{\iota}{\mathit{exp}}$ maps
states from $\states{\typenv}$ to states from
$\states{\typenv\cup\{\res\}}$, where $\res$ is a special variable for
storing the value of $\mathit{exp}$.
A command denotation $\cinter{\typenv}{\iota}{com}$ maps states to
states, in presence of $\iota\in\interpretations$.
The function $\newobj{\class}$ creates a new instance of $\class$ with
$\integer$ fields initialized to $0$ and reference fields initialized
to $\semnil$, while $\newloc{\statem{\statesym}}$ returns the first
\emph{free} location, i.e., the first $\ell \notin
\dom(\statem{\statesym})$ according to the total ordering on
locations.  The function $\lookup$ resolves the method call according
to the runtime type of the object, and returns the signature of the
method to be invoked.
The \emph{concrete denotational semantics} of a program is defined as
the \emph{least fixpoint} (\emph{lfp}) of the following transformer of
interpretations~\cite{BossiGabbrielliLeviMartelli94}.

\begin{definition}
  \label{def:conc-den-semantics}
  
  The denotational semantics of a program $P$ is the \emph{lfp} of

  \vspace{-4mm}
  \[
  \tp{P}(\iota) = \left\{
    \methodsig\mapsto\lambda\sigma{\in}\states{\inp{\methodsig}}.
    \project {
      \cinter{\scope{\methodsig}\cup\{\out\}}{\iota}{\body{\methodsig}}
      (\extend{\sigma}{\methodsig}) } {\typenv{\setminus}\out}
  \right\}_{\methodsig \in P}
  \]

  \vspace{-2mm}
  \noindent
  where
  $\extend{\sigma}{\methodsig}=\state{\statef{\statesym}[\forall
    v\in\locals{\methodsig}\cup\{\out\}. v\mapsto
    0/\semnil]}{\statem{\statesym}}$.
\end{definition}

\noindent
The denotation for a method signature $\methodsig\in P$ is computed by
$\tp{P}$ as follows: it
(1) extends (using $\extend{\sigma}{\methodsig}$) the input state
$\statesym \in \states{\inp{\methodsig}}$ such that local variables
are set to $0$ or $\semnil$;
(2) computes the denotation of the code of \methodsig, using
$\cinter{\scope{\methodsig}\cup\{\out\}}{\iota}{\_}$; and
(3) restricts the resulting denotation to $\out$, using
$\exists\typenv{\setminus}\out$.


\section{The abstract domains}
\label{sec:abstractDomains}

This section formalizes the analysis \FIELDBASED by means of Abstract
Interpretation~\cite{Cousot77}, relying on the notion of
\emph{abstract domain}.
The following example shows a class hierarchy which will be used in
the rest of this section.

\begin{example}[(class hierarchy)]
  \label{ex:classHierarchy}
  Let the class hierarchy $\langle \classes, \subclass \rangle$ under
  study be defined as follows.  Objects of class \employee model
  employees, which can be of level 1 (\levelone) or 2 (\leveltwo).  An
  employee has one main device (\md), which is a laptop (\laptop);
  level-2 employees also have a tablet (\tablet) as an accessory
  device (\ad) which is associated (\lnk) to a laptop.  Devices
  (\device) are also associated to their owner (\owner).

  \noindent
  \begin{minipage}{65mm}
    \begin{lstlisting}[numbers=none]
class $\employee$ { $\laptop$ $\md$; }
class $\levelone$ extends $\employee$ { }
class $\leveltwo$ extends $\employee$ { $\tablet$ $\ad$; }
class $\device$ { $\employee$ $\owner$; }
class $\laptop$ extends $\device$ { }
class $\tablet$ extends $\device$ { $\laptop$ $\lnk$; }
    \end{lstlisting}
  \end{minipage}
  \begin{minipage}{55mm}
    \begin{tikzpicture}
      \tikzstyle{classnode} = [circle,draw,inner sep=0.5pt]
      \node[classnode] (employee) at (0.3,0) {$\employee$};
      \node[classnode] (level1) at (-0.2,-1) {$\levelone$};
      \node[classnode] (level2) at (0.8,-1) {$\leveltwo$};
      \draw (employee) -- (level1);
      \draw (employee) -- (level2);
      
      \node[classnode] (device) at (4,0) {$\device$};
      \node[classnode] (laptop) at (3,-1) {$\laptop$};
      \node[classnode] (tablet) at (5,-1) {$\tablet$};
      \draw (device) -- (laptop);
      \draw (device) -- (tablet);
      
      \draw[dotted,->] (employee) -> node[below] {$\md$} (laptop);
      \draw[dotted,->] (level2) .. controls (2,-1.7) and (4,-1.7) .. node[below] {$\ad$} (tablet);
      \draw[dotted,->] (tablet) -> node[above] {$\lnk$} (laptop);
      \draw[dotted,->] (device) -> node[above] {$\owner$} (employee);
    \end{tikzpicture}
  \end{minipage}
  
  \noindent
  In the figure, solid lines correspond to $\subclass$; dotted lines
  represent fields.
\end{example}

\subsection{Background in Logic}
\label{sec:backgroundInLogic}

A \emph{Boolean function} is a function $f : \bool^n \mapsto \bool$
with $n\geq 0$, and can be represented as a \emph{propositional
  formula} over a set $X$ with cardinality $n$.  In this paper,
Boolean functions and propositional formul\ae~will be used
interchangeably.  Moreover, a \emph{truth assignment} of Boolean
variables will be often represented as the set of variables which are
true under that assignment.  In this framework, $X$ will be the set
$\fpropositions = \{ \fproposition{\field}~|~\field \in \fields \}$ of
propositions $\fproposition{\field}$ corresponding to program fields.
Such propositions are called \emph{f-propositions}.  Propositional
formul\ae~over $\fpropositions$ are called \emph{path-formul\ae}.  As
usual, a truth assignment $\finterpretation \subseteq \fpropositions$
is a \emph{model} of a path-formula $\pformula$ if $\pformula$
evaluates to $\true$ under $\finterpretation$.  The set of models of
$\pformula$ is denoted by $\fmodels{\pformula}$.

The path-formula $\onlyfields{\finterpretation}$ is defined as
$\bigwedge \{ ~\fproposition{\field}~|~\fproposition{\field} {\in}
\finterpretation ~\} \wedge \bigwedge \{~ \lnot
\fproposition{\field}~|~ \fproposition{\field} {\in}
\fpropositions{\setminus}\finterpretation~ \}$, and represents the
formula whose only model is $\finterpretation$.  An important special
case is $\pformulaempty = \bigwedge \{ \lnot \fproposition{\field}|
\fproposition{\field} {\in} \fpropositions \}$; moreover,
$\onlyfields{\fproposition{\field_1},..\fproposition{\field_k}}$ will
be a shorthand for $\onlyfields{\{
  \fproposition{\field_1}..\fproposition{\field_k} \}}$.

A formula $\pformula$ is \emph{monotone} if, for every two assignments
$\finterpretation$ and $\finterpretation'$, $\finterpretation
\subseteq \finterpretation'$ and $\finterpretation \in
\fmodels{\pformula}$ imply $\finterpretation' \in
\fmodels{\pformula}$.  It is \emph{positive} if $\fpropositions \in
\fmodels{\pformula}$.  It is \emph{definite} if, for every two
assignments $\finterpretation'$ and $\finterpretation''$,
$\finterpretation' \in \fmodels{\pformula}$ and $\finterpretation''
\in \fmodels{\pformula}$ implies $\finterpretation' \cap
\finterpretation'' \in \fmodels{\pformula}$.  Finally, set conjunction
$\bigwedge X$ and set disjunction $\bigvee X$ will be, respectively,
$\true$ and $\false$ whenever $X = \emptyset$.  $\pformulae$ denotes
the set of all path-formul\ae.

\subsection{Paths, cycles, and fields}
\label{sec:pathCyclesFields}

The abstract domains used by \FIELDBASED are based on the notion of
\emph{reachable heap locations}, i.e., the part of the heap which can
be reached starting from a given location (or the variable pointing to
it).  Given a heap $\heap$, a \emph{path} $\hpath$ from
$\ell'\in\dom(\heap)$ to $\ell''\in\dom(\heap)$ is a sequence
$\tuple{\ell_0,..,\ell_k}$ of locations such that (1) $k \geq 0$; (2)
$\ell_0 = \ell'$; (3) $\ell_k = \ell''$; and (4) for every $0\leq
i\leq k{-}1$, it holds that $\ell_{i+1} \in
\codom(\objframe{\heap(\ell_i)})$, i.e., $\ell_{i+1}$ is the location
bounded to a field of the object to which $\ell_i$ is bound.  The
\emph{length} of a path $\tuple{\ell_0,..,\ell_k}$ is $k$;
\emph{empty} paths are those with length 0.  A \emph{cycle} is a path
from $\ell$ to $\ell$ itself; it is an \emph{empty} cycle if its
length is 0.  Given $\hpath_1 = \tuple{\ell_0,..,\ell_k}$ and
$\hpath_2 = \tuple{\ell_k,..,\ell_m}$, the \emph{concatenation}
$\hpath_1 \hpconcat \hpath_2$ is the path
$\tuple{\ell_1,..,\ell_k,..,\ell_m}$.

\begin{definition}[(reachable locations, similar to \cite{RossignoliS06-short})]
  \label{def:reachability}
  The set of all reachable locations from $\ell\in\dom(\heap)$ is
  $\reachable{\heap}{\ell}{=}\cup\{ \reachablei{i}{\heap}{\ell} \mid
  i\geq 0 \}$,
  where
  $\reachablei{0}{\heap}{\ell} = \{\ell\}$, and
  $\reachablei{i+1}{\heap}{\ell}$ is
  $\cup\{\codom(\objframe{\heap(\ell')}) \cap \locations \mid \ell'
  \in \reachablei{i}{\heap}{\ell}\}$.
\end{definition}

\noindent
The rest of this section is developed in the context of a type
environment $\typenv$, which will be often left implicit.
$\FIELDBASED$ considers \emph{fields} or \emph{field identifiers} when
collecting information about paths; to this end, domains introduced in
Sections \ref{sec:theFieldReachabilityDomain} and
\ref{sec:theFieldCyclicityDomain} are based on the notion of
\emph{field-reachable heap locations}, i.e., the part of the heap
which can be reached from a location by traversing (dereferencing)
certain fields.

\begin{definition}[(field traversal)]
  \label{def:fieldTraversal}
  A path $\hpath$ is said to \emph{traverse} a field $\field \in
  \fields$ in the state $\statesym$ if (1) it is a path in
  $\statem{\statesym}$; (2) $\hpath = \tuple{\ell_0, .., \ell_i, ..,
    \ell_{i+1}, .., \ell_k}$ with $k>i\geq 0$; (3) an object $o$ of
  class $\class'\subclasseq \class$ (i.e., $\objtag{o} = \class'$) is
  stored in $\ell_i$ (i.e., $\statem{\statesym}(\ell_i) = o$); and (4)
  $o.\field$ points to the location $\ell_{i+1}$, i.e.,
  $\objframe{o}(\field) = \ell_{i+1}$.
\end{definition}

\begin{example}[(field traversal)]
  \label{ex:fieldTraversal}
  The path depicted below is compatible with the class hierarchy of
  Example \ref{ex:classHierarchy}, and traverses fields
  $\leveltwo.\ad$, $\tablet.\lnk$ and $\device.\owner$.
  \begin{center}
    \begin{tikzpicture}
      \tikzstyle{objnode} = [rectangle,draw,inner sep=1pt]
      \node[objnode] (l1) at (0,0) {$o_1:\leveltwo$};
      \node[objnode] (l2) at (2.5,0) {$o_2:\tablet$};
      \node[objnode] (l3) at (5,0) {$o_3:\laptop$};
      \node[objnode] (l4) at (7.5,0) {$o_4:\levelone$};
      \draw[->] (l1) -- node[above] {\ad} (l2);
      \draw[->] (l2) -- node[above] {\lnk} (l3);
      \draw[->] (l3) -- node[above] {\owner} (l4);      
    \end{tikzpicture}
  \end{center}
\end{example}

\begin{definition}[(p-satisfaction)]
  \label{def:pathsAndPathFormulae}
  A path $\hpath$ is said to \emph{p-satisfy} an f-proposition
  $\fproposition{\field}$ iff it traverses $\field$.  Given a
  path-formula $\pformula$, the p-satisfaction of $\pformula$ by
  $\hpath$, written $\traverses{\hpath}{\pformula}$, follows the usual
  logical rules:
  \[ \begin{array}{r@{~~}c@{~~}l@{\qquad}r@{~~}c@{~~}l}
    \traverses{\hpath}{\fproposition{\field}} & \mbox{iff} &
    \hpath\mbox{ traverses }\field &
    \traverses{\hpath}{\lnot \pformula} & \mbox{iff} &
    \traverses{\hpath}{\pformula} ~\mbox{does not hold}\\
    \traverses{\hpath}{\pformula' \wedge \pformula''} & \mbox{iff} &
    \traverses{\hpath}{\pformula'} \mbox{ and
    }\traverses{\hpath}{\pformula''} &
    \traverses{\hpath}{\pformula' \vee \pformula''} & \mbox{iff} &
    \traverses{\hpath}{\pformula'} \mbox{ or
    }\traverses{\hpath}{\pformula''}
  \end{array} \]
\end{definition}

\noindent
As usual, $\true$ stands for a tautology, and $\false$ stands for a
contradiction.  Ordering on path-formul\ae~is logical implication:
$\pformula' \leq \pformula''$ iff $\pformula' \Rightarrow \pformula''$
is valid.  The meaning is straightforward: for every $\hpath$, if
$\pformula' \leq \pformula''$ and $\traverses{\hpath}{\pformula'}$,
then $\traverses{\hpath}{\pformula''}$.

\begin{example}[(p-satisfaction)]
  \label{ex:pathFormulae}
  The path of Example \ref{ex:fieldTraversal} p-satisfies
  $\fproposition{\ad} \wedge \fproposition{\owner}$, and any
  path-formula which is implied by it, such as $\fproposition{\ad}$.
  On the other hand, it does not p-satisfy $\lnot
  \fproposition{\owner} \vee \fproposition{\md}$.
\end{example}

A truth assignment $\finterpretation \subseteq \fpropositions$ is said
to be \emph{viable} if there exists some path $\hpath$ in some state
$\statesym$ which p-satisfies $\onlyfields{\finterpretation}$.  To
rule out non-viable truth assignments allows obtaining a Galois
insertion (i.e., without superfluous elements in the abstract domain)
rather than a Galois connection in the definition of the abstract
domains for reachability and cyclicity.

\begin{lemma}[(viability)]
  \label{lemma:decideViability}
  The viability of a truth assignment $\finterpretation$ is
  decidable\footnote{Proofs are available in Appendix
    \ref{sec:proofs}.}.
\end{lemma}

Note that viability is not related to an assignment that satisfies a
formula, but rather to a property of the assignment itself: that it
can represent a ``real'' path in a ``real'' heap according to class
declarations.

\begin{example}[(viability of truth assignments)]
  \label{ex:decideViability}
  Given the class hierarchy introduced in Example
  \ref{ex:classHierarchy}, the truth assignment $\{ \ad, \lnk, \owner
  \}$ is viable, as shown by the path of Example
  \ref{ex:fieldTraversal}.  On the other hand, $\{ \md, \lnk \}$ is
  not viable.  In fact, a path only traversing $\md$ and $\lnk$ should
  contain at least one \tablet object $o_\tablet$ and one \laptop
  object $o_\laptop$ since \lnk must be traversed.  It must also
  include one \employee (or a subclass) object $o_\employee$ since \md
  must also be traversed.  Suppose $o_\tablet$ is the first object on
  the path: then the second must be $o_\laptop$ since \owner cannot be
  traversed, and there is no way to reach $o_\employee$.  On the other
  hand, if the $o_\laptop$ is the first object, then no other object
  can be reached without traversing \owner.  Finally, if $o_\employee$
  is the first object, then the second one must be $o_\laptop$, and,
  again, no other object can be reached from it.
\end{example}

\begin{definition}[(equivalence)]
  \label{def:equivalenceRelationOnPathFormulae}
  The set $\pformulae$ of path-formul\ae~can be partitioned according
  to the following equivalence relation: $\pformula$ and $\pformulag$
  are equivalent unless there is a path in some state which
  p-satisfies one and only one of them.  Note that this relation is
  ``coarser'' than (i.e., implied by) standard logical equivalence
  since the discriminating path must be compatible with the class
  hierarchy.
\end{definition}

\begin{lemma}[(equivalence)]
  \label{lemma:decideEquivalence}
  The equivalence of two path-formul\ae~is decidable.
\end{lemma}

In the following, $\pformulaeq$ will be $\pformulae$ with the
equivalence relation of Definition
\ref{def:equivalenceRelationOnPathFormulae}.

\subsection{The Field-Reachability domain}
\label{sec:theFieldReachabilityDomain}

First, the definition of field-reachability between program variables
is given.

\begin{definition}[(field-reachability on variables)]
  \label{def:fieldReachabilityBetweenVariables}
  A variable $v$ is said to \emph{reach} another variable $w$ in
  $\statesym$ if there exists a path from $\statef{\statesym}(v)$ to
  $\statef{\statesym}(w)$.  Moreover, given some $\pformula$, $v$ is
  said to \emph{$\pformula$-reach} $w$ in $\statesym$ if \emph{every}
  path from $\statef{\statesym}(v)$ to $\statef{\statesym}(w)$
  p-satisfies $\pformula$.  This definition implies that any variable
  $v$ $\false$-reaches $w$ if and only if there is no path between
  them.
\end{definition}

\begin{example}[(reachability)]
  \label{ex:fieldReachabilityBetweenVariables}
  Consider the heap depicted below, based on Example
  \ref{ex:classHierarchy}.
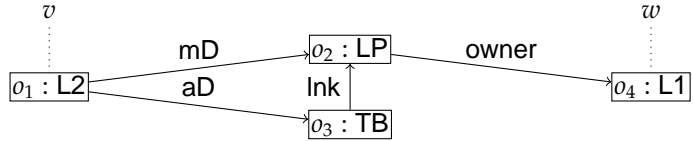
  \begin{center}
    \begin{tikzpicture}
      \tikzstyle{objnode} = [rectangle,draw,inner sep=1pt]
      \node (v) at (0,1) {$v$};
      \node (w) at (8,1) {$w$};

      \node[objnode] (o1) at (0,0) {$o_1:\leveltwo$};
      \node[objnode] (o2) at (4,0.5) {$o_2:\laptop$};
      \node[objnode] (o3) at (4,-0.5) {$o_3:\tablet$};
      \node[objnode] (o4) at (8,0) {$o_4:\levelone$};

      \draw[dotted] (v) -- (o1);
      \draw[dotted] (w) -- (o4);

      \draw[->] (o1) -- node[above] {\md} (o2);
      \draw[->] (o1) -- node[above] {\ad} (o3);
      \draw[->] (o3) -- node[auto] {\lnk} (o2);

      \draw[->] (o2) -- node[above] {\owner} (o4);
    \end{tikzpicture}
  \end{center}
  Among the path-formul\ae~$\pformula$ such that $v$
  $\pformula$-reaches $w$, there are:

  \medskip
  \begin{center}
    \begin{tabular}{lcl}
      $\fproposition{\md} \vee
      \fproposition{\ad}$ & : & all paths traverse either
      $\md$ or $\ad$ \\
      $\fproposition{\owner}$ & : & all paths traverse
      $\owner$ \\
      $\lnot \fproposition{\lnk} \vee \lnot \fproposition{\md}$
      & : & at most one between $\lnk$ and $\md$ is traversed
    \end{tabular}
  \end{center}
\end{example}

An extension of the equivalence relation on path-formul\ae~is needed
here: $\pformulaeq^2$ is a function which takes a pair of variables
$(v,w)$, and returns the set $\pformulae$ equipped by the following
equivalence relation $\equiv^{v,w}$: $\pformula \equiv^{v,w}
\pformulag$ unless there is a path from $\statef{\statesym}(v)$ to
$\statef{\statesym}(w)$ in some state which p-satisfies one and only
one between $\pformula$ and $\pformulag$.  The only difference
w.r.t.~the original $\equiv$ is that the path must connect $v$ to $w$.

The reachability abstract domain is formalized similarly to $\domr$
\cite{GenaimZ13}, and is actually a refinement of it (Section
\ref{sec:domainComparison}).
The next definition shows the lattice of abstract values representing
reachability between variables.  In the following, functions are often
represented by \emph{$\lambda$-notation}, and $\typenv$ is omitted.

\begin{definition}
  \label{def:abs-dom-reach}
  The \emph{field-reachability abstract domain} is the complete
  lattice
  \[ \domfr=\tuple{\reachsetf{\typenv},~ \sqsubseteqfr,~\botfr,~\topfr,~
    \sqcapfr,~\sqcupfr} \]
  \begin{itemize}
  \item the set $\reachsetf{\typenv}$ is the set of functions whose
    domain is $\variables \times \variables$, and that return an
    element of $\pformulaeq^2(v,w) = \pformulae_{\equiv^{v,w}}$ for a
    pair of arguments $(v,w)$;
  \item $\sqsubseteqfr$ is $\leq$ on path-formul\ae, applied
    point-wise;
  \item $\botfr = \lambda(v,w).\false$ and $\topfr =
    \lambda(v,w).\true$;
  \item $f' {\sqcapfr} f'' = \lambda(v,w). f'(v,w) \wedge f''(v,w)$
    and $f' {\sqcupfr} f'' = \lambda(v,w). f'(v,w) \vee f''(v,w)$.
  \end{itemize}
\end{definition}

The meaning of an abstract value $\abselemfr$ is the following: it
represents all the states where, for every $v$ and $w$ (possibly the
same variable), all paths from $v$ to $w$ p-satisfy $\pformula =
\abselemfr(v,w)$.  Note that $\pformula \not\equiv \false$ does not
mean that there is some path from $v$ to $w$ in a concrete state: this
a ``possible'' analysis, so that non-reachability is always a
possibility.  On the other hand, $\abselemfr(v,w) = \false$ excludes
reachability since no path p-satisfies $\false$.  The bottom $\botfr$
models the (non-empty) set of all states where all reference variables
are \nil, whereas $\topfr$ represents $\states{\typenv}$.  Note that
$\abselemfr(v,v) \geq \pformulaempty$ (recall that $\pformulaempty$ is
$\bigwedge \{ \lnot \fproposition{\field}~|~ \fproposition{\field} \in
\fpropositions \}$) whenever $v$ is not \nil.  Unlike $\domr$, this
abstract domain can also represent \emph{aliasing} \cite{Hind01}
because empty paths are also considered.  This explains the different
definition of the bottom element in $\domfr$ and $\domr$.

\begin{example}[(abstract values)]
  \label{ex:abstractState}
  In a program where $\class.\field_1$ and $\class.\field_2$ are the
  only fields, the abstract value $\abselemfr$ such that
  \[ \begin{array}{rcl}
    \abselemfr(v,v) & = & (\lnot \fproposition{\field_1} {\wedge} \lnot
    \fproposition{\field_2}) \vee (\fproposition{\field_1} {\wedge}
    \fproposition{\field_2}) \\
    \abselemfr(v,w) = \abselemfr(v,w') & = & \fproposition{\field_1}
    \\
    \abselemfr(w,v) = \abselemfr(w',v) & = & \false \\
    \abselemfr(w,w) = \abselemfr(w,w') = \abselemfr(w',w)
    = \abselemfr(w',w') & = & \lnot \fproposition{\field_1} {\wedge} \lnot
    \fproposition{\field_2}
  \end{array} \]
  represents heaps (a) and (b), but not (c).  The last line allows $w$
  and $w'$ to alias, and this is compatible with all heaps: in the
  first, $w'$ is $\nil$, so that there are no paths starting from it,
  not even empty paths; in the second, they actually alias; in the
  third, self-aliasing holds for both, but they do not alias with each
  other.  Heap (a) is represented by $\abselemfr(v,v)$ since $v$ is
  not cyclic (only self-aliasing), by $\abselemfr(v,w)$ because the
  path from $v$ to $w$ actually traverses $\field_1$, and by
  $\abselemfr(v,w')$ since there are no paths between $v$ and $w'$.
  Note that $\abselemfr(v,w)$ also allows paths to traverse
  $\field_2$, as in this case.  Heap (b) is also represented by
  $\abselemfr(v,v)$ because $v$ is self-reachable by a path traversing
  both fields, and there is no other path only traversing one of them.
  It is also represented by $\abselemfr(v,w)$ and $\abselemfr(v,w')$
  since $v$ does not need to actually reach $w$ or $w'$.  Heap (c) is
  not represented by $\abselemfr(w,v)$ because $\false$ means that
  there can be no reachability from $w$ to $v$.
 
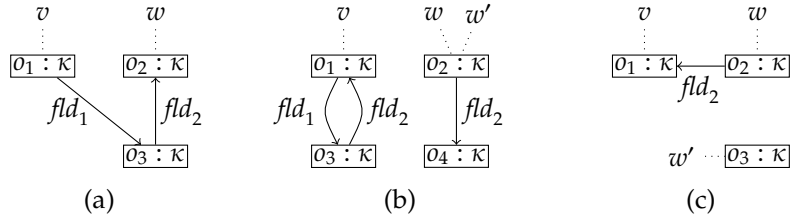
  \begin{center}
    \begin{tikzpicture}
      \tikzstyle{objnode} = [rectangle,draw,inner sep=1pt]

      \node at (0.75,-2.5) {(a)};

      \node (va) at (0,0) {$v$};
      \node (wa) at (1.5,0) {$w$};
      \node[objnode] (l1a) at (0,-0.7) {$o_1:\class$};
      \node[objnode] (l2a) at (1.5,-0.7) {$o_2:\class$};
      \node[objnode] (l3a) at (1.5,-1.9) {$o_3:\class$};
      \draw[dotted] (va) -- (l1a);
      \draw[dotted] (wa) -- (l2a);
      \draw[->] (l1a) -- node[left] {$\field_1$} (l3a);
      \draw[->] (l3a) -- node[right] {$\field_2$} (l2a);
      
      \node at (4.75,-2.5) {(b)};
      
      \node (vb) at (4,0) {$v$};
      \node (wb) at (5.2,0) {$w$};
      \node (w1b) at (5.8,0) {$w'$};
      \node[objnode] (l1b) at (4,-0.7) {$o_1:\class$};
      \node[objnode] (l2b) at (5.5,-0.7) {$o_2:\class$};
      \node[objnode] (l3b) at (4,-1.9) {$o_3:\class$};
      \node[objnode] (l4b) at (5.5,-1.9) {$o_4:\class$};
      \draw[dotted] (vb) -- (l1b);
      \draw[dotted] (wb) -- (l2b);
      \draw[dotted] (w1b) -- (l2b);
      \draw[->] (l1b) .. controls (3.7,-1.3) .. node[left] {$\field_1$} (l3b);
      \draw[->] (l3b) .. controls (4.3,-1.3) .. node[right] {$\field_2$} (l1b);
      \draw[->] (l2b) -- node[right] {$\field_2$} (l4b);
      
      \node at (8.75,-2.5) {(c)};
      
      \node (vc) at (8,0) {$v$};
      \node (wc) at (9.5,0) {$w$};
      \node (w1c) at (8.5,-1.9) {$w'$};
      \node[objnode] (l1c) at (8,-0.7) {$o_1:\class$};
      \node[objnode] (l2c) at (9.5,-0.7) {$o_2:\class$};
      \node[objnode] (l3c) at (9.5,-1.9) {$o_3:\class$};
      \draw[dotted] (vc) -- (l1c);
      \draw[dotted] (wc) -- (l2c);
      \draw[dotted] (w1c) -- (l3c);
      \draw[->] (l2c) -- node[below] {$\field_2$} (l1c);
    \end{tikzpicture}
  \end{center}

  \noindent
  In general, the representation of most path-formul\ae~can be
  simplified by using the $\onlyfields{}$-notation: for example,
  $\abselemfr(v,v)$ can be written as $\pformulaempty \vee
  \onlyfields{\fproposition{\field_1},\fproposition{\field_2}}$.
\end{example}

\begin{definition}[(abstraction and concretization)]
  The \emph{abstraction} and \emph{concretization} functions between
  $\domfr$ and the concrete domain $\condom$ are:
  \[
  \begin{array}{rcl}
    \alphafr(\statesym) & = & \lambda(v,w).~\bigwedge~\{~
    \pformula~|~\mbox{$v$
      $\pformula$-reaches $w$ in}~\statesym~\} \\
    \alphafr(\concelem) & = & \lambda(v,w).~\bigvee~\{~
    \alphafr(\statesym)(v,w)~|~\statesym \in \concelem~\} \\
    \gammafr(\abselemfr) & = & \{~\statesym\in\states{\typenv} ~|~ ~\forall
    v,w\in\typenv.~\exists \pformula \leq \abselemfr(v,w).~\mbox{$v$
      $\pformula$-reaches $w$ in $\statesym$}~\}
  \end{array}
  \]
\end{definition}

$\alphafr$ is computed as follows: for every $\statesym$, the
conjunction of all the $\pformula$ s.t.~$v$ $\pformula$-reaches $w$
comes to be the \emph{strongest condition} p-satisfied by all paths
from $v$ to $w$ in $\statesym$ (recall that $\pformula$-reachability
means that \emph{all} paths p-satisfy $\pformula$).  All strongest
conditions are combined by disjunction on states.  On the other hand,
$\gammafr$ is the adjoint function required by Lemma
\ref{lemma:gi-reach}.  As expected, given $\abselemfr$
s.t.~$\abselemfr(v,w) = \pformulag$, a state where $v$ does not reach
$w$ is still compatible with $\gammafr(\abselemfr)$ (i.e., it belongs
to the concretization unless other variables make it incompatible)
since $\false \leq \pformulag$, and $v$ $\false$-reaches $w$.

\begin{lemma}[(insertion)]
  \label{lemma:gi-reach}
  $\alphafr$ and $\gammafr$ define a \emph{Galois insertion} between
  $\domfr$ and $\condom$.
\end{lemma}

\subsection{The Field-Cyclicity domain}
\label{sec:theFieldCyclicityDomain}

The abstract domain $\domfc$ for cyclicity is similar to $\domfr$, so
that most technical details will not be repeated.  The following
definition is similar to Definition
\ref{def:fieldReachabilityBetweenVariables}.

\begin{definition}[(field-cyclic variables)]
  \label{def:fieldCyclicity}
  A variable $v$ is said to be \emph{cyclic} in a state $\statesym$ if
  there exists a path from $\statef{\statesym}(v)$ containing a cycle.
  Given a path-formula $\pformula$, $v$ is said to be
  $\pformula$-\emph{cyclic} in $\statesym$ if \emph{all} cycles
  reachable from $\statef{\statesym}(v)$ in $\statesym$ p-satisfy
  $\pformula$.
\end{definition}

Note that the p-satisfaction of $\pformula$ is not required for the
whole path starting at $\statef{\statesym}(v)$; rather, it must hold
when only the cyclic part of the path is considered.

A new extension of the equivalence relation on path-formul\ae~is
needed in order to deal with cyclicity: $\pformulaeq^1$ will be a
function which takes a variable $v$, and returns the set $\pformulae$
equipped by the equivalence relation $\equiv^v$: $\pformula \equiv^v
\pformulag$ unless there is a path in some state which (1) starts at
$\statef{\statesym}(v)$; and (2) contains a cycle which p-satisfies
one and only one between $\pformula$ and $\pformulag$.

\begin{definition}[(cyclicity abstract domain)]
  \label{def:abs-dom-cyc}
  The abstract domain for field-cyclicity is similar to the
  field-reachability domain: it is the complete lattice
  \[ \domfc=\tuple{\cycsetf{\typenv},~ \sqsubseteqfc,~\botfc,~\topfc,~
    \sqcapfc,~ \sqcupfc} \] where $\cycsetf{\typenv}$ is the set of
  functions mapping each reference variable $v$ to an element of
  $\pformulaeq^1(v)$; $\sqsubseteqfc$ is $\leq$, applied point-wise;
  $\botfc = \lambda v.\false$, and $\topfc = \lambda v. \true$; and
  $\sqcapfc$ and $\sqcupfc$ are, respectively, $\wedge$ and $\vee$,
  applied point-wise.
\end{definition}

\begin{definition}[(abstraction and concretization)]
  The functions 
  \[
  \begin{array}{rcl}
    \alphafc(\statesym) & = & \lambda v.~\bigwedge~\{~
    \pformula~|~\mbox{$v$ is $\pformula$-cyclic in}~\statesym~\} \\
    \alphafc(\concelem) & = & \lambda
    v.~\bigvee~\{~\alphafc(\statesym)(v)~|~\statesym \in \concelem~\}
    \\
    \gammafc(\abselemfc) & = & \{~\statesym\in\states{\typenv} ~|~
    ~\forall v\in\typenv.~ \exists \pformula \leq
    \abselemfc(v).~v~\mbox{is $\pformula$-cyclic in}~\statesym~\}
  \end{array}
  \]
  are the \emph{abstraction} and \emph{concretization} functions
  between $\domfc$ and $\condom$.
\end{definition}

\begin{lemma}[(insertion)]
  \label{lemma:gi-cyc}
  $\alphafc$ and $\gammafc$ define a \emph{Galois insertion} between
  $\domfc$ and $\condom$.
\end{lemma}

An abstract value such that $\abselemfc(v) = \pformula$ represents
states where all cyclic sub-paths of paths starting at
$\statef{\statesym}(v)$, if any, have to p-satisfy $\pformula$.
Similarly to reachability, the non-nullity of $v$ implies that
$\abselemfc(v) \geq \pformulaempty$ since there always exists an empty
path from $\statef{\statesym}(v)$ to $\statef{\statesym}(v)$.

\begin{example}
  Let $\abselemfc$ be an abstract state, and $\abselemfc(v) =
  \pformulaempty \vee
  \onlyfields{\fproposition{\field_1},\fproposition{\field_2}}$ be the
  path-formula whose only models are $\emptyset$ and
  $\{\fproposition{\field_1},\fproposition{\field_2}\}$.  Consider the
  three heaps below.

  \begin{minipage}{3cm}
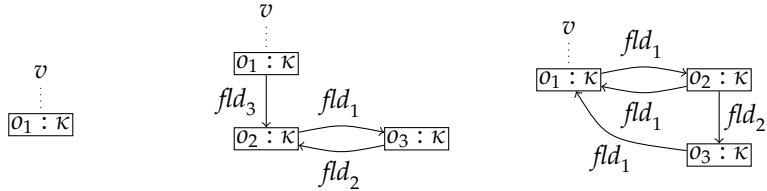
    \begin{center}
      \begin{tikzpicture}
        \tikzstyle{objnode} = [rectangle,draw,inner sep=1pt]
        
        \node (v) at (0,0) {$v$};
        \node[objnode] (o1) at (0,-0.7) {$o_1:\class$};
        \draw[dotted] (v) -- (o1);
      \end{tikzpicture}
    \end{center}
  \end{minipage}
  \begin{minipage}{45mm}
    \begin{center}
      \begin{tikzpicture}
        \tikzstyle{objnode} = [rectangle,draw,inner sep=1pt]
        
        \node (v) at (0,0) {$v$};
        \node[objnode] (o1) at (0,-0.7) {$o_1:\class$};
        \draw[dotted] (v) -- (o1);
        \node[objnode] (o2) at (0,-1.7) {$o_2:\class$};
        \node[objnode] (o3) at (2,-1.7) {$o_3:\class$};
        \draw[->] (o1) -- node[left] {$\field_3$} (o2);
        \draw[->] (o2) .. controls (1,-1.5) .. node[above]
        {$\field_1$} (o3);
        \draw[->] (o3) .. controls (1,-1.9) .. node[below]
        {$\field_2$} (o2);
      \end{tikzpicture}
    \end{center}
  \end{minipage}
  \begin{minipage}{4cm}
    \begin{center}
      \begin{tikzpicture}
        \tikzstyle{objnode} = [rectangle,draw,inner sep=1pt]
        
        \node (v) at (0,0) {$v$};
        \node[objnode] (o1) at (0,-0.7) {$o_1:\class$};
        \node[objnode] (o2) at (2,-0.7) {$o_2:\class$};
        \node[objnode] (o3) at (2,-1.7) {$o_3:\class$};
        \draw[dotted] (v) -- (o1);
        \draw[->] (o1) .. controls (1,-0.5) .. node[above]
        {$\field_1$} (o2);
        \draw[->] (o2) .. controls (1,-0.9) .. node[below]
        {$\field_1$} (o1);
        \draw[->] (o2) -- node[right] {$\field_2$} (o3);
        \draw[->] (o3) .. controls (0.5,-1.5) .. node[below]
        {$\field_1$} (o1);
      \end{tikzpicture}
    \end{center}
  \end{minipage}

  \vspace{3mm}

  \noindent
  The heap depicted on the left is correctly represented by this
  abstract value because the empty cycle from $\statef{\statesym}(v)$
  to $\statef{\statesym}(v)$ p-satisfies $\emptyset$ which is a model
  of $\abselemfc(v)$.  The second heap is also represented because the
  only non-trivial cycle starts from $o_2$ and traverses both
  $\field_1$ and $\field_2$; note that $\abselemfc(v)$ does not need
  to account for $\field_3$ since this field is not traversed by the
  cycle.  On the other hand, the heap on the right is not correctly
  represented because there is a cycle only traversing $\field_1$, and
  $\{\fproposition{\field_1}\}$ is not a model of $\abselemfc(v)$.
\end{example}

\subsection{The reduced product}
\label{sec:theReducedProduct}

The (direct) product of the abstract domains presented in this section
is the set of pairs $\abselemfrc = (\abselemfr,\abselemfc)$, and the
theory of Abstract Interpretation guarantees that it identifies a
Galois connection with $\gamma(\abselemfr,\abselemfc) =
\gammafr(\abselemfr) \cap \gammafc(\abselemfc)$.  In the following,
$\abselemfrc(v_1,v_2)$ will be a shorthand for $\abselemfr(v_1,v_2)$,
where $\abselemfr$ is the reachability part of $\abselemfrc$, and
$\abselemfrc(v)$ will stand for $\abselemfc(v)$, where $\abselemfc$ is
the cyclicity part of $\abselemfrc$.

Usually, the \emph{reduced product} \cite{Cousot79} is more
interesting than the direct product since it happens to generate a
Galois insertion.  It is obtained by ``unifying'' (by means of an
equivalence relation) abstract values with the same concretization
(i.e., representing the same set of concrete states).  Two different
abstract values are mapped to the same set of states when
discrepancies between them do not ``include'' or ``exclude'' any
state.  This happens when $\abselemfr$ contains information which is
not compatible with $\abselemfc$, similarly to the abstract domains
used in the reachability-based analysis $\REACHBASED$ described in
Section \ref{sec:introduction} \cite[Lemma 4.7]{GenaimZ13}.

\begin{example}[(reachability vs.~cyclicity)]
  \label{ex:reducedProduct}
  Let $\abselemfrc' = (\abselemfr',\abselemfc)$ and $\abselemfrc'' =
  (\abselemfr'',\abselemfc)$ only differ in the self-reachability part
  about \xx; i.e., the cyclicity part is the same, and
  $\abselemfc(\xx) = \abselemfr'(\xx,\xx) = \pformula$, but
  $\abselemfr''(\xx,\xx) = \pformulag > \pformula$.  In this case,
  there is, in general, a set $X$ of states which are represented by
  $\abselemfr''$ but not by $\abselemfr'$.  In such states, this
  happens because of paths from \xx to \xx which p-satisfy
  $\pformulag$ but \emph{not} $\pformula$.  However, states in $X$ are
  incompatible with $\abselemfc$ since a path from \xx to \xx is a
  cycle, but $\abselemfc(\xx) = \pformula$ would not be p-satisfied by
  such a path.  Therefore, the difference $X$ between
  $\gammafr(\abselemfr'(\xx,\xx))$ and
  $\gammafr(\abselemfr''(\xx,\xx))$ is a set of concrete states which
  are \emph{not} represented by $\gammafc(\abselemfc(\xx))$, so that
  $\gammafr(\abselemfr') \cap \gammafc(\abselemfc) =
  \gammafr(\abselemfr'') \cap \gammafc(\abselemfc)$.  As a conclusion,
  both $\abselemfrc'$ and $\abselemfrc''$ actually represent the same
  states.
\end{example}

\begin{definition}[(normal form)]
  \label{def:reducedProductNormalForm}
  An abstract value $(\abselemfr,\abselemfc)$ is \emph{in normal form}
  if, for every $v \in \typenv$, $\abselemfc(v) \geq \abselemfr(v,v)$.
  The \emph{normalization} $\normalize{\abselemfr,\abselemfc}$ is
  defined as
  \[ \begin{array}{rcl@{\qquad\qquad}rcl@{\qquad}l}
    \left(\normalize{\abselemfr,\abselemfc}\right)(v) & = &
    \abselemfc(v) \vee \abselemfr(v,v) & \left(\normalize{\abselemfr,\abselemfc}\right)(v,w) & = &
    \abselemfr(v,w)
  \end{array} \]
\end{definition}

The \emph{reduced product} of the reachability and cyclicity domains
is the set of normal-form pairs $(\abselemfr,\abselemfc)$, as proved
in the following lemma.

\begin{lemma}[(reduced product)]
  \label{lemma:reducedProduct}
  The lattice based on

  \centerline{
    $\{~(\abselemfr,\abselemfc)~|~\abselemfr{\in}\reachsetf{\typenv},
    \abselemfc{\in}\cycsetf{\typenv}, (\abselemfr,\abselemfc)~\mbox{is
      in normal form}~\}$}

  \noindent
  with $\gammafrc(\abselemfr,\abselemfc) = \gammafr(\abselemfr) \cap
  \gammafc(\abselemfc)$ is the reduced product between $\domfr$ and
  $\domfc$.
\end{lemma}

In the following, operators on abstract values will be extended to the
reduced product.  For example, since their domains are disjoint,
$\abselemfc \sqcupf \abselemfr$ will be the function $f$ such that (1)
$f(v) = \abselemfc(v)$; and (2) $f(v,w) = \abselemfr(v,w)$.  Moreover,
reachability and cyclicity abstract values can be mixed: e.g.,
$\abselemfrc \sqcupf \abselemfr$ will be the function $g$ such that
(1) $g(v) = \abselemfrc(v)$ (i.e., the cyclicity component); and (2)
$g(v,w) = \abselemfrc(v,w) \vee \abselemfr(v,w)$.

\subsection{Comparison with other approaches to the problem}
\label{sec:domainComparison}

This section refers to reachability; its extension to cyclicity is
straightforward.  The domain $\domfr$ presented in Section
\ref{sec:theFieldReachabilityDomain} is very expressive since it can
predicate a number of interesting facts about paths.  This section
compares $\domfrc$ with a number of abstract domains which are meant
to tackle the same problem.

\subsubsection{An abstract domain without field information}
\label{sec:withoutFields}

Such an abstract domain was inspired by a static analysis for C
programs \cite{DBLP:conf/popl/GhiyaH96}, and formalized as an abstract
domain by \cite{GenaimZ13}.  It is structurally similar to $\domfr$,
but field information is not considered.

\begin{definition}[(without fields \cite{GenaimZ13})]
  \label{def:abs-dom-withoutFields}
  This abstract domain is the complete lattice
  $\domr=\tuple{\wp(X^\rightsquigarrow),\subseteq,
    \emptyset,X,\cap,\cup}$, where
  \[
  X^\rightsquigarrow = \left\{\REACHES{v}{w} ~\left| \begin{array}{l} v, w \in
        \dom(\typenv)\mbox{, and there exist
          $\class_1{\subclasseq}\typenv(v)$ and
          $\class_2{\subclasseq}\typenv(w)$} \\ \mbox{such that
          $\class_2$ is reachable from $\class_1$} \end{array}
    \right. \right\}
  \]
  where the notion of reachability between classes is taken from
  \cite{spoto:pair_sharing}: ``$\class_2$ is reachable from
  $\class_1$'' means that it is possible to have a heap where an
  object of class $\class_1$ reaches an object of class $\class_2$.
\end{definition}

An abstract value $\abselemr$ is a set of statements
$\REACHES{\_}{\_}$: if $\REACHES{v}{w} \notin \abselemr$, then the
concretization of $\abselemr$ will not include any state where $v$
reaches $w$.  $\domr$ is an abstraction of $\domfr$.

\begin{lemma}
  \label{lemma:gi-withoutFields}
  The abstract domain $\domr$ is an abstraction of $\domfr$.
\end{lemma}

As already mentioned, $\domfr$ is also able to represent
\emph{aliasing} \cite{Hind01}.  Indeed, it is also a refinement of the
standard abstract domain for aliasing analysis.

\begin{lemma}
  \label{lemma:gi-aliasing}
  The abstract domain $\domfr$ is a refinement of the aliasing domain.
\end{lemma}

A sound abstract semantics based on $\domr$ has been proposed by
\cite{GenaimZ13}; very similar analyses can be found in the works by
\cite{DBLP:conf/popl/GhiyaH96}, and by \cite{NikolicS14}.  In terms
of precision, $\domfr$ is more precise than $\domr$ since the field
information can rule out states where paths do not p-satisfy a given
formula.  As discussed in the introduction, this is more evident when
dealing with cyclicity, since such an extra information about cycles
can lead to prove the termination of algorithms which traverse cyclic
data structures.  In terms of efficiency, it is clear that an abstract
semantics based on $\domfr$ and $\domfc$ instead of $\domr$ and
$\domc$ is more expensive since (1) operators on path-formul\ae~are
more complex (e.g., $\odot$ in Section \ref{sec:abstractSemantics});
and (2) the convergence of the global fixpoint (Section
\ref{sec:fixpoint}) can be slower; in fact, for every pair $(v,w)$,
$\domfr$ allows ascending chains of path-formul\ae~whose length is
exponential on the number of fields, while $\domr$ only allows 2-long
chains (``does not reach'' $<$ ``reaches'').

\subsubsection{An abstract domain based solely on class reachability}
\label{sec:onlyClasses}

Another abstract domain which can be studied is the one where just the
class hierarchy is considered: a variable $v$ is regarded as
potentially reaching $w$ whenever the class of $w$ is reachable from
the class of $v$ \cite{spoto:pair_sharing}.  Such a domain will be
denoted by $\domkr$; needless to say, it is an abstraction of both
$\domr$ and (by transitivity) $\domfr$.

\begin{definition}[(class-based)]
  \label{def:abs-dom-onlyClasses}
  The domain $\domkr$ is defined as the lattice
  \[ \domkr = \tuple{ \wp(X^{\classes \times\classes})_\equiv,
    \subseteq, \emptyset, X^{\classes \times\classes}, \cap, \cup } \]
  where (1) $X^{\classes \times\classes}$ is the set of all pairs
  $(\class_1, \class_2)$ s.t.~$\class_2$ is reachable from $\class_1$;
  and (2) the equivalence relation is such that $S_1 \equiv S_2$ are
  equivalent if they have the same downward closure
  w.r.t.~$\subclasseq$.  Formally: let $S^\subclasseq = \{
  (\class'_1,\class'_2)~|~\exists (\class_1,\class_2) \in S, \class'_1
  \subclasseq \class_1, \class'_2 \subclasseq \class_2 \}$, i.e.,
  pairs obtained by adding all subclasses of classes belonging to a
  pair.  Then $S_1 \equiv S_2$ iff $S_1^\subclasseq =
  S_2^\subclasseq$.
\end{definition}

An abstract value $\abselemkr$ contains pairs of classes, and the
intended meaning is that it represents all the states where a path
goes from a $\class_1$ object to a $\class_2$ object only if
$(\class_1,\class_2) \in \abselemkr$.  As discussed below, this is a
very rough approximation of the concrete semantics.  Due to how the
equivalence relation is defined, abstract values can be considered as
being closed on subclasses, like $S^\subclasseq$.

\begin{lemma}
  \label{lemma:gc-onlyClasses}
  $\domkr$ is an abstraction of $\domr$.
\end{lemma}

It is straightforward to see that an abstract semantics based on
$\domkr$ and the corresponding $\domkc$ would be much less precise
than any other approaches to reachability/cyclicity analysis discussed
in this paper: a variable of type $\class$ is condemned to be
potentially cyclic as long as there is some \emph{possibility} to
create a cycle starting from a $\class$ object.  On the other hand,
the analysis is fully computable: it is only necessary to take the
class hierarchy into account.  This also implies that the
reachability/cyclicity information does \emph{not} depend on the
program point: it can be computed once and used whenever needed.  The
interest of $\domkc$ is mainly theoretical because its lack of
precision makes it impractical as the core of a static analyzer.
However, it could be used as a first approximation which rules out
some paths or cycles without the need of running more precise, but
more expensive analyses like the other ones discussed here.

\subsubsection{Abstract domains with restrictions on path-formul\ae}
\label{sec:restrictedPFormulae}

The abstract domains introduced in this section are very similar to
$\domfr$, the only difference being the restriction of
path-formul\ae~to some specific class of propositional formul\ae.
Domains $\dompr$, $\dommr$, and $\domdr$ restrict path-formul\ae~to,
respectively, \emph{positive}, \emph{monotone}, and \emph{definite}
Boolean functions (Section \ref{sec:backgroundInLogic}).

The domain $\dompr$ deals with positive Boolean functions with the
addition of the bottom element $\false$.  The class of
path-formul\ae~that can be represented includes monotone functions
(note that the addition of $\false$ is needed to have this property),
so that $\dompr$ can be easily proved to be a refinement of $\dommr$.

The restriction to monotone Boolean functions makes sense because a
monotone function (with the exception of $\false$, which is $p \wedge
\lnot p$ for some $p$, and $\true$) can be represented by a
\emph{conjunctive normal form} where all literals are positive.  In
terms of paths and fields, a monotone formula can say that paths have
to traverse a field, but not that they do \emph{not} have to.
Monotonicity implies that if a path $\hpath$ p-satisfies a monotone
path-formula $\pformula$, then any path which contains $\hpath$ as a
part of it will also p-satisfy $\pformula$.

Finally, the use of $\domdr$ can be motivated by the fact that, given
a definite formula $\pformula$, and two paths p-satisfying it and
sharing a common part in the heap, their common part is guaranteed to
p-satisfy $\pformula$.  For example, let $\fields$ be $\{ \ffield,
\gfield, \hfield \}$; in this case, the formula
$\fproposition{\gfield}$ is definite.  Consider the heap depicted
below: both $\hpath_1$ and $\hpath_2$ p-satisfy
$\fproposition{\gfield}$, and their intersection $\hpath$ is also
guaranteed to p-satisfy it.
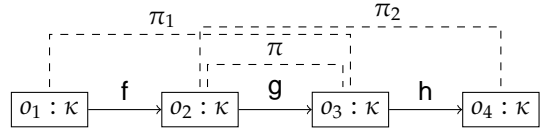
\begin{center}
  \begin{tikzpicture}
    \node(p1) at (1.5,1.2) {$\hpath_1$};
    \node(p) at (3,0.8) {$\hpath$};
    \node(p2) at (4.5,1.3) {$\hpath_2$};
    \node[rectangle,draw] (l1) at (0,0) {$o_1:\class$};
    \node[rectangle,draw] (l2) at (2,0) {$o_2:\class$};
    \node[rectangle,draw] (l3) at (4,0) {$o_3:\class$};
    \node[rectangle,draw] (l4) at (6,0) {$o_4:\class$};
    \draw[->] (l1) -- node[above] {\ffield} (l2);
    \draw[->] (l2) -- node[above] {\gfield} (l3);
    \draw[->] (l3) -- node[above] {\hfield} (l4);
    \draw[dashed] (0,0.3) -- (0,1.0) -- (4,1.0) -- (4,0.3);
    \draw[dashed] (2.1,0.3) -- (2.1,0.6) -- (3.9,0.6) -- (3.9,0.3);
    \draw[dashed] (2,0.3) -- (2,1.1) -- (6,1.1) -- (6,0.3);
  \end{tikzpicture}
\end{center}

The rest of this section will formally define $\dommr$ and demonstrate
that it is a strict abstraction of $\domfr$; similar results can be
also proved for $\dompr$ and $\domdr$.

\begin{definition}[(monotone reachability)]
  \label{def:abs-dom-monotone}
  The \emph{monotone field-reachability abstract domain} is the
  complete lattice
  \[ \dommr=\tuple{\reachsetmr,~ \sqsubseteqmr,~\botmr,~\topmr,~
    \sqcapmr,~\sqcupmr} \]
  \begin{itemize}
  \item $\reachsetmr$ is the set of functions from
    $\variables\times\variables$ to monotone path-formul\ae, equipped
    with an equivalence relation similar to $\equiv^{v,w}$;
  \item $\sqsubseteqmr$ is $\leq$ on path-formul\ae, applied
    point-wise;
  \item $\botmr = \lambda(v,w).\false$, and $\topmr =
    \lambda(v,w).\true$;
  \item $\sqcapmr$ is $\wedge$ applied point-wise, and $\sqcupmr$ is
    $\vee$ applied point-wise.
  \end{itemize}
\end{definition}

\begin{lemma}
  \label{lemma:gc-monotonic}
  The following abstraction and concretization functions define a
  Galois connection between $\domfr$ and $\dommr$: the latter strictly
  abstracts the former.
  \[ \begin{array}{rcl} (\alpha(\abselemfr))(v,w) & = &
    \left\{ \begin{array}{ll}
        \false & \mbox{if}~\abselemfr(v,w){\models} \false \\
        \true & \mbox{if}~\true{\models}\abselemfr(v,w) \\
        \bigwedge \{~ \fproposition{\field_1} \vee..\vee
        \fproposition{\field_k}~|~ \abselemfr(v,w){\models}
        \fproposition{\field_1} \vee..\vee\fproposition{\field_k}~\} &
        \mbox{otherwise} \end{array} \right.  \\ \gamma(\abselemmr) &
    = & \abselemmr
  \end{array} \]
\end{lemma}

$\dommr$ is strictly more abstract than $\domfr$, as shown by the
following example.

\begin{example}[(monotone reachability)]
  Part (a) of the figure below shows a heap where $v$ can reach $w$ by
  traversing two paths.
  \begin{center}
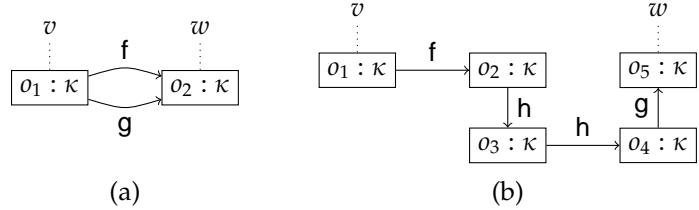
    \begin{minipage}{4cm}
      \begin{center}
        \begin{tikzpicture}
          \node (v) at (0,0.8) {$v$};
          \node (w) at (2,0.8) {$w$};
          \node[rectangle,draw] (l1) at (0,0) {$o_1:\class$};
          \node[rectangle,draw] (l2) at (2,0) {$o_2:\class$};
          \draw[dotted] (v) -- (l1);
          \draw[dotted] (w) -- (l2);
          \draw[->] (l1) .. controls (1,0.3) .. node[above] {\ffield} (l2);
          \draw[->] (l1) .. controls (1,-0.3) .. node[below] {\gfield} (l2);
        \end{tikzpicture}
      \end{center}
    \end{minipage}
    \begin{minipage}{6cm}
      \begin{center}
        \begin{tikzpicture}
          \node (v) at (0,0.8) {$v$};
          \node (w) at (4,0.8) {$w$};
          \node[rectangle,draw] (l1) at (0,0) {$o_1:\class$};
          \node[rectangle,draw] (l2) at (2,0) {$o_2:\class$};
          \node[rectangle,draw] (l3) at (2,-1) {$o_3:\class$};
          \node[rectangle,draw] (l4) at (4,-1) {$o_4:\class$};
          \node[rectangle,draw] (l5) at (4,0) {$o_5:\class$};
          \draw[dotted] (v) -- (l1);
          \draw[dotted] (w) -- (l5);
          \draw[->] (l1) -- node[above] {\ffield} (l2);
          \draw[->] (l2) -- node[auto] {\hfield} (l3);
          \draw[->] (l3) -- node[above] {\hfield} (l4);
          \draw[->] (l4) -- node[auto] {\gfield} (l5);
        \end{tikzpicture}
      \end{center}
    \end{minipage}

    \vspace{2mm}

    \begin{minipage}{4cm}
      \begin{center}
        (a)
      \end{center}
    \end{minipage}
    \begin{minipage}{6cm}
      \begin{center}
        (b)
      \end{center}
    \end{minipage}
  \end{center}
  The abstract value which best represents such a heap in $\domfr$ is
  $\abselemfr$ such that $\abselemfr(v,w) = (\fproposition{\ffield}
  \vee \fproposition{\gfield}) \wedge (\lnot \fproposition{\ffield}
  \vee \lnot \fproposition{\gfield})$ (\emph{exclusive disjunction}).
  On the other hand, the best abstract value from $\dommr$ would be
  such that $\abselemmr(v,w) = \fproposition{\ffield} \vee
  \fproposition{\gfield}$.  It can be easily seen that $\abselemmr$
  also represents heaps like part (b), where a path
  traverses both $\ffield$ and $\gfield$, whereas $\abselemfr$ does
  not.
\end{example}

As mentioned before, $\dommr$ is an abstraction of $\dompr$, while
$\domdr$ can be compared with neither $\dommr$ nor $\dompr$: for
example, (1) $p \vee q$ is monotone but not definite, whereas $\lnot p
\vee q$ is definite but not monotone; and (2) $p \vee q$ is positive
but, again, not definite, whereas $\lnot p \wedge \lnot q$ is definite
but not positive.

\subsubsection{An domain excluding fields from paths}
\label{sec:negation}

The abstract domain introduced by \cite{ScapinSpotoMScThesis}, which
will be denoted by $\domnr$ in this paper, also considers field
information to improve on existing techniques
\cite{DBLP:conf/popl/GhiyaH96,NikolicS14,GenaimZ13}.  The property
tracked by $\domnr$ is ``there are no paths from $v$ to $w$ which
traverse any field belonging to a set $\mathsf{F}$''.  The following
definition is taken from \cite[Def.~5.1]{ScapinSpotoMScThesis}, and
slightly modified in order to adapt notation and only consider
reachability.

\begin{definition}[(Scapin's)]
  \label{def:abs-dom-negation}
  The complete lattice $\domnr$ is $\tuple{\reachsetnr, \sqsubseteqnr,
    \sqcapnr, \sqcupnr}$, where $\reachsetnr$ is $\wp(\variables
  \times \variables \times \wp(\fields))$, $\sqsubseteqnr$ is
  $\subseteq$, $\sqcapnr$ is $\cap$, and $\sqcupnr$ is $\cup$.
\end{definition}

An abstract value containing a triple $(v,w,\mathsf{B})$\footnote{To
  avoid confusion with path-formul\ae, $\mathsf{B}$ is used here
  instead of the original $\mathsf{F}$ to denote field sets.},
originally expressed as $v \not\rightsquigarrow^{\mathsf{B}} w$,
represents states where $v$ can only reach $w$ without traversing any
$\field \in \mathsf{B}$.

\begin{lemma}
  \label{lemma:gc-negation}
  The following functions define a Galois insertion between $\domfr$
  and $\domnr$: the latter is a strict abstraction of the former.
  \begin{eqnarray*}
    \alpha(\abselemfr) & = & \left\{ v
    \not\rightsquigarrow^{\mathsf{B}} w~|~\forall \field \in
    \mathsf{B}.~\abselemfr(v,w) \models \lnot \fproposition{\field}
    \right\} \\ \gamma(\abselemnr) & = & \lambda
    v,w.~\bigwedge_{\field \in \mathsf{B}}^{} \lnot
    \fproposition{\field}\qquad \mbox{where $\mathsf{B}$ is the
      maximal set s.t.}~v \not\rightsquigarrow^{\mathsf{B}} w \in
    \abselemnr
  \end{eqnarray*}
\end{lemma}

The abstract semantics and the complete analysis based on this domain
is described by \cite{ScapinSpotoMScThesis}.  Importantly, it is
\emph{not} able to express the property that every cycle \emph{has} to
traverse certain fields, so that termination of the double-linked-list
or the cyclic-tree example cannot be proved.  However, the convergence
of the global fixpoint is likely to be faster since $\domnr$ only
allows ascending chains of path-formul\ae~whose length is linear on
the number of fields.

\subsubsection{An analysis detecting that all paths have to traverse
  certain fields}
\label{sec:requirement}

The analysis presented by \cite{DBLP:conf/cav/BrockschmidtMOG12} uses
some kind of field-sensitive information in order to prove
termination.  In fact, it is able to detect situations where all
cycles which can occur in a data structure \emph{must} traverse a
certain set of fields, as in the example of Section
\ref{sec:anExample2}.  It is easy to see that such a piece of
information, which is obtained by a component of their work, can be
formalized into an abstract domain which is strictly less refined
than $\dommr$.  Unlike the other domains discussed in this section,
the following definition refers to cyclicity instead of reachability
since cyclicity is represented more explicitly by
\cite{DBLP:conf/cav/BrockschmidtMOG12}.

\begin{definition}
  \label{def:abs-dom-requirement}
  The complete lattice $\domqc$ is $\tuple{\cycsetqc, \sqsubseteqqc,
    \sqcapqc, \sqcupqc}$, where $\cycsetqc$ is the set of partial
  functions from $\variables$ to $\wp(\fields)$.  An abstract value
  $\abselemqc$ represents concrete states where (1) for every $v \in
  \dom(\abselemqc)$ such that $\abselemqc(v) = \mathsf{B}$, $v$ can
  only be cyclic by means of paths which traverse all $\field \in
  \mathsf{B}$; and (2) for every $w \notin \dom(\abselemqc)$, $w$
  cannot be cyclic.  Moreover,
  \begin{itemize}
  \item $\abselemqc^1 \sqsubseteqqc \abselemqc^2$ iff, for every $v
    \in \dom(\abselemqc^1)$, it holds that $v \in \dom(\abselemqc^2)$
    and $\abselemqc^2(v) \subseteq \abselemqc^1(v)$ (i.e.,
    $\abselemqc^1$ allows less variables to be cyclic and, in this
    case, puts stricter conditions on paths);
  \item $\abselemqc = \abselemqc^1 \sqcapqc \abselemqc^2$ is such that
    $\dom(\abselemqc) = \dom(\abselemqc^1) \cap \dom(\abselemqc^2)$,
    and, for every $v \in \dom(\abselemqc^1) \cap \dom(\abselemqc^2)$,
    it holds that $\abselemqc(v) = \abselemqc^1(v) \cup
    \abselemqc^2(v)$;
  \item $\abselemqc = \abselemqc^1 \sqcupqc \abselemqc^2$ is such that
    $\dom(\abselemqc) = \dom(\abselemqc^1) \cup \dom(\abselemqc^2)$,
    and (1) for every $v \in \dom(\abselemqc^1) \cap
    \dom(\abselemqc^2)$, it holds that $\abselemqc(v) =
    \abselemqc^1(v) \cap \abselemqc^2(v)$; (2) for every $v \in
    \dom(\abselemqc^1) \setminus \dom(\abselemqc^2)$, it holds that
    $\abselemqc(v) = \abselemqc^1(v)$; and (3) for every $v \in
    \dom(\abselemqc^2) \setminus \dom(\abselemqc^1)$, it holds that
    $\abselemqc(v) = \abselemqc^2(v)$.
  \end{itemize}
\end{definition}

\begin{lemma}
  \label{lemma:gc-requirement}
  The following functions define a Galois insertion between $\domfc$
  and $\domqc$: the latter is a strict abstraction of the former.
  \begin{eqnarray*}
    \alpha(\abselemfc) & = & \abselemqc~\mbox{with domain}~D =
    \{~v~|~\abselemfc(v) \neq \false \}~\mbox{and such that} \\
    & & \abselemqc(v) = \{~\field~|~\abselemfc(v) \models
    \fproposition{\field}~\} \\
    \gamma(\abselemqc) & = & \lambda
    v.\left\{ \begin{array}{l@{\qquad}l}
        \bigwedge_{\field \in \abselemqc(v)}
        \fproposition{\field} & \mbox{if}~v \in
        \dom(\abselemqc) \\
        \false & \mbox{otherwise}
      \end{array} \right.
  \end{eqnarray*}
\end{lemma}

As a matter of fact, $\domqc$ is also an abstraction of the cyclicity
counterpart of $\dommr$ since monotone boolean functions can capture
the desired property.  Indeed, the path formul\ae~returned by the
function $\gamma$ presented in Lemma \ref{lemma:gc-requirement} (i.e.,
either $\false$ or $\bigwedge_{\field \in \abselemqc(v)}
\fproposition{\field}$) are monotone.

\subsubsection{Even more expressive abstract domains}
\label{sec:moreExpressive}

Most domains discussed so far follow a similar pattern: an abstract
value assigns to a pair of variables $(v,w)$ (or to a single variable,
in the case of cyclicity) a logical formula which is in charge of
describing a property of all paths between $v$ and $w$.  This
observation leads to consider more refined logics capturing
finer-grained properties of paths.

For example, one could be interested in the \emph{order} in which a
path traverses fields.  Such an order could be either a \emph{total}
or a \emph{partial} order, stating that, for example, every path from
$v$ to $w$ only traverses $\field''$ after traversing $\field'$.
Another potentially interesting property is the (minimum or maximum)
number of occurrences of a given field in a path.  In principle, these
properties can be combined to represent even more precise properties
such as \emph{all paths traverse $\field'$ at least once, and
  $\field''$ at least twice; the first occurrence of $\field''$ comes
  before the first of $\field'$; the second occurrence of $\field''$
  comes after the first of $\field'$}.

To define such domains and discuss their applicability is beyond the
scope of this paper.  Anyway, it is likely that this kind of
properties of paths could be represented by \emph{first-order logic}
or some version of \emph{temporal logic} or \emph{separation logic}.

\section{The field-sensitive abstract semantics}
\label{sec:abstractSemantics}

This section defines an abstract semantics $\FIELDBASED$ based on
$\domfrc$.  The semantics has to take into account any modification to
the heap which may occur at runtime.  In particular, paths can be
created and removed by means of field updates.  On the contrary,
updating a reference variable (not one of its fields) does \emph{not}
modify the heap structure, but has to be reflected anyway in the
resulting abstract values.
An \emph{abstract denotation} $\absden$ from $\typenv_1$ to
$\typenv_2$ is a partial map from $\domfrctau{\typenv_1}$ to
$\domfrctau{\typenv_2}$.  It describes how the abstract input state
changes when a piece of code is executed.  The set of all abstract
denotations from $\typenv_1$ to $\typenv_2$ is denoted by
$\absdenset{\typenv_1}{\typenv_2}$.
As in the concrete setting, interpretations provide abstract
denotations for methods in terms of their input and output arguments.
An \emph{interpretation} $\absinterp$ maps method signatures to
abstract denotations, and is such that $\absinterp(\methodsig) \in
\absdenset{\inp{\methodsig}}{\inp{\methodsig}\cup\{\out\}}$ for every
$\methodsig$.
Note that the range of denotations is $\inp{\methodsig}\cup\{\out\}$,
unlike the concrete semantics where only $\out$ is needed since
changes in the memory are directly observable in the heap.
The set of all abstract interpretations is denoted by
$\absinterpretations$.

\subsection{Preliminaries}
\label{sec:preliminaries}

\subsubsection{Auxiliary analyses}
\label{sec:auxiliaryAnalyses}

$\FIELDBASED$ uses \emph{deep-sharing} and \emph{purity}
\cite{GenaimS08} analyses as pre-existent components; i.e., programs
are assumed to have been analyzed w.r.t.\ these properties using
state-of-the-art tools.  Two reference variables $v$ and $w$
\emph{deep-share} in $\statesym$ iff they both reach a common location
by traversing non-empty paths, i.e.,
$\reachableplus{\statem{\statesym}}{\statef{\statesym}(v)}\cap\reachableplus{\statem{\statesym}}{\statef{\statesym}(w)}
\neq\emptyset$, where $\reachableplus{\cdot}{\cdot}$ is like
$\reachable{\cdot}{\cdot}$ but excludes empty paths.  This property,
written as $\SHARE{v}{w}$, is different from standard sharing
\cite{spoto:pair_sharing} since paths from $\statef{\statesym}(v)$ and
$\statef{\statesym}(w)$ to the common location must have length $\geq
1$.  A variable deep-shares with itself if the depth of the data
structure pointed to by it is at least 2; the relation is symmetric.

\begin{example}[(deep-sharing)]
  In the following heap, $\xx$ deep-shares with itself and with $\yy$;
  $\yy$ deep-shares with itself, with $\xx$, and with $\zz$; $\zz$
  only deep shares with itself and with $\yy$; $\mm_1$ and $\mm_2$
  alias but do not deep-share, not even with themselves.

  
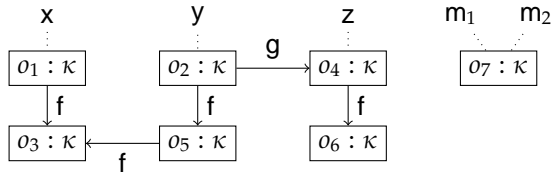
\begin{center}
    \begin{tikzpicture}
      \node (x) at (0,-.3) {$\xx$};
      \node (y) at (2,-.3) {$\yy$};
      \node (z) at (4,-.3) {$\zz$};
      \node (m1) at (5.5,-.3) {$\mm_1$};
      \node (m2) at (6.5,-.3) {$\mm_2$};
      \node[rectangle,draw] (l1) at (0,-1) {$o_1:\class$};
      \node[rectangle,draw] (l2) at (2,-1) {$o_2:\class$};
      \node[rectangle,draw] (l3) at (0,-2) {$o_3:\class$};
      \node[rectangle,draw] (l4) at (4,-1) {$o_4:\class$};
      \node[rectangle,draw] (l5) at (2,-2) {$o_5:\class$};
      \node[rectangle,draw] (l6) at (4,-2) {$o_6:\class$};
      \node[rectangle,draw] (l7) at (6,-1) {$o_7:\class$};
      \draw[dotted] (x) -- (l1);
      \draw[dotted] (y) -- (l2);
      \draw[dotted] (z) -- (l4);
      \draw[dotted] (m1) -- (l7);
      \draw[dotted] (m2) -- (l7);
      \draw[->] (l1) -- node[auto] {\ffield} (l3);
      \draw[->] (l2) -- node[auto] {\ffield} (l5);
      \draw[->] (l2) -- node[auto] {\gfield} (l4);
      \draw[->] (l5) -- node[auto] {\ffield} (l3);
      \draw[->] (l4) -- node[auto] {\ffield} (l6);
    \end{tikzpicture}
  \end{center}
\end{example}

\noindent
Note that two variables may deep-share without being reachable from
each other, and one may reach the other without deep-sharing with it.
This property is not exactly like $\ssearrow\!\sswarrow$ of
\cite{DBLP:conf/cav/BrockschmidtMOG12} since it requires both paths to
have length $\geq 1$, not only one of them.  However, it can be (and
actually is, see Section \ref{sec:practicalIssues}) easily implemented
as a variation of standard sharing analysis.  Importantly, it is a
\emph{possible analysis}, i.e., a deep-sharing statement has to be
added to the abstract description of the heap whenever there is the
possibility of deep-sharing.

The $i$-th argument of a method $\method$ is said to be \emph{pure} if
$\method$ does not update the data structure to which the argument
initially pointed.  The analysis proposed by \cite{GenaimS08}, based
on previous work by \cite{spoto:pair_sharing}, can be used as purity
analysis.

For each $\methodsig$, a denotation $\shden{\methodsig}$ is given: for
$\abselemsp\tuple{\methodsig}$ safely describing the deep-sharing and
purity between actual arguments in the input state,
$\abselemsp'\tuple{\methodsig}=\shden{\methodsig}(\abselemsp\tuple{\methodsig})$
is such that (1) if $\SHARE{v}{w}\in \abselemsp'\tuple{\methodsig}$,
then $v$ and $w$ might become deep-sharing during the execution of
$\methodsig$; and (2) $\PURE{v}_i\in \abselemsp'\tuple{\methodsig}$
means that the $i$-th argument might be impure.  In the following, the
domain $\cI_{sp}^{\typenv}$ will combine deep-sharing and purity
information: $\SHARE{v}{w} \in \abselemsp$ means that $\abselemsp$
allows $v$ and $w$ to deep-share; and $\PURE{v}_i \in \abselemsp$
means that $\abselemsp$ allows the $i$-th argument of the method under
consideration to be impure.

\subsubsection{Operations on abstract values}
\label{sec:operationsOnAbstractValues}

\emph{Projection} $\exists v \abselemfrc$ (easily extensible to sets
of variables) of $\abselemfrc$ sets $\abselemfc(v)$,
$\abselemfr(v,v)$, any $\abselemfr(w_1,v)$, and any
$\abselemfr(v,w_2)$ to $\false$, leaving the rest unchanged.

\emph{Renaming} $\abselemfrc[v/w]$ replaces $v$ by $w$: the result
$\abselemfrc'$ is such that $\abselemfc'(w) = \abselemfc(v)$ and
$\abselemfc'(v) = \false$; $\abselemfr'(v,v) = \abselemfr'(v',v) =
\abselemfr'(v,v') = \false$ for every $v'$; moreover,
$\abselemfr'(w,w) = \abselemfr'(v,v)$, $\abselemfr'(w,v') =
\abselemfr(v,v')$ and $\abselemfr'(v',w) = \abselemfr(v',v)$ for every
$v' \neq v$.

\emph{Copy} $\abselemfrc[v+w]$ is similar to renaming but $v$ is not
removed: the result $\abselemfrc'$ is s.t.
\begin{itemize}
\item $\abselemfc'(w) = \abselemfc'(v) = \abselemfc(v)$; 
\item $\abselemfr'(w,w) = \abselemfr'(v,v) = \abselemfr(v,v)$ and
  $\abselemfr'(v,w) = \abselemfr'(w,v) = \abselemfr(v,v)$;
\item $\abselemfr'(w,v') = \abselemfr'(v,v') = \abselemfr(v,v')$ and
  $\abselemfr'(v',w) = \abselemfr'(v',v) = \abselemfr(v',v)$ if
  $v'\notin\{v,w\}$.
\end{itemize}

Finally, \emph{update} $\UPDATE{\abselemfrc}{(v,w)}{\pformula}$ sets
$\abselemfr(v,w)$ to $\pformula$, leaving the rest unchanged, and
$\UPDATE{\abselemfrc}{v}{\pformula}$ sets $\abselemfc(v)$ to
$\pformula$.

\subsubsection{Path-formul\ae}
\label{sec:operationsOnPathFormulae}

The \emph{path-concatenation} operator $\odot: \pformulae \times
\pformulae \mapsto \pformulae$ is used to combine formul\ae~when
concatenating paths.  The path-formula $\pformula \odot \pformulag$
has the following models: $\{~\finterpretation' \cup
\finterpretation''~|~\finterpretation'{\in}\fmodels{\pformula}~\wedge~
\finterpretation''{\in}\fmodels{\pformulag}~\wedge~
\finterpretation'~\mbox{and}~\finterpretation''~\mbox{are viable}~\}$.
In other words, the models of $\pformula \odot \pformulag$ are
obtained by ``concatenating'' the models of $\pformula$ with those of
$\pformulag$.  This makes sense because of the following lemma.

\begin{lemma}[(path-concatenation)]
  \label{lemma:odotAndPathConcatenation}
  Let $\hpath'$ and $\hpath''$ be two paths such that the last
  location of $\hpath'$ is the first of $\hpath''$.  Then,
  $\traverses{\hpath'}{\pformula}$ and
  $\traverses{\hpath''}{\pformulag}$ imply $\traverses{\hpath'
    \hpconcat \hpath''}{\pformula \odot \pformulag}$.
\end{lemma}

It is easy to see that $\odot$ preserves equivalence of
path-formul\ae: if $F_1 \equiv F_2$ and $G_1 \equiv G_2$, then $F_1
\odot F_2 \equiv G_1 \odot G_2$ since only viable models are
considered.

The \emph{path-difference} operator $\ominus: \pformulae \times
\pformulae \mapsto \pformulae$ defines $\fmodels{\pformula \ominus
  \pformulag}$ to be

\centerline{$\{ \finterpretation'{\setminus}
  X~|~\finterpretation'{\in}\fmodels{\pformula} \wedge
  \finterpretation''{\in}\fmodels{\pformulag} \wedge
  X{\subseteq}\finterpretation'' \wedge
  \finterpretation'~\mbox{and}~\finterpretation''~\mbox{are
    viable}\}$}
\noindent
Note that every model of $\pformula$ is still a model of $\pformula
\ominus \pformulag$, since $\emptyset$ is a subset of all sets.  The
use of this operation is motivated by Lemmas \ref{lemma:ominus2} and
\ref{lemma:ominus3}: $\ominus$ models path difference.

\begin{lemma}
  \label{lemma:ominus2}
  Let $\hpath$ be $\hpath' \hpconcat \hpath''$; let
  $\traverses{\hpath}{\pformula}$ and
  $\traverses{\hpath'}{\pformulag}$.  Then,
  $\traverses{\hpath''}{\pformula \ominus \pformulag}$.
\end{lemma}

\begin{lemma}
  \label{lemma:ominus3}
  Let $\hpath$ be $\tuple{\ell_0,\ell_1..,\ell_k}$ and $\hpath'$ be
  $\tuple{\ell_1,..,\ell_k}$.  Let the path from $\ell_0$ to $\ell_1$
  traverse $\field$, and $\hpath$ p-satisfy $\pformula$.  Then,
  $\traverses{\hpath'}{\pformula \ominus \onlyfieldsp{\field}}$.
\end{lemma}


\noindent
\begin{figure}[t]
  \begin{minipage}{12cm}
    \[
    \small
    \begin{array}{@{}l@{~}rl@{}}
      (1_e) & \EXPASEMANTICS{\typenv}{\absinterp}{n}(\abselemfrc) =
      & \abselemfrc \\
      (2_e) & \EXPASEMANTICS{\typenv}{\absinterp}{\nil}(\abselemfrc)
      = & \abselemfrc \\
      (3_e) &
      \EXPASEMANTICS{\typenv}{\absinterp}{\newk{\class}}(\abselemfrc)
      = &
    \left(\UPDATE{\abselemfr}{(\res,\res)}{\pformulaempty},\UPDATE{\abselemfc}{\res}{\pformulaempty}\right)   
      \\
      (4_e) & \EXPASEMANTICS{\typenv}{\absinterp}{v}(\abselemfrc) =
      & \mbox{if $\typenv(v){=}\integer$ then $\abselemfrc$ else
        $\abselemfrc[v+\res]$}
      \\
      (5_e) & \EXPASEMANTICS{\typenv}{\absinterp}{\mathit{exp}_1{\oplus}
        \mathit{exp}_2}(\abselemfrc) = & \exists
      \rho.\EXPASEMANTICS{\typenv}{\absinterp}{\mathit{exp}_2}\left(\exists
      \rho.\EXPASEMANTICS{\typenv}{\absinterp}{\mathit{exp}_1}\left(\abselemfrc\right)\right)
      \\
      (6_e) &
      \EXPASEMANTICS{\typenv}{\absinterp}{v.\field}(\abselemfrc) = &
      \mbox{if $\field$ has type \integer then $\abselemfrc$ else}~
      \abselemfrc \sqcupf \abselemfrc'~\mbox{where} \\
      & & \quad \begin{array}{rcl@{\quad}l}
        \abselemfc'(\res) & = & \abselemfc(v) \\
        \abselemfc'(w) & = & \false & ~\mbox{for
          every}~w\neq \rho \\
        \abselemfr'(\res,\res) & = & \abselemfc(v) & \\
        \abselemfr'(\res,w) & = & \abselemfr(v,w)
        \ominus \onlyfieldsp{\field} & ~\mbox{for
          every}~w\neq \rho \\
        \abselemfr'(w,\res) & = & \left\{ \begin{array}{l@{\quad}l}
            \abselemfr(w,v) \odot \onlyfieldsp{\field} \\
            \true \end{array} \right. &
            \begin{array}{l}
            \mbox{if}~\SHARE{w}{v} \notin \abselemsp\\
            \mbox{if}~\SHARE{w}{v} \in \abselemsp \end{array}
         \\
        \abselemfr'(w_1,w_2) & = & \false &
        ~\mbox{in all the other cases}
      \end{array}\\
    \end{array} 
    \]
  \end{minipage}
  
  \caption{The abstract semantics for expressions}
  \label{fig:abs-sem-exp}
\end{figure}


\noindent
\begin{figure}[t]
  \begin{minipage}{12cm}
    \[
    \small
    \begin{array}{@{}l@{~}rl@{}}
      (7_e)
      &\EXPASEMANTICS{\typenv}{\absinterp}{v_0.\method(v_1,..,v_n)}(\abselemfrc)
      =& \abselemfrc ~ \sqcupf ~ \abselemfrc'' ~ \sqcupf ~
      \abselemfrc''' ~ \sqcupf ~
      \abselemfrc''''~\mbox{where}~\bar{v}{=}\{v_0,..,v_n\}~\mbox{and}
      \\ 
      &\multicolumn{2}{l}{~~~~~~~\abselemfrc' {=}
        \project{\abselemfrc}{(\typenv{\setminus}\bar{v})} \qquad
        \qquad \abselemsp'= 
        \project{\abselemsp}{(\typenv{\setminus}\bar{v})}} \\
      &\multicolumn{2}{l}{~~~~~~~
        \abselemfrc'' = \sqcupf~\{ 
        ~(\absinterp(\methodsig)(\abselemfr'[\bar{v}/\inp{\methodsig}]))[\inp{\methodsig}/\bar{v},\out/\res]~
        ~|~~~\methodsig \mbox{ can be called here}~
        \} }
      \\
      &\multicolumn{2}{l}{~~~~~~~\abselemsp''=\cup\{~
        \shden{\methodsig}(\abselemsp'[\bar{v}/\inp{\methodsig}])[\inp{\methodsig}/\bar{v},\out/\res]
        ~|~~~\methodsig \mbox{ can be called here}~\}
      }\\
      &\multicolumn{2}{l}{~~~~~~~\abselemfr^{ij}(w_1,w_2) = \left\{
        \begin{array}{l@{\quad}l}
          \abselemfr(w_1,v_i) \odot \abselemfr''(v_i,v_j) \odot
          \abselemfr(v_j,w_2) & 
          \mbox{if}~\left( \begin{array}{ll}
            \SHARE{w_1}{v_i} \notin \abselemsp \wedge
            \SHARE{v_i}{v_j} \notin \abselemsp'' & \wedge \\
            \abselemfr(v_j,w_2) \neq \false \wedge \PURE{v}_i \in
            \abselemsp''
            \end{array}
          \right) \\
          \abselemfr(w_1,v_i) \odot \true & 
          \mbox{if}~\left( \begin{array}{ll}
            \SHARE{w_1}{v_i} \notin \abselemsp \wedge
            \SHARE{v_i}{v_j} \in \abselemsp'' & \wedge \\
            \abselemfr(v_j,w_2) \neq \false \wedge \PURE{v}_i \in
            \abselemsp''
            \end{array}
          \right) \\
          \true \odot \abselemfr(v_j,w_2) & 
          \mbox{if}~\left( \begin{array}{ll}
            \SHARE{w_1}{v_i} \in \abselemsp \wedge
            \SHARE{v_i}{v_j} \notin \abselemsp'' & \wedge \\
            \abselemfr(v_j,w_2) \neq \false \wedge \PURE{v}_i \in
            \abselemsp''
            \end{array}
          \right) \\
          \true &
          \mbox{if}~\left( \begin{array}{ll}
            \SHARE{w_1}{v_i} \in \abselemsp \wedge
            \SHARE{v_i}{v_j} \in \abselemsp'' & \wedge \\
            \abselemfr(v_j,w_2) \neq \false \wedge \PURE{v}_i \in
            \abselemsp''
            \end{array}
          \right) \\
          \false & \mbox{otherwise}
        \end{array} \right.
      }
      \\
      &\multicolumn{2}{l}{~~~~~~~
        \abselemfr''' = \sqcupf~\{~
        \abselemfr^{ij}~|~i,j \in \{1..n\} \} }
      \\
      &\multicolumn{2}{l}{~~~~~~~
        \abselemfr''''
    = \UPDATE{\abselemfr'''}{(\res,w)}{\vee_{0\leq i\leq
    n} \pformula_i} \quad \mbox{for each}~w}
      \\
      &\multicolumn{2}{l}{~~~~~~~\qquad \mbox{where}~\pformula_k
    = \left\{
    \begin{array}{l@{\qquad}l}
      \true & \mbox{if}~\SHARE{v_k}{\rho} \in \abselemsp'' \\
      (\abselemfr''(\res,v_k)\odot\abselemfr(v_k,w)) \vee
      (\abselemfr(v_k,w) \ominus \abselemfr''(v_k,\res)) &
      \mbox{otherwise}
      \end{array}
    \right.
      }\\
      &\multicolumn{2}{l}{~~~~~~~
        \abselemfc^i(w) = \left\{ \begin{array}{l@{\quad}l}
          \abselemfc''(v_i) &
          \mbox{if}~\PURE{v_i}\in \abselemsp'' \wedge \left(\SHARE{w}{v_i}
          \in \abselemsp \vee \abselemfr(w,v_i) \neq\false
            \vee \abselemfr(v_i,w) \neq\false \right) \\
          \false & \mbox{otherwise} \end{array} \right. }
      \\
      &\multicolumn{2}{l}{~~~~~~~
        \abselemfc''' = \sqcupf~\{~
        \abselemfc^{i}~|~i \in \{1..n\} \} }
      \\
      &\multicolumn{2}{l}{~~~~~~~
      \abselemfc'''' = \UPDATE{\abselemfc'''}{~\res}{\bigvee \{~\abselemfc(v_k)
    ~|~0\leq k \leq n,~\abselemfr''(v_k,\res) \neq \false~\}~}} \\
    \end{array} 
    \]
  \end{minipage}
  
  \caption{The abstract semantics for method calls}
  \label{fig:abs-sem-meth}
\end{figure}

\subsection{Expressions}
\label{sec:abs-sem-exp}

Figures \ref{fig:abs-sem-exp} and \ref{fig:abs-sem-meth} describe how
the abstract semantics $\EXPASEMANTICS{\typenv}{\absinterp}{\_}$ works
on expressions.  It is based on a type environment $\typenv$ (left
implicit) and an interpretation $\absinterp$ on methods.  The special
variable $\res$ represents the result of evaluating the expression.
It is easy to see that p-formul\ae~which are not in normal form are
never generated.

\subsubsection*{Easy cases}

As expected, the evaluation of an $\integer$ value (case $1_e$) or
$\nil$ (case $2_e$) does not modify the current abstract value (i.e.,
there is no new reachability/cyclicity) since $\res$ will not have any
relation with any existing variable (either because it is a number or
because it does not point to a valid heap location).  When a new
object is created (case $3_e$), the sharing information does change,
but this is left implicit.  More importantly, $\res$ is correctly
represented as (1) reaching itself through an empty path (only this
kind of paths can p-satisfy $\pformulaempty$), which means that it
aliases with itself; and (2) similarly, being cyclic because of an
empty path.

In case $4_e$, information about $v$ is copied to $\res$, without
removing the original information as $v$ is still accessible.  In case
$5_e$, $\oplus$ stands for a binary operation on $\integer$;
\emph{side effects} are the only possible source of new information.

\subsubsection*{Field access}

Case $6_e$ is harder: if the declared type of $\field$ is a reference
type, then the new abstract value is obtained by adding to the old one
the following information.
\begin{itemize}
\item The cyclicity information about $v$ contained in $\abselemfrc$
  affects $\abselemfrc'$ in two ways: the path-formula is ``copied''
  into both $\abselemfc'(\res)$ and $\abselemfr'(\res,\res)$
  indicating that the cyclicity of $v$ implies that (1) the data
  structure reachable from $v.\field$ is still possibly cyclic; and
  (2) it is also possible that the location pointed to by $v.\field$
  is part of the cycle, so that it can be reachable from itself.  The
  corresponding path-formula is copied as it is: the new path-formula
  $\abselemfc'(\res) = \abselemfr'(\res,\res)$ is not greater than
  $\abselemfc(v)$ because the set of cycles reachable from $v.\field$
  is a subset of those reachable from $v$, so that they will satisfy
  the same condition (i.e., p-satisfy the same path-formula); on the
  other hand, it cannot be smaller because, by soundness, it is not
  possible to refine the condition.  Note that the definition also
  works if $\abselemfc(v) = \false$: in this case, $v$ was guaranteed
  to be acyclic, so that $\res$ is still guaranteed to be acyclic.
\item The cyclicity of all the other variables is not modified:
  $\abselemfc'(w) = \false$ implies that the final cyclicity
  information $\abselemfc(w) \sqcupfc \abselemfc'(w)$ for any $w$ is
  still the old $\abselemfc(w)$.
\item If $v$ can $\pformula$-reach some $w$ in $\abselemfrc$, then
  $v.\field$ can also reach $w$ since it could be exactly on the path
  from $v$ to $w$.  Therefore, $\abselemfr'(\res,w)$ is set to
  $\pformula'$ where $\pformula'$ is obtained from $\abselemfr(v,w)$
  via $\ominus$ (Section \ref{sec:operationsOnPathFormulae}).  This
  means that it is no longer possible to guarantee that $\field$ will
  be traversed by a path from $\statef{\statesym}(\res)$ to
  $\statef{\statesym}(w)$.
\item Every $w$ possibly deep-sharing with $v$ may reach $\res$.  In
  fact, deep-sharing means that there is a location which is reachable
  from both, and such a location could be exactly the one pointed to
  by $\res$.  In this case, the corresponding path-formula is $\true$
  because it is not possible to put any condition on paths
  (deep-sharing as it is used by $\FIELDBASED$ is field-insensitive).
  On the contrary, if $w$ and $v$ do not deep-share, then the
  following information is added: if $w$ $\pformula$-reaches $v$, then
  $\abselemfr'(w,\res)$ is set to $\pformula \odot
  \onlyfieldsp{\field}$, indicating that any new path from
  $\statef{\statesym}(w)$ to $\statef{\statesym}(\res)$ will traverse
  $\field$ (old paths are already accounted for by $\abselemfr$).
  Note that $\false \odot \pformula = \false$, so that any $w$ not
  reaching $v$ will not reach $\res$, as expected.
\end{itemize}

\begin{example}[(field access)]
  \label{ex:abs-sem-exp}
  Consider the heap depicted in the right-hand side of the figure,
  which is the result of executing the program on the left-hand side.
  \begin{center}
    \begin{minipage}{5cm}
      \begin{lstlisting}
  x := new $\class$;
  y := new $\class$;
  x.f := y;
  x.g := new $\class$;
  z := new $\class$;
  y.h := z;
  y := null;
  x := x.f;
      \end{lstlisting}
    \end{minipage}
    \begin{minipage}{5cm}
      \begin{tikzpicture}
        \node (v) at (2,1.2) {$\xx$};
        \node (w) at (4,1.2) {$\zz$};
        \node[rectangle,draw] (l1) at (0,0) {$o_1:\class$};
        \node[rectangle,draw] (l2) at (2,0.5) {$o_2:\class$};
        \node[rectangle,draw] (l3) at (4,0.5) {$o_3:\class$};
        \node[rectangle,draw] (l4) at (2,-0.5) {$o_4:\class$};
        \draw[dotted] (v) -- (l2);
        \draw[dotted] (w) -- (l3);
        \draw[->] (l1) -- node[above] {\ffield} (l2);
        \draw[->] (l2) -- node[above] {\hfield} (l3);
        \draw[->] (l1) -- node[above] {\gfield} (l4);
      \end{tikzpicture}
    \end{minipage}
  \end{center}
  Let the abstract value $\abselemfrc^7$ computed by the analysis
  after line 7 be such that $\abselemfr^7(\xx,\zz) =
  \onlyfields{\fproposition{\ffield}, \fproposition{\hfield}}$ (the
  computation of such an abstract value will be explained in Example
  \ref{ex:abs-sem-com}).  The new path-formula $\abselemfr^8(\xx,\zz)$
  after line 8 is obtained by $\ominus$: first,
  $\abselemfr^7(\res,\zz)$ is updated with $\abselemfr^7(\xx,\zz)
  \ominus \onlyfieldsp{\ffield}$; afterward (see case $2_c$ in Figure
  \ref{fig:abs-sem-com}), the new path-formula is copied to
  $\abselemfr^8(\xx,\zz)$.  The final path-formula will be
  $\fproposition{\hfield} \wedge \bigwedge \{ \lnot
  \fproposition{\field}~|~\field {\notin} \{ \ffield, \hfield \} \}$
  since it is no longer possible to guarantee that all paths from
  $\xx$ to $\zz$ traverse $\ffield$.  Anyway, note that $\ffield$ does
  not appear as a negative literal either, as it is still possible
  that some path traverses it.
\end{example}

\subsubsection*{Method call}

Finally, case $7_e$ in Figure \ref{fig:abs-sem-meth} describes the
behavior of $\FIELDBASED$ on method calls.  Note that methods without
return value are not included in the language; however, they could be
easily dealt with by slightly modifying this case.  As usual in
Object-Oriented programs, a reference variable $v$ with declared type
$\dectype(v) = \class$ may store at runtime any object of type
$\class'{\subclasseq}\class$.  The set of possible runtime types of
$v$ can be computed statically by \emph{class analysis}
\cite{SpotoJ03} whenever needed; if such an analysis is not available,
then it can be taken, conservatively, as $\{
\class~|~\class{\subclasseq} \dectype(v) \}$.  Abstract values
$\abselemfrc'$ and $\abselemsp'$ are obtained by restricting the
corresponding initial values to the actual parameters $\bar{v}$ of
$\method$.  Also, $\abselemfrc''$ and $\abselemsp''$ come from
applying the denotation of $\methodsig$ for, resp.,
reachability/cyclicity (see Section \ref{sec:fixpoint}) and
deep-sharing/purity (which is taken as pre-computed information, see
Section \ref{sec:auxiliaryAnalyses}).

For every two actual parameters $v_i$ and $v_j$, the non-purity of
$v_i$ implies that it is possible to create a path in $\methodsig$
from $v_i$ (or any $w_1$ reaching it or deep-sharing with it) to $v_j$
(or any $w_2$ reachable from it).  This is taken into account by
$\abselemfr^{ij}$ (note that there is an abstract value
$\abselemfr^{ij}$ for every pair of parameters $(v_i,v_j)$, and all
the values are combined into $\abselemfr'''$), and happens because
modifying the data structure pointed to by $v_i$ during the execution
of $\method$ can possibly create a new path from $w_1$ to $w_2$.  Four
cases (plus an ``otherwise'' fifth case) are considered, depending on
whether deep-sharing between $w_1$ and $v_i$ before the call or
between $v_i$ and $v_j$ after the call is possible; all cases apply
only if $v_j$ may reach $w_2$\footnote{Actually, in the first and
  third case there is no need of such a condition since $\false$ is
  the absorbing element of $\odot$; for instance, in the first case,
  $\abselemfr(w_1,v_i) \odot \abselemfr''(v_i,v_j) \odot \false$ comes
  to be $\false$ anyway.} and $v_i$ is not pure.  Basically, each of
the four cases deals with the different scenarios with respect to
deep-sharing: since this property is not field-sensitive, some field
information about paths is lost whenever two variables may deep-share,
since there can be complex paths in the heap for which no
field-sensitive information is available (and which may be hidden from
reachability analysis).
\begin{itemize}
\item The first case models a scenario where all new paths from $w_1$
  to $v_i$ and from $v_i$ to $v_j$ are captured by reachability
  path-formul\ae, so that it is possible to say that any path from
  $w_1$ to $w_2$ must traverse a sub-path captured by
  $\abselemfr(w_1,v_i)$, then another sub-path captured by
  $\abselemfr''(v_i,v_j)$ and, finally a third one captured by
  $\abselemfr(v_j,w_2)$.  In order to account for this situation, the
  path-formula $\abselemfr(w_1,v_i) \odot \abselemfr''(v_i,v_j) \odot
  \abselemfr(v_j,w_2)$ is returned.
\item The second case models a partial loss of information due to
  possible deep-sharing between $v_i$ and $v_j$.  In this case, it is
  still possible to say that the first part of any new path from $w_1$
  to $w_2$ is captured by $\abselemfr(w_1,v_i)$, but nothing can be
  inferred about the rest of the path.
\item The third case is dual: here, the loss of information occurs in
  the first part of the path, since there is no field-sensitive
  information about paths starting from $w_1$ and reaching $v_i$.
\item Finally, the fourth case happens when deep-sharing is possible
  both between $w_1$ and $v_i$ and between $v_i$ and $v_j$.  In this
  case, reachability from $w_1$ to $w_2$ must be admitted as a
  possibility, but no field-sensitive information can be gathered, so
  that $\true$ has to be returned.
\end{itemize}

The following example describes the loss of information due to
deep-sharing.

\begin{example}[(method call)]
  \label{ex:methodCall}
  Consider the figure below (only solid lines): let $v_i$ and $v_j$ be
  parameters of $\method$ which are initially not sharing; let $w_1$
  (which is not a parameter) deep-share with $v_i$, and $v_j$ reach
  $w_2$.
  
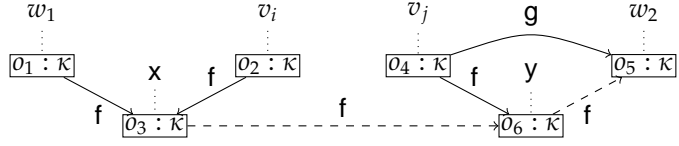
\begin{center}
    \begin{tikzpicture}
      \tikzstyle{objnode} = [rectangle,draw,inner sep=1pt]
      \node (w1) at (0,0.7) {$w_1$};
      \node (vi) at (3,0.7) {$v_i$};
      \node (vj) at (5,0.7) {$v_j$};
      \node (w2) at (8,0.7) {$w_2$};
      \node[objnode] (l1) at (0,0) {$o_1:\class$};
      \node[objnode] (l2) at (3,0) {$o_2:\class$};
      \node[objnode] (l3) at (1.5,-0.8) {$o_3:\class$};
      \node[objnode] (l4) at (5,0) {$o_4:\class$};
      \node[objnode] (l5) at (8,0) {$o_5:\class$};
      \node[objnode] (l6) at (6.5,-0.8) {$o_6:\class$};
      \draw[dotted] (w1) -- (l1);
      \draw[dotted] (vi) -- (l2);
      \draw[dotted] (vj) -- (l4);
      \draw[dotted] (w2) -- (l5);
      \draw[->] (l1) -- node[below] {\ffield} (l3);
      \draw[->] (l2) -- node[above] {\ffield} (l3);
      \draw[->] (l4) .. controls (6.5,0.5) .. node[above] {\gfield} (l5);
      \draw[->] (l4) -- node[above] {\ffield} (l6); 
      \node (x) at (1.5,-0.1) {\xx};
      \node (y) at (6.5,-0.1) {\yy};
      \draw[->,dashed] (l3) -- node[above] {\ffield} (l6);
      \draw[->,dashed] (l6) -- node[below] {\ffield} (l5);
      \draw[dotted] (x) -- (l3);
      \draw[dotted] (y) -- (l6);
    \end{tikzpicture}
  \end{center}
  If the instructions \lstinline!x:=$v_i$.f; y:=$v_j$.f; x.f:=y!  are
  executed in $\method$, then $v_i$ and $v_j$ become deep-sharing.
  Moreover, executing \lstinline!y.f:=$v_j$.g! creates a path from
  $w_1$ to $w_2$, depicted by dashed lines in the figure above.  Note
  that (1) the new deep-sharing between $v_i$ and $v_j$ is reflected
  by the denotation of $\method$; (2) there is never any reachability
  between $v_i$ and $v_j$, so that field-sensitive information is not
  available outside $\method$ (i.e., when applying its denotation);
  and (3) there is no way to create the path unless $v_j$ is reaching
  $w_2$.  This situation falls into the fourth case of the semantics,
  so that $\true$ is returned as the reachability from $w_1$ to $w_2$.
\end{example}

Note that, in the example, $o_3$ happens to be a \emph{cutpoint}
\cite{RinetzkyBRSW05}, i.e., an object which is (a) reachable from a
parameter of $\method$ in at least one step; and (b) also reachable by
traversing a path which does not include any object which is reachable
from any parameter of $\method$.  The existence of a cutpoint is
possible in the third and fourth case above, where $\SHARE{w_1}{v_i}
\in \abselemsp$.  Techniques similar to \cite{RinetzkyBRSW05} can be
used to deal with such cases; alternatively, the analysis could be
limited to cutpoint-free programs \cite{KreikerRRSWY13}, since
cutpoint-freeness is a decidable property.

Observation (3) in Example \ref{ex:methodCall} shows that the
reachability between two variables $w_1$ and $w_2$ will be certainly
taken into account by some of the $\abselemfr^{ij}$, namely, the ones
where $v_i$ and $v_j$ become sharing (in the normal sense, not
deep-sharing), and $w_1$ shares with $v_i$, and $v_j$ reaches $w_2$.

The next step is to propagate the information about some $v_i$ in
$\abselemfr'''$ to $\res$ whenever $v_i$ (standard-)shares with $\res$
after the call.  This is needed in order to take into account some
cases similar to Example \ref{ex:rhoCloning}.

\begin{example}[(return value)]
  \label{ex:rhoCloning}
  Consider the following fragment:

  \begin{minipage}{7cm}
    \begin{lstlisting}
  $\class$ m() { 
    $\class$ a;  a := new $\class$;
    a.f := this.f;
    return a; }
    \end{lstlisting}
  \end{minipage}
  \begin{minipage}{5cm}
    \begin{lstlisting}[firstnumber=5]
  x1.g := y;
  x2.f := x1;
  z := x2.m();
    \end{lstlisting}
  \end{minipage}

  \noindent
  The abstract semantics computes the formula $\pformula =
  \onlyfields{\fproposition{\ffield},\fproposition{\gfield}}$ for the
  reachability from \xtwo to \yy after line 6.  After line 7, \zz
  should be reaching \yy, but this cannot be taken into account by the
  denotation of \mm since \yy is not a parameter.  In fact, the
  denotation of \mm can only detect that $\res$ and \lstinline!this!
  share at the end of \mm, i.e., that \zz and \xtwo share after the
  call.  The only way to be sound here (note that \lstinline!this! is
  pure, so that this situation is not detected by any
  $\abselemfr^{ij}$) is to copy the reachability from \xtwo to \yy to
  the reachability from $\res$ to \yy, which is in turn copied into
  $\abselemfr(\zz,\yy)$.  However, this is still unsound because the
  condition on paths could change: the two cases for $\pformula_k$ in
  the abstract semantics account for the different scenarios: (1) if
  $\res$ reaches \lstinline!this!, then the new path is obtained by
  using $\odot$; (2) if \lstinline!this! reaches $\res$, then
  $\ominus$ is used; or (3) if they only deep-share, as in the
  example, then no information can be gathered, and $\true$ is
  returned.
\end{example}

As for cyclicity, each $\abselemfc^i$ deals with cases where a cycle
is built in $\method$ which is reachable from $v_i$, and $w$ was
sharing\footnote{Note that standard sharing is the disjunction between
  deep-sharing and both directions of reachability.} with $v_i$; in
this case, $w$ also becomes possibly cyclic.  The rest of the
treatment of cyclicity is similar to reachability.


\begin{figure}[t]
    \begin{minipage}{12cm}
      \[
      \small
      \begin{array}{@{}l@{~}rl@{}}
        (1_c) & 
        \ASEMANTICS{\typenv}{\absinterp}{\skipc}(\abselemfrc) = & \abselemfrc
        \\[1mm]
        (2_c) & 
        \ASEMANTICS{\typenv}{\absinterp}{v\assign
          \mathit{exp}}(\abselemfrc) = & 
        (\exists
        v.\EXPASEMANTICS{\typenv}{\absinterp}{\mathit{exp}}(\abselemfrc))[\res/v] 
        \\[1mm]
        (3_c) & 
        \ASEMANTICS{\typenv}{\absinterp}{v.\field\assign
          \mathit{exp}}(\abselemfrc) =&
        \exists\res.(\abselemfrc' \sqcupf \abselemfr'' \sqcupf
        \abselemfc'') ~\mbox{where} \\
        & \abselemfrc' = &
        \EXPASEMANTICS{\typenv}{\absinterp}{\mathit{exp}}(\abselemfrc)
        \\
        &
        \abselemfr''(w_1,w_2) = & \abselemfr'(w_1,v) \odot
        (\onlyfieldsp{\field}{\vee}(\onlyfieldsp{\field}\odot\abselemfr'(\res,v)))
        \odot \abselemfr'(\res,w_2)
        \\
        & \abselemfc''(w) = & \left \{
          \begin{array}{l@{\qquad\qquad}l}
            \left(\abselemfr'(\res,v) \odot \onlyfieldsp{\field}\right)
            \vee \abselemfc(\res)&
            \mbox{if}~\abselemfr'(w,v) \neq \false \\
            \false & \mbox{otherwise}
          \end{array}
        \right.
        \\[1mm]
        (4_c) & \ASEMANTICS{\typenv}{\absinterp}{
          \begin{array}{@{}rl@{}}
            \ifte~\mathit{exp} & \iftethen~\mathit{com}_1 \\
            & \ifteelse~\mathit{com}_2 
          \end{array}
        }
        (\abselemfrc) = & 
        \abselemfrc^1 \sqcupf \abselemfrc^2 ~ \mbox{where} \\
        & & \abselemfrc^1 =
        \ASEMANTICS{\typenv}{\absinterp}{\mathit{com}_1}(\abselemfrc) \\
        & & \abselemfrc^2 =
        \ASEMANTICS{\typenv}{\absinterp}{\mathit{com}_2}(\abselemfrc) \\[1mm]
        (5_c) &
        \ASEMANTICS{\typenv}{\absinterp}{\while~\mathit{exp}~\whilebody~\mathit{com}}(\abselemfrc)
        =& \absden(\abselemfrc) \mbox{ where } \absden \mbox{ is the
          least fixpoint of}\\
        & & 
        \lambda w. \lambda
        \abselemfrc.w(\ASEMANTICS{\typenv}{\absinterp}{\mathit{com}}(\abselemfrc))
        \\[1mm]
        (6_c) & \ASEMANTICS{\typenv}{\absinterp}{\mathit{com}_1; \mathit{com}_2}( \abselemfrc) =&
        \ASEMANTICS{\typenv}{\absinterp}{\mathit{com}_2}(\ASEMANTICS{\typenv}{\absinterp}{\mathit{com}_1}(\abselemfrc))
        \\[1mm]
        (7_c) & \ASEMANTICS{\typenv}{\absinterp}{\return~\mathit{exp}}(\abselemfrc) =& 
        \EXPASEMANTICS{\typenv}{\absinterp}{\mathit{exp}}(\abselemfrc)[\res/\out]
        \\
      \end{array} 
\]
    \end{minipage}

  \caption{The abstract semantics for commands}
  \label{fig:abs-sem-com}
\end{figure}

\subsection{Commands}
\label{sec:abs-sem-com}

Figure \ref{fig:abs-sem-com} shows the behavior of the abstract
semantics $\ASEMANTICS{\typenv}{\absinterp}{\_}$ on commands.  Easy
cases are considered first, leaving field update at the end.

Case $1_c$ is trivial.  Case $2_c$ for variable assignment is also
easy: the semantics evaluates $\mathit{exp}$, which could have side
effects, and copies the information about $\res$ to $v$, after
removing the information about $v$ since its initial value will be
lost.  Note that the information about \emph{the location pointed to
  by} $v$ needs not be lost, since there could be other variables
pointing to it.

In cases $4_c$ and $5_c$, standard principles for the design of
abstract semantics are followed.  Both branches of the conditional are
analyzed\footnote{Recall that guards have no side effects.};
$\ASEMANTICS{\typenv}{\absinterp}{\_}$ is path-insensitive in that the
results obtained for each branch are simply combined by means of
$\vee$, applied point-wise.  In the case of loops, standard fixpoint
design is used.  Termination of the fixpoint is guaranteed by the fact
that $\domfrc$ does not allow infinite ascending chains $\langle
\pformula_0, .., \pformula_i, ..\rangle$ where $\pformula_j
\sqsubseteqfrc \pformula_{j+1}$ for each $j \geq 0$.  However, in
principle, there can be chains whose length is exponential on the
cardinality of $\fields$, as discussed in Section
\ref{sec:domainComparison}, so that convergence can be slow unless
some mechanism for speeding it up is used (e.g., some \emph{widening}
\cite{Cousot79} operator mapping path-formul\ae~to $\true$ whenever
abstract values have been updated more than $k$ times for some fixed
number $k$).

The last two cases, $6_c$ and $7_c$, are straightforward.

\subsubsection*{Field update}

In $v.\field\assign\mathit{exp}$, the heap is modified, and new paths
can be created.  In particular, a path $\hpath$ from the location
pointed to by $v$ to the location pointed to by $\rho$ is created,
which p-satisfies $\fproposition{\field}$.  The information about
$\hpath$ must be \emph{joined} with the original abstract value by
means of $\sqcupf$.

The abstract semantics focuses on two kinds of variables: the first
kind, \emph{pre-variables}, contains those $w_1$ which can reach $v$.
The second kind of variables, \emph{post-variables}, contains those
$w_2$ which can be reached from $\res$ after $\mathit{exp}$ has been
evaluated, and the abstract value $\abselemfrc'$ has been computed.
Clearly, $v$ is a pre-variable since it reaches itself.  Moreover, a
variable may belong to both kinds; in this case, it will be considered
twice.  The new reachability information $\abselemfr''$ must take into
account paths from all pre-variables $w_1$ to all-post variables
$w_2$, due to the creation of $\hpath$.  The new paths certainly
p-satisfy $\abselemfr'(w_1,v) \odot
(\onlyfieldsp{\field}{\vee}(\onlyfieldsp{\field}\odot\abselemfr'(\res,v)))
\odot \abselemfr'(\res,w_2)$ (see Lemma
\ref{lemma:odotAndPathConcatenation}).

Note that both the newly-created path (represented by
$\onlyfieldsp{\field}$) and the possible cycle (represented by
$\onlyfieldsp{\field}\odot\abselemfr'(\res,v)$) which is created by
the update are considered; this will become more clear when discussing
the example in Section \ref{sec:backExample2}.  Note also that this
formula will \emph{not} necessarily be the final reachability
information about $w_1$ and $w_2$.  In fact, suppose that there
existed another path from $w_1$ to $w_2$, completely disjoint from the
new one, and p-satisfying $\pformulag$: in this case, the final
reachability between $w_1$ and $w_2$ will be $\pformulag \vee
(\abselemfr'(w_1,v) \odot
(\onlyfieldsp{\field}{\vee}(\onlyfieldsp{\field}\odot\abselemfr'(\res,v)))
\odot \abselemfr'(\res,w_2))$, due to the use of $\sqcupf$.

The cyclicity information comes from two cases: the definition of
$\abselemfc''$ considers cycles created by closing existing paths from
$\res$ to $v$.  Every $w$ reaching $v$ will become possibly cyclic,
and the associated path-formula will be the combination of the old
path with $\fproposition{\field}$.  Note that the conditions on the
cyclicity of $w$ do not need to take into account $\abselemfr'(w,v)$,
since, in general, it is not in the cycle (if it is, then this
information is already contained in $\abselemfc'$).

\begin{example}[(field update)]
  \label{ex:abs-sem-com}
  Consider the code of Example \ref{ex:abs-sem-exp}: it is easy to see
  that the abstract value $\abselemfrc^5$ after line 5 is such that
  $\abselemfr^5(\xx,\yy) = \onlyfieldsp{\ffield}$.  Moreover,
  $\abselemfr^5(\xx,\xx) = \abselemfr^5(\yy,\yy) =
  \abselemfr^5(\zz,\zz) = \pformulaempty$, and $\abselemfr^5(\_,\_) =
  \false$ everywhere else.  Such a value is obtained by observing that
  neither $\xx$ nor $\yy$ reach any other variable before line 5.
  Consequently, $\abselemfr^5(\xx,\yy)$ comes to be $\pformulaempty
  \odot \onlyfieldsp{\ffield} \odot \pformulaempty =
  \onlyfieldsp{\ffield}$ since $\pformulaempty$ is the neutral element
  for $\odot$.  Similarly, $\abselemfrc^6(\yy,\zz) =
  \onlyfieldsp{\hfield}$.  The path-formula $\abselemfrc^6(\xx,\zz) =
  \abselemfrc^7(\xx,\zz)$ is computed as $\onlyfieldsp{\ffield} \odot
  \onlyfieldsp{\hfield} = \onlyfields{\fproposition{\ffield},
    \fproposition{\hfield}}$, and its only model is $\{
  \fproposition{\ffield}, \fproposition{\hfield} \}$.
\end{example}

\subsection{Global fixpoint}
\label{sec:fixpoint}

The following definition defines the abstract denotational semantics
of a program $P$ as the \emph{least fixpoint} (\lfp) of an (abstract)
transformer of interpretations.  Variables $\bar{u}$ play the role of
\emph{shallow variables}.  Note that shallow variables appear at the
level of the semantics, rather than as a result of program
transformation; they are introduced in order to keep track of the data
structures to which input variables point at the beginning of a
method, since otherwise they could be lost if the corresponding
variables are updated.

\begin{example}
  Consider the following method $\methodsig$ (lines 1--4), invoked in
  line 5:

  \begin{minipage}{6cm}
    \begin{lstlisting}
  $\class$  mth($\class$ x1, $\class$ x2) {
    x1.f := x2;
    x1 := null;
    return x2; }
    \end{lstlisting}
  \end{minipage}
  \begin{minipage}{6cm}
    \begin{lstlisting}[firstnumber=5]
  z := mth(y1,y2);
    \end{lstlisting}
  \end{minipage}

  \noindent 
  Here, the reachability from the first to the second parameter of
  $\methodsig$ after its execution would be lost if a copy of them is
  not maintained as a shallow variable.  On the other hand, copying
  the information about \lstinline!x1! and \lstinline!x2! into shallow
  variables $u_1$ and $u_2$ allows detecting that $u_1$ is reaching
  $u_2$ and $\res$ at the end of $\methodsig$.  Afterward, this is
  copied back to actual parameters \lstinline!y1! and \lstinline!y1!
  (line 5) in order to be available after the call.
\end{example}

\begin{definition}
  \label{def:abs-den-semantics}
  The \emph{abstract denotational semantics} of a program $P$ is the
  \lfp of
  \[
  \abstp{P}(\absinterp) = \left\{ \methodsig\mapsto\lambda \abselemfrc
  \in \domrc{\inp{\methodsig}} \left(\project
      {\ASEMANTICS{\typenv}{\absinterp}{\body{\methodsig}}\left(
        \abselemfrc[\bar{w}+\bar{u}]\right)}
      {X}\right)[\bar{u}/\bar{w}] ~|~ \methodsig\in P \right\}
  \]
  
  \noindent
  where $\inp{\methodsig} = \{\this,w_1,\ldots,w_n\}$, and $\bar{u}$
  is a variable set $\{u_1,\ldots,u_n\}$ such that $\bar{u}\cap
  \scope{\methodsig}=\emptyset$; moreover, $\dom(\tau)=
  \locals{\methodsig}\cup \bar{u}$, and
  $X=\dom(\tau){\setminus}(\bar{u}\cup \{\this,\out\})$.
\end{definition}

\noindent
The operator $\abstp{P}$ is quite standard, and transforms the
interpretation $\absinterp$ by assigning a new denotation for each
$\methodsig$ defined in $P$, using existing denotations from
$\absinterp$.
The new denotation maps a given input value
$\abselemfrc\in\domfrctau{\inp{\methodsig}}$ to an output value from
$\domfrctau{\inp{\methodsig}\cup\{\out\}}$, as follows:
\begin{enumerate}
\item it obtains an abstract value $\abselemfrc^0=
  \abselemfrc[\bar{w}+\bar{u}]$ in which the parameters $\bar{w}$ are
  cloned into the shallow variables $\bar{u}$;
\item it applies the denotation of the code of $\methodsig$ to
  $\abselemfrc^0$, getting
  $\abselemfrc^1=\ASEMANTICS{\typenv}{\absinterp}{\body{\methodsig}}(\abselemfrc^0)$;
\item all variables but $\bar{u}\cup\{\this,\out\}$ are eliminated
  from $\abselemfrc^1$ (using $\exists X$); and
\item shallow variables $\bar{u}$ are finally renamed back to
  $\bar{w}$.
\end{enumerate}

\subsection{Soundness}
\label{sec:soundness}

This section discusses the soundness of \FIELDBASED.  Consider a
command $C$: soundness amounts to say that, for every initial state
$\statesym$ correctly represented by the initial abstract value (i.e.,
such that $\statesym \in \gammafrc(\abselemfrc)$), the final concrete
state $\statesym^* = \cinter{\typenv}{\interp}{C}(\statesym)$ is such
that $\statesym^* \in \gammafrc(\abselemfrc^*)$ where $\abselemfrc^* =
\ASEMANTICS{\typenv}{\absinterp}{C}(\abselemfrc)$.  This means that
every possible path from $w_1$ to $w_2$ created by $C$, and traversing
a set $X$ of fields, has to be reflected by a model $\finterpretation
= \{ \fproposition{\field}~|~\field\in X \}$ of
$\abselemfrc^*(w_1,w_2)$.  A similar observation holds for cyclicity.
Most of the evidence for soundness has been given while discussing the
abstract semantics; this section summarizes and completes the proof.
Only the most interesting cases are presented.

\subsubsection*{Field update, reachability}

Consider a field update $v.\ffield \assign v'$.  Given two variables
$w_1$ and $w_2$, every path $\hpath$ from $\statef{\statesym}(w_1)$ to
$\statef{\statesym}(w_2)$ is either an old one (already in the heap
before the update) or a newly-created one traversing $\ffield$ at
$\statef{\statesym}(v)$.  If $\hpath$ was already in the heap, then,
by hypothesis, it is represented by $\abselemfr$, and also by
$\abselemfr^*$ since $\abselemfr^* \sqsupseteqfr \abselemfr$.  On the
other hand, if $\hpath$ is new, then it is the concatenation of the
following sub-paths:
\begin{itemize}
\item a path $\hpath_1$ from $\statef{\statesym}(w_1)$ to
  $\statef{\statesym}(v)$, which, by hypothesis, is correctly
  represented by the initial abstract value $\abselemfr(w_1,v)$;
\item one of the following: either the path $\hpath'$ of length 1
  going from $\statef{\statesym}(v)$ to $\statef{\statesym}(v')$ and
  p-satisfying $\onlyfieldsp{\ffield}$; or the path $\hpath''$ going
  from $\statef{\statesym}(v)$ to $\statef{\statesym}(v')$, then to
  $\statef{\statesym}(v)$ and back to $\statef{\statesym}(v')$, and
  p-satisfying $\onlyfieldsp{\ffield} \odot \abselemfr(v',v)$; and
\item a path $\hpath_2$ from $\statef{\statesym}(v')$ to
  $\statef{\statesym}(w_2)$, which is correctly represented by
  $\abselemfr(v',w_2)$.
\end{itemize}

\noindent
The alternative of $\hpath''$ has to be taken into account because it
is possible that the field update closes a cycle, i.e., that there was
already a path from $\statef{\statesym}(v')$ to
$\statef{\statesym}(v)$ which has now become a cycle from
$\statef{\statesym}(v')$ to $\statef{\statesym}(v')$.  In this case,
there is a path from $\statef{\statesym}(w_1)$ to
$\statef{\statesym}(w_2)$ which goes until $\statef{\statesym}(v)$ and
$\statef{\statesym}(v')$, then traverses the cycle until reaching
$\statef{\statesym}(v')$ again, and finally reaches
$\statef{\statesym}(w_2)$.  Note that $\false$ is the absorbing
element for $\odot$, so that $\hpath''$ will be guaranteed to
p-satisfy $\false$ (i.e., not to exist) if the initial information was
able to exclude paths from $\statef{\statesym}(v')$ to
$\statef{\statesym}(v)$ (i.e., if $\abselemfr(v',v) = \false$).

\begin{example}
  Consider the following heap before executing $v.\ffield \assign v'$
  (which is the dashed line).  After the field update, $w_2$ is
  reachable from $w_1$ either directly or ``touching'' $o_5$ any
  number of times.  Note that a path touching $o_5$ once p-satisfies
  the same path-formula as one touching $o_5$ more than once.
  Reachability from $w_1$ to $w_2$ p-satisfies $\onlyfieldsp{\ffield}
  \vee \onlyfields{\fproposition{\ffield},\fproposition{\gfield}}$.
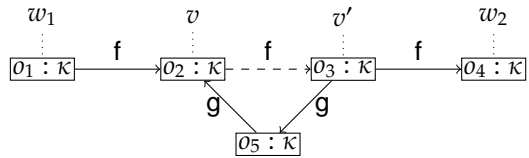
  \begin{center}
    \begin{tikzpicture}
      \tikzstyle{objnode} = [rectangle,draw,inner sep=1pt]
      \node (w1) at (0,0) {$w_1$};
      \node (v) at (2,0) {$v$};
      \node (vp) at (4,0) {$v'$};
      \node (w2) at (6,0) {$w_2$};
      \node[objnode] (o1) at (0,-0.7) {$o_1:\class$};
      \node[objnode] (o2) at (2,-0.7) {$o_2:\class$};
      \node[objnode] (o3) at (4,-0.7) {$o_3:\class$};
      \node[objnode] (o4) at (6,-0.7) {$o_4:\class$};
      \node[objnode] (o5) at (3,-1.7) {$o_5:\class$};
      \draw[dotted] (w1) -- (o1);
      \draw[dotted] (v) -- (o2);
      \draw[dotted] (vp) -- (o3);
      \draw[dotted] (w2) -- (o4);
      \draw[->] (o1) -- node[above] {$\ffield$} (o2);
      \draw[->] (o3) -- node[above] {$\ffield$} (o4);
      \draw[->,dashed] (o2) -- node[above] {$\ffield$} (o3);
      \draw[->] (o3) -- node[right] {$\gfield$} (o5);
      \draw[->] (o5) -- node[left] {$\gfield$} (o2);
    \end{tikzpicture}
  \end{center}
\end{example}

It is easy to see that both cases ($\hpath'$ and $\hpath''$) are dealt
with by the definition of $\abselemfr''$ in case $(3_c)$ of Figure
\ref{fig:abs-sem-com}.  This follows from the definition of $\odot$
and Lemma \ref{lemma:odotAndPathConcatenation}.

It is worth noting that the field update could break some old paths.
Conservatively, this is not taken into account by the semantics: the
removal of a path from a heap can never imply that a path-formula is
no longer p-satisfied, so that to update abstract values is not
needed.  More formally, it is possible to consider a field update as
the combination of two operations: first set $v.\ffield$ to $\semnil$,
then assigning it to $v'$.  Let $\statesym$ and $\statesym^-$ be two
states such that the heap of $\statesym^-$ is obtained by breaking
(e.g., by setting some object fields to $\semnil$) one or more paths
in the heap of $\statesym$; then, $\statesym \in
\gammafrc(\abselemfrc)$ implies $\statesym^- \in
\gammafrc(\abselemfrc)$.  This choice is sound but could lead, in some
cases, to losing precision.  A possible improvement (already discussed
by \cite{GenaimZ13} for \REACHBASED) would be to distinguish cases
where the class of $v$ only declares the field $\ffield$; in this
case, to update $\ffield$ would be guaranteed to break all existing
paths starting from $\statef{\statesym}(v)$, and this information
could be used in order to set all $\abselemrc(v,w)$ to $\false$ for
every $w$ before computing the final $\abselemfrc^*$.

\subsubsection*{Field update, cyclicity}

As for cyclicity, the initial information $\abselemfc(w)$ about a
variable $w$ is updated only if $w$ can possibly reach $v$.  In fact,
new cycles reachable from $w$ can appear only if (1) an existing cycle
is made reachable from $v$ by the field update; or (2) a new cycle is
created which touches both $\statef{\statesym}(v)$ and
$\statef{\statesym}(v')$.  In both cases, the cycle under study will
not be reachable from $w$ unless $v$ is reachable from $w$: (1) in the
first case, because the new path from $w$ to the cycle has to touch
$\statef{\statesym}(v)$; and (2) in the second case, because a
variable reaching a cycle must reach every single heap location
belonging to the cycle itself.

\begin{example}
  Consider the following heap before the field update (dashed line):
  the cycle which will be created will not be reachable from $w$, even
  though it deep-shares with (but does not reach) $v$.
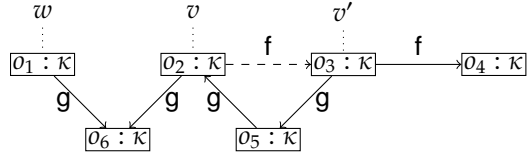
  \begin{center}
    \begin{tikzpicture}
      \tikzstyle{objnode} = [rectangle,draw,inner sep=1pt]
      \node (w) at (0,0) {$w$};
      \node (v) at (2,0) {$v$};
      \node (vp) at (4,0) {$v'$};
      \node[objnode] (o1) at (0,-0.7) {$o_1:\class$};
      \node[objnode] (o2) at (2,-0.7) {$o_2:\class$};
      \node[objnode] (o3) at (4,-0.7) {$o_3:\class$};
      \node[objnode] (o4) at (6,-0.7) {$o_4:\class$};
      \node[objnode] (o5) at (3,-1.7) {$o_5:\class$};
      \node[objnode] (o6) at (1,-1.7) {$o_6:\class$};
      \draw[dotted] (w) -- (o1);
      \draw[dotted] (v) -- (o2);
      \draw[dotted] (vp) -- (o3);
      \draw[->] (o1) -- node[left] {$\gfield$} (o6);
      \draw[->] (o2) -- node[right] {$\gfield$} (o6);
      \draw[dashed,->] (o2) -- node[above] {$\ffield$} (o3);
      \draw[->] (o3) -- node[right] {$\gfield$} (o5);
      \draw[->] (o3) -- node[above] {$\ffield$} (o4);
      \draw[->] (o5) -- node[left] {$\gfield$} (o2);
    \end{tikzpicture}
  \end{center}
\end{example}

On the other hand, reachability from $w$ to $v$ implies that new
cycles can become reachable from $w$.  The path-formula
$\left(\abselemfr(v',v) \odot \onlyfieldsp{\ffield} \right) \vee
\abselemfc(v')$, where $v'$ plays the role of $\res$, accounts exactly
for both kinds of loops discussed above: (1) newly-created loops
p-satisfying $\abselemfr(v',v) \odot \onlyfieldsp{\ffield}$; and (2)
existing loops which were reachable from $v'$ and now are reachable
from $v$ and (by transitivity) from $w$.

\subsubsection*{Method call, reachability}

Consider a method call $v :=
v_0.\methodsig\mbox{\lstinline!(!}v_1..v_k\mbox{\lstinline!)!}$, and
two variables $w_1$ and $w_2$ which are visible from the caller.  Let
the interest be on the reachability from $w_1$ to $w_2$ after
executing $\methodsig$.  Let also $M_i$ be the portion of the heap
which is reachable from some $v_i$: it is clear that nothing will be
modified by $\methodsig$ outside $M = M_0 \cup .. \cup M_k$.

The path-formula describing paths from $w_1$ to $w_2$ which are
possibly created inside $\methodsig$ depends on the position of
$\ell_1 = \statef{\statesym}(w_1)$ and $\ell_2 =
\statef{\statesym}(w_2)$ w.r.t.~$v_0 .. v_k$ and $M_0 .. M_k$.  The
abstract semantics analyzes the position of $\ell_1$ and $\ell_2$ with
respect to each pair $(v_i, v_j)$ separately, and joins the results
together.  The definition of $\abselemfr^{ij}(w_1,w_2)$ in Figure
\ref{fig:abs-sem-meth}), consists of five cases, named (a), (b), (c),
(d), and (e) in the following.  If $\ell_2$ does not belong to $M_j$,
then no path can be created which reaches it, and the abstract
semantics correctly returns $\false$ (case (e)).  Also, $v_i$ must not
be pure, otherwise there is no way to create new paths in
$\methodsig$.  Cases (a), (b), (c), and (d) capture the different
scenarios w.r.t.~$\ell_1$, $v_i$ and $v_j$.

\begin{itemize}
\item If $\ell_1$ reaches $v_i$, but does not deep-share with it, and
  $v_i$ does not deep-share with $v_j$ after the call, then the
  information is completely captured by reachability abstract values,
  i.e., no field-sensitive information is lost.  The formula returned
  in case (a) correctly represent the reachability between $w_1$ and
  $w_2$ since new paths will touch $\statef{\statesym}(v_i)$ (there
  can be new paths not touching $\statef{\statesym}(v_i)$, but they
  will be captured by some different $\abselemfr^{i'j'}$).  Note that
  a new path does not need to touch any other parameter than $v_i$; in
  this case, it will be accounted for, at least, when $j$ coincides
  with $i$.  Note also that $\false$ is the absorbing element for
  $\odot$, so that non-reachability between $w_1$ and $v_i$, or
  between $v_i$ and $v_j$, or between $v_j$ and $w_2$, implies that no
  path is actually created.
\item If $\ell_1$ reaches $v_i$, but does not deep-share with it, and
  $v_i$ may deep-share with $v_j$, then some field-sensitive
  information may be lost after $\methodsig$.  This is reflected by
  case (b), where the second part of the formula is set to $\true$.
  The first part is still the reachability from $w_1$ to $v_i$.
\item If $\ell_1$ deep-shares with $v_i$, and $v_i$ does not
  deep-share with $v_j$, then, again, some field-sensitive information
  may be lost after $\methodsig$.  This is accounted for by case (c),
  where the first part of the formula is $\true$.  On the other hand,
  the second part is the reachability from $v_j$ to $w_2$ since it is
  possible to guarantee that the last part of the new path will follow
  reachability paths between these two variables.
\item Case (d) is trivially sound since $\true$ is returned.
\end{itemize}

\subsubsection*{Method call, cyclicity}

Consider, again, the method call $v :=
v_0.\methodsig\mbox{\lstinline!(!}v_1..v_k\mbox{\lstinline!)!}$ and a
variable $w$ of the caller.  Suppose the interest is on the cyclicity
of $w$ after executing $\methodsig$.  Again, the portion $M = M_0 \cup
..\cup M_k$ of the heap which can be affected by the execution of
$\methodsig$ is the set of locations which are reachable from some of
the $v_i$ in zero or more steps.  The location $\ell =
\statef{\statesym}(w)$ can fall into one of the following cases:
\begin{itemize}
\item $\ell \notin M_i$ and no location of $M_i$ is reachable from
  $\ell$; in this case, no new cycles will be reachable from $w$, and
  there is nothing to prove.
\item $\ell \notin M_i$ and some $v_i$ is reachable from it; in this
  case, suppose that a new cycle is created in the heap by the
  execution of $\methodsig$, which is reachable from $v_i$: soundness
  requires that the cycle has to be also reachable from $w$ after the
  call.  This is satisfied by the abstract semantics since:
  \begin{itemize}
  \item $v_i$ is impure (otherwise, no new cycle can be created);
    and
  \item there is reachability from $w$ to $v_i$, by hypothesis.
  \end{itemize}
  Under these circumstances, the path-formula returned for the
  cyclicity of $w$ will be $\abselemfc''(v_i)$, which is correct since
  every cycle reachable from $v_i$ will be also inferred to be
  reachable from $w$.
\item $\ell \in M_i$: in this case, $\ell$ is reachable from $v_i$
  and, again, the path-formula $\abselemfc''(v_i)$ is a correct
  description of cycles reachable from $w$.  In fact,
  $\abselemfc^i(w)$ is not greater than $\abselemfc''(v_i)$ since the
  portion of the heap which is reachable from $w$ is a subset of the
  portion reachable from $v_i$, so that the same path-formula is
  certainly satisfied; on the other hand, it cannot be smaller
  because, in general, it is not possible to identify any cycle
  reachable from $v_i$ which is not reachable from $w$.
\item $\ell \notin M_i$ and $\SHARE{w}{v_i}$; in this case, the first
  location reachable from both happens to be a \emph{cutpoint}
  \cite{RinetzkyBRSW05}.  Anyway, cycles reachable from $w$ are still
  the same cycles reachable from $v_i$ (this is different
  w.r.t.~reachability because path-formul\ae~for cyclicity ignore the
  acyclic part of the path).
\end{itemize}

\subsection{Back to the cyclic-tree example}
\label{sec:backExample1}

Consider the code of Section \ref{sec:anExample1}.  Suppose that the
input parameters of \lstinline!join! are represented by the abstract
values depicted on the right-hand side of the figure, probably
inferred by previous steps of the analysis.  Different input abstract
values would yield comparable results.

\noindent
\begin{minipage}{75mm}
  \begin{lstlisting}[firstnumber=6]
  Tree join(Tree l, Tree r) {
    Tree t;  t := new Tree;
    t.left := l;
    t.right := r;
    if (l!=null) then l.parent := t;
    if (r!=null) then r.parent := t;
    return t;
  }
  \end{lstlisting}
\end{minipage}
\begin{minipage}{45mm}
  \begin{tabular}{l@{~}l@{~}l@{~}l@{~}l}
    $\abselemfr(\varl,\varl)$ & $=$ & $\abselemfc(\varl)$ & $=$ &
    $\pformulaempty$ \\
    $\abselemfr(\varr,\varr)$ & $=$ & $\abselemfc(\varr)$ & $=$ &
    $\pformulaempty$ \\
    $\abselemfr(\_,\_)$ & $=$ & $\abselemfc(\_)$ & $=$ & $\false$
    ~~ elsewhere
  \end{tabular}
\end{minipage}

\noindent
This information indicates that both parameters, which will be the
sub-trees of the tree created by \lstinline!join!, represent trees of
at most one node (they could be also $\nil$ pointers representing
$0$-node trees, since the previous analysis could have been overly
conservative).  Next tables show the result of the analysis.
$\abselemfrc^{i} = (\abselemfr^{i},\abselemfc^{i})$ is the abstract
value computed after line $i$, while \pleft, \pright~and
\pparent~stand, resp., for, $\fproposition{\mathsf{left}}$,
$\fproposition{\mathsf{right}}$, and $\fproposition{\mathsf{parent}}$.
\[ \begin{array}{|c|c|c|c|c|c|c|c|c|c|}
  \hline i & \abselemfr^i(\varl,\varl) & \abselemfr^i(\varl,\varr) &
  \abselemfr^i(\varl,\vart) & \abselemfr^i(\varr,\varl) &
  \abselemfr^i(\varr,\varr) & \abselemfr^i(\varr,\vart) &
  \abselemfr^i(\vart,\varl) & \abselemfr^i(\vart,\varr) &
  \abselemfr^i(\vart,\vart) \\\hline
  6 & \pformulaempty & \false & \false & \false & \pformulaempty &
  \false & \false & \false & \false \\\hline
  7 & \pformulaempty & \false & \false & \false & \pformulaempty &
  \false & \false & \false & \pformulaempty \\\hline
  8 & \pformulaempty & \false & \false & \false & \pformulaempty &
  \false & \pleft{\wedge}\lnot \pright{\wedge}\lnot \pparent & \false
  & \pformulaempty \\\hline
  9 & \pformulaempty & \false & \false & \false & \pformulaempty &
  \false & \pleft{\wedge}\lnot \pright{\wedge}\lnot \pparent &
  \pright{\wedge}\lnot \pleft{\wedge}\lnot \pparent & \pformulaempty
  \\\hline
  10 & \pformulaempty{\vee} & \false & \pparent{\wedge}\lnot\pright &
  \false & \pformulaempty & \false & \pleft{\wedge}\lnot\pright &
  \pright{\wedge}\lnot \pleft{\wedge}\lnot \pparent &
  \pformulaempty{\vee} \\
  & (\pparent{\wedge}\pleft{\wedge}\lnot \pright) & & & & & & & &
  (\pparent{\wedge}\pleft{\wedge}\lnot \pright) \\\hline
  11 & \pformulaempty{\vee}(\pparent{\wedge}\pleft) &
  \pparent{\wedge}\pright & \pparent & \pparent{\wedge}\pleft &
  \pformulaempty{\vee} & \pparent &
  \pleft{\wedge}(\pright\rightarrow\pparent) &
  \pright{\wedge}(\pleft\rightarrow\pparent) &
  \pformulaempty{\vee} \\
  & & & & & (\pparent{\wedge}\pright) & & & &
  (\pparent{\wedge}(\pleft{\vee}\pright))
  \\\hline
\end{array} \]

\[ \begin{array}{|c|c|c|c|c|c|c|c|c|c|c|c|c|}
  \hline i & \abselemfc^i(\varl) & \abselemfc^i(\varr) &
  \abselemfc^i(\vart) \\\hline
  6 & \pformulaempty & \pformulaempty & \false \\\hline
  7 & \pformulaempty & \pformulaempty & \pformulaempty \\\hline
  8 & \pformulaempty & \pformulaempty & \pformulaempty \\\hline
  9 & \pformulaempty & \pformulaempty & \pformulaempty \\\hline
  10 & \pformulaempty{\vee}(\pparent{\wedge}\pleft{\wedge}\lnot
  \pright) & \pformulaempty &
  \pformulaempty{\vee}(\pparent{\wedge}\pleft{\wedge}\lnot \pright)
  \\\hline
  11 & \pformulaempty{\vee}(\pparent{\wedge}\pleft) &
  \pformulaempty{\vee}(\pparent{\wedge}\pright) &
  \pformulaempty{\vee}(\pparent{\wedge}(\pleft{\vee}\pright))
  \\\hline
\end{array} \]

\noindent
The tables give an idea of how abstract values are obtained; however,
line 11 alone is enough to appreciate the results.

The final value for $\abselemfc(\varl)$ indicates that a cycle
starting at $\statef{\statesym}(\varl)$ is either an empty path or
one involving both \pparent and \pleft.  In fact, to traverse \pparent
and \pleft is needed to reach $\statef{\statesym}(\vart)$ and go back
to $\statef{\statesym}(\varl)$.  There are also cycles which touch
$\statef{\statesym}(\varr)$ and traverse all fields at least once;
this is also taken into account since the truth assignment $\{
\mathsf{left}, \mathsf{right}, \mathsf{parent} \}$ is still a model of
$\abselemfc(\varl)$.  A similar reasoning holds for
$\abselemfc(\varr)$.

Finally, $\abselemfc(\vart)$ represents all kinds of cycles starting
from $\statef{\statesym}(\vart)$, each of them corresponding to a
model of the path-formula: (1) an empty one; (2) one that reaches
$\statef{\statesym}(\varl)$ at least once, without reaching
$\statef{\statesym}(\varr)$; (3) a dual one which only reaches
$\statef{\statesym}(\varr)$; and (4) one that reaches both
$\statef{\statesym}(\varl)$ and $\statef{\statesym}(\varr)$ at least
once.

It is important to point out that it is not possible to precisely
describe the heap structure by simply inferring a set of fields which
\emph{have} to be traversed by every cycle
\cite{DBLP:conf/cav/BrockschmidtMOG12}.  In fact, $\mathsf{parent}$
is the only field which is involved in every cycle, but this
information alone would not be enough to guarantee that a loop going
from a leaf to the root terminates.

\subsection{Back to the double-linked-list example}
\label{sec:backExample2}

Consider again the code discussed in Section \ref{sec:anExample2}.
The following annotated code shows the deep-sharing information after
each line upon reaching the fixpoint.

\noindent
\begin{minipage}{\textwidth}
  \begin{lstlisting}
  i := 1;
  tmp := new Node;
  while (i<10) {
    x := new Node;      // $\SHARE{\tmp}{\tmp}$
    x.n := tmp;         // $\SHARE{\xx}{\xx}$,$\SHARE{\xx}{\tmp}$,$\SHARE{\tmp}{\tmp}$
    tmp.p := x;         // $\SHARE{\xx}{\xx}$,$\SHARE{\xx}{\tmp}$,$\SHARE{\tmp}{\tmp}$
    tmp := x;           // $\SHARE{\xx}{\xx}$,$\SHARE{\xx}{\tmp}$,$\SHARE{\tmp}{\tmp}$
    i := i+1;           // $\SHARE{\xx}{\xx}$,$\SHARE{\xx}{\tmp}$,$\SHARE{\tmp}{\tmp}$
  }                     // $\SHARE{\xx}{\xx}$,$\SHARE{\xx}{\tmp}$,$\SHARE{\tmp}{\tmp}$
  \end{lstlisting}
\end{minipage}

\noindent
Next table shows how the present analysis works on this example.
Again, the first column refers to the program line; primed numbers
correspond to the second time a given line is considered;
$\abselemfrc^{i} = (\abselemfr^{i},\abselemfc^{i})$ is the abstract
value computed after line $i$.
\[ \begin{array}{|c@{}|c|c|c|c|c|c|}
  \hline i & \abselemfr^i(\xx,\xx) & \abselemfr^i(\xx,\tmp) &
  \abselemfr^i(\tmp,\xx) & \abselemfr^i(\tmp,\tmp) & \abselemfc^i(\xx)
  & \abselemfc^i(\tmp) \\\hline
  1 & \false & \false & \false & \false & \false & \false \\\hline
  2,3 & \false & \false & \false & \pformulaempty & \false &
  \pformulaempty \\\hline
  4 & \pformulaempty & \false & \false & \pformulaempty &
  \pformulaempty & \pformulaempty \\\hline
  5 & \pformulaempty & \onlyfields{\pnext} & \false &
  \pformulaempty & \pformulaempty & \pformulaempty \\\hline
  6 & \pformulaempty {\vee} \onlyfields{\pnext,\pprev} & \pnext & \pprev &
  \pformulaempty {\vee} \onlyfields{\pnext,\pprev} & \pformulaempty {\vee}
  \onlyfields{\pnext,\pprev} & \pformulaempty {\vee} \onlyfields{\pnext,\pprev}
  \\\hline
  7,8 & \pformulaempty {\vee} \onlyfields{\pnext,\pprev} & \pformulaempty
  {\vee} \onlyfields{\pnext,\pprev} & \pformulaempty {\vee}
  \onlyfields{\pnext,\pprev} & \pformulaempty {\vee} \onlyfields{\pnext,\pprev}
  & \pformulaempty {\vee} \onlyfields{\pnext,\pprev} & \pformulaempty {\vee}
  \onlyfields{\pnext,\pprev} \\\hline
  3' & \pformulaempty {\vee} \onlyfields{\pnext,\pprev} & \pformulaempty
  {\vee} \onlyfields{\pnext,\pprev} & \pformulaempty {\vee}
  \onlyfields{\pnext,\pprev} & \pformulaempty {\vee} \onlyfields{\pnext,\pprev}
  & \pformulaempty {\vee} \onlyfields{\pnext,\pprev} & \pformulaempty {\vee}
  \onlyfields{\pnext,\pprev} \\\hline
  4' & \pformulaempty & \false & \false & \pformulaempty {\vee}
  \onlyfields{\pnext,\pprev} & \pformulaempty & \pformulaempty {\vee}
  \onlyfields{\pnext,\pprev} \\\hline
  5' & \pformulaempty & \pnext & \false & \pformulaempty {\vee}
  \onlyfields{\pnext,\pprev} & \pformulaempty {\vee}
  \onlyfields{\pnext,\pprev} & \pformulaempty {\vee}
  \onlyfields{\pnext,\pprev} \\\hline
  6' & \pformulaempty {\vee} \onlyfields{\pnext,\pprev} & \pnext & \pprev &
  \pformulaempty {\vee} \onlyfields{\pnext,\pprev} & \pformulaempty {\vee}
  \onlyfields{\pnext,\pprev} & \pformulaempty {\vee} \onlyfields{\pnext,\pprev}
  \\\hline
  7',8' & \pformulaempty {\vee} \onlyfields{\pnext,\pprev} & \pformulaempty
  {\vee} \onlyfields{\pnext,\pprev} & \pformulaempty {\vee}
  \onlyfields{\pnext,\pprev} & \pformulaempty {\vee} \onlyfields{\pnext,\pprev}
  & \pformulaempty {\vee} \onlyfields{\pnext,\pprev} & \pformulaempty {\vee}
  \onlyfields{\pnext,\pprev} \\\hline
\end{array} \]

\noindent
The first reachability information is added at line $2$: the creation
of an object implies that $\tmp$ is reachable from itself and cyclic
by means of an empty path.  The same happens at line 4 with $\xx$.
The update of \lstinline!x.n! implies that \tmp is reachable from \xx,
and it can be guaranteed that the new path will only traverse
\lstinline!n!.  In the following line, a cycle is created by the field
update of \lstinline!tmp.p!.  At this point, both variables are
cyclic, and it is possible to guarantee that a cycle will either (1)
be empty; or (2) traverse both fields.  The value for
$\abselemfr^6(\tmp,\xx)$ is computed as the disjunction between
$\abselemfr^5(\tmp,\xx) = \false$ and the new formula $\pformulaempty
\odot (\onlyfieldsp{\pp} \vee (\onlyfieldsp{\pp} \odot
\abselemfr^5(\xx,\tmp))) \odot \pformulaempty = \pprev$, giving
$\pprev$ as the final result.  Note that the disjunction between
$\onlyfieldsp{\pp}$ and $\onlyfieldsp{\pp} \odot
\abselemfr^5(\xx,\tmp)$ takes into account both ways to go from \tmp
to \xx: either directly, or traversing the new cycle and touching \xx
twice.  Therefore, it is not possible to guarantee that no path will
traverse $\nn$.  Line $7$ assigns \xx to \lstinline!tmp!, so that the
reachability information $\abselemr^7(\xx,\tmp)$ and
$\abselemr^7(\tmp,\xx)$ is copied from $\abselemr^6(\xx,\xx)$.  The
second iteration is needed since line $3$ is entered the second time
with a different abstract value (row $3'$ of the table).  At row $4'$,
reachability and cyclicity information about \xx is removed, but
non-empty cyclicity of \tmp is still admitted.  A fixpoint is reached
after the second iteration, as shown by the fact that row $8$ and row
$8'$ of the table are equal.  The final result, which is finally
copied to $\abselemfrc^9$, correctly shows that any non-empty cycle
reachable from either $\xx$ or $\tmp$ has to involve both $\nn$ and
$\pp$, so that a loop which only traverses one of these fields (as the
one depicted in Section \ref{sec:anExample2}, lines 10--12) is
guaranteed to terminate.

\section{Practical issues}
\label{sec:practicalIssues}

The analysis has been partially implemented\footnote{Available at
  \texttt{http://costa.ls.fi.upm.es/\texttildelow
    damiano/reachCycle/}.} in its \emph{intra-procedural} part, based
on the Chord Java bytecode analyzer \cite{Naik11}.  The implementation
covers most sequential Java bytecode instructions which may occur in
the single method the analysis focuses on.  The examples of Sections
\ref{sec:anExample1} and \ref{sec:anExample2} have been analyzed, and
the result appears in the tables of Sections \ref{sec:backExample1}
and \ref{sec:backExample2}.

A path-formula is explicitly represented as the set of its models,
which are, in turn, sets of fields: for example, $\{ \{ \field_1 \},
\{ \field_2 \} \}$ represents the exclusive disjunction
$(\fproposition{\field_1}{\wedge}\lnot\fproposition{\field_2}) \vee
(\fproposition{\field_2}{\wedge}\lnot\fproposition{\field_1})$ if
$\field_1$ and $\field_2$ are the only fields in the program.  During
the fixpoint computation, path-formul\ae~are always combined by
disjunction, which amounts to add new sets of fields (models) to their
representation.

\begin{example}[(implementation of $\odot$)]
  \label{ex:odotImplementation}
  Consider the field update $v.\field\assign v'$.  Let $\pformula_1$
  and $\pformula_2$ be the reachability information from some $w_1$ to
  $v$ and from $v'$ to some $w_2$, respectively.  Both formul\ae~are
  represented as sets of sets of fields:
  \[ \pformula_i =
  \left\{~~\left\{~\field^i_{11},~..~,~\field^i_{1k_1}~\right\},~~..~~,
  \left\{~\field^i_{n1},~..~,~\field^i_{1k_n}~\right\}~~\right\} \]
  \noindent where each truth assignment $\{
  \field^i_{j1},..,\field^i_{jk_j} \}$ is a model of $\pformula_i$.
  The newly-computed formula
  \[\abselemfr''(w_1,w_2) =
  \abselemfr'(w_1,v) \odot
  \left(\onlyfieldsp{\field}{\vee}\left(\onlyfieldsp{\field}\odot\abselemfr'(v',v)\right)\right)
  \odot \abselemfr'(v',w_2)\] can be obtained by taking each pair of
  sets $(X_1,X_2)$ belonging to $\pformula_1 \times \pformula_2$, and
  each $Y$ belonging to (the set representation of) $\pformulag =
  \abselemfr'(v',v)$, and computing the unions $X_1 \cup \{ \field \}
  \cup X_2$ and $X_1 \cup \{ \field \} \cup Y \cup X_2$.  More
  formally,
  \[ \abselemfr''(w_1,w_2)~~=~~\{~X_1 \cup \{ \field \} \cup X_2~|~X_i
  \in \pformula_i~\}~~\cup~~\{~X_1 \cup \{ \field \} \cup Y \cup X_2~|~X_i \in
  \pformula_i, Y \in \pformulag~\} \]
\end{example}

\noindent
The global fixpoint terminates when no new models are added to
path-formul\ae.

\subsection{Scalability}

With respect to similar analyses discussed as related work, the main
threats to scalability seem to be the complexity of operations on
path-formul\ae~and the potentially large number of reference fields in
a program.  The implementation deals with path-formul\ae~operations
such as $\odot$ in the way suggested in Example
\ref{ex:odotImplementation}.  A set of fields can be easily
represented as a \emph{bit vector}, so that union between sets of
fields can be efficiently obtained by standard bitwise operations.  A
large number of fields would imply the need of more memory to store
bit vectors, but no significant slowdown in bitwise computations.

Moreover, to have a large number of fields to be tracked does not
necessarily mean that operations like the one described in the example
above have to deal with a huge number of field sets $X_i$ and $Y$:
there is no reason to have path-formul\ae~with more models if the
program has, globally, more fields.  Suppose, for example, that the
code presented in Section \ref{sec:anExample2} is part of a bigger
program with many class declarations: the path-formula computed, say,
for $\abselemfc^8(\tmp) = \pformulaempty \vee (\pnext{\wedge}\pprev)$
would have the same two models $\{\}$ and $\{ \mathsf{n}, \mathsf{p}
\}$, the difference being that the bit vectors representing them would
be much longer vectors with almost all bits set to 0.

However, there are at least two ways to go in the direction of
improving efficiency; both are based on the observation that the
analysis of data structures in the heap can benefit from knowledge of
what the information to be inferred will be used for.  In the case of
proving termination of loops traversing such data structures, which is
one of the main indirect goals of the present paper, to know in
advance \emph{which loops} will possibly traverse the data structure
under study is a valuable piece of information which can significantly
improve efficiency.  The next two paragraphs discuss two improvements
relying on the example of Section \ref{sec:anExample1}: the goal is to
gather information on the structure of the tree built by

\centerline{\lstinline!x:=join(tree1,tree2)!}

\noindent
knowing in advance that it will be traversed by the loop

\centerline{\lstinline^while (x!=null) x:=x.left^}

\noindent
and that the final goal is to prove termination of such a loop.

\subsubsection*{Backward program slicing}

In order to reduce the portion of code to be analyzed, \emph{backward
  static slicing} \cite{XuQZWC05} can be used.  A slice can be
computed backwards from the value of \lstinline!x! at the program
point before the loop, which includes the part of the program
potentially affecting the value of the variable \lstinline!x! at that
point.  Backward static slicing is a well-known technique which has
been applied to programming languages like Java and has reached a
considerable level of maturity.

\subsubsection*{Field abstraction}

The problem of dealing with a large number of reference fields can be
alleviated by observing that some fields are not really relevant when
focusing on specific parts of the code.  Consider the example
presented above: if it is known in advance that the only loop
traversing the tree pointed to by \lstinline!x! will do it on the
\lstinline!left! field, then the only question to be answered is: ``is
$\{ \mbox{\lstinline!left!} \}$ a model of
$\abselemfc(\mbox{\lstinline!x!})$?''; in other words, is there the
possibility of a cycle which only traverses \lstinline!left!?  This
question can be given an answer without knowing exactly which fields
are involved in cycles and under which circumstances, so that a
simplification can be performed on the implementation of the analysis.

If \lstinline!left! is the only field to be tracked explicitly, as in
this case, then all the other fields can be abstracted to a special
field $\any$, so that the analysis only has to take \lstinline!left!
and $\any$ into account instead of all reference fields in the
program.  Operations on field sets are also simplified: e.g., $\{
\any, \field_1 \} \cup \{ \any, \field_2 \} = \{ \any, \field_1,
\field_2 \}$ where $\field_1$ and $\field_2$ are the fields to be
tracked explicitly.  On the other hand, $\{ \any, \field_1 \}
\setminus \{ \field_2 \}$ results in both $\{ \any, \field_1 \}$ and
$\{ \field_1 \}$ since $\any$ could be representing, at that point,
exactly $\field_2$ alone.  In the example above, the models of
$\abselemfc(\mbox{\lstinline!x!})$ come to be $\{\}$, $\{ \any \}$,
and $\{ \any, \mbox{\lstinline!left!} \}$, thus making it possible to
detect that no cycle will only traverse \lstinline!left!.  This
feature has been added to the implementation described in this
section: it is possible to specify manually which fields have to be
tracked explicitly; if not specified, then all fields are tracked.


\section{Conclusions}
\label{sec:conclusions}

The present paper describes a novel approach to cyclicity analysis
which is able to provide, even for \emph{possibly cyclic} data
structures, information that is useful to prove termination or bounds
on the resource usage of programs traversing them.  This is
accomplished by considering the \emph{fields} through which a cyclic
path goes through.  If, for example, it is possible to prove that any
cycle has to involve certain fields, then the result of the analysis
can be successfully used to prove the termination of a loop which
never traverses any cycle completely.

A typical example is a tree with edges to parent nodes: this is
clearly a cyclic data structure, but loops traversing it
\emph{one-way} (e.g., from the root to a leaf, or the other way
around) terminate.  Existing cyclicity analyses cannot give enough
information about this example, so that termination of loops which
traverse the data structure cannot be guaranteed.  On the contrary,
this analysis provides the required field-sensitive information in
form of propositional formul\ae, which are expressive enough to
capture relevant properties of cycles.

Future work will be mainly devoted to complete the prototypical
implementation discussed in Section \ref{sec:practicalIssues}.  This
will involve adding more features like static fields, making the
analysis interprocedural, and build a user interface.  Moreover, the
precision of the abstract semantics and the implementation could be
improved, especially on method calls, by a field-sensitive version of
the deep-sharing analysis of Section \ref{sec:auxiliaryAnalyses}.
More speculative work will look for further insight into the abstract
domains mentioned in Section \ref{sec:moreExpressive}.

\appendix


\section{Proofs}
\label{sec:proofs}

\newcommand{\mylemma}[2]{\textsc{Lemma \ref{#1}}. \emph{#2}}

\mylemma{lemma:decideViability}{The viability of a truth assignment
  $\finterpretation$ is decidable.}

\begin{proof}
  The goal is to find a state where there is a path $\hpath$ which
  traverses all and only the fields whose f-proposition belongs to
  $\finterpretation$, i.e., $\hpath$ should traverse all and only the
  fields in $\varphi = \{ \field | \fproposition{\field} \in
  \finterpretation \}$.

  The first step is to compute, for every class $\class \in \classes$,
  the set of types of objects which can be reached from objects of
  type $\class$ \emph{by only traversing fields in $\varphi$},
  i.e., the reflexive and transitive closure $R^*$ of the relation
  \[ R = \{ (\class',\class'')~|~\class'\in \classes \wedge \exists
  \class'.\field \in \varphi.~\class'.\field~\mbox{has declared
    type}~\class,~\mbox{with}~\class'' \subclasseq \class \} \] Note
  that paths of length 0 are also considered.  The closure is
  computable since $\classes \times \classes$ is finite and $R^*
  \subseteq \classes \times \classes$, so it is guaranteed that a
  fixpoint will be reached in a finite number of steps.  In the end,
  $(\class',\class'') \in R^*$ means that it is possible to reach a
  $\class''$ object from a $\class'$ object in zero or more steps, by
  only traversing $\varphi$.

  The second step is to consider, one at a time, all the permutations
  of $\varphi$ (which, as a subset of $\fields$, is finite).  For each
  permutation $\langle \class_{p_1}.\field_{p_1}, .. ,
  \class_{p_n}.\field_{p_n} \rangle$, a path $\hpath$ is searched for,
  which has the following form: it traverses
  $\class_{p_1}.\field_{p_1}$, then goes through $\varphi$ from the
  second location to a location where an object of class
  $\class_{p_2}$ is stored, then traverses
  $\class_{p_2}.\field_{p_2}$, then goes through $\varphi$ from the
  last-obtained location to a location where an object of class
  $\class_{p_3}$ is stored, and so on until it traverses
  $\class_{p_n}.\field_{p_n}$.  Such a path would have the desired
  property of traversing all and only the fields in $\varphi$,
  regardless of how many times every single field is traversed.

  The existence of $\hpath$ depends on $R^*$.  Let $\ell'_i$ and
  $\ell''_i$ be, resp., the locations where the objects \emph{before}
  and \emph{after} traversing $\class_{p_i}.\field_{p_i}$ are stored;
  the object stored at $\ell'_i$ would have type $\class_{p_i}$, while
  the object stored at $\ell''_i$ would have some type compatible with
  class declarations.  The possibility to connect some $\ell''_i$ to
  $\ell'_{i+1}$ (i.e., to \emph{fill the gap} between consecutive
  fields in the permutation) depends on whether there is some
  $(\class_i^? ,\class_{p_{i+1}}) \in R^*$ s.t.~$\class_i^?$ is a
  subclass of the declared type of $\class_{p_i}.\field_{p_i}$.  If
  such a pair exists, then a heap can be picked where the object at
  $\ell''_i$ has type $\class_i^?$.

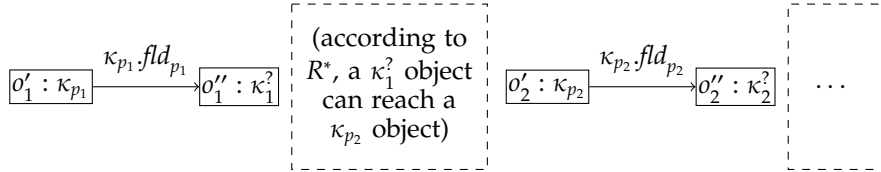
  \begin{center}
    \begin{tikzpicture}
      \tikzstyle{objnode} = [rectangle,draw,inner sep=1pt]
      \node[objnode] (l1) at (0,0) {$o'_1:\class_{p_1}$};
      \node[objnode] (l1a) at (2.5,0) {$o''_1:\class_1^?$};
      \draw[->] (l1) -- node[above] {$\class_{p_1}.\field_{p_1}$} (l1a);
      
      \node[text width=2.5cm, text centered] at (4.5,0) {(according to $R^*$, a
        $\class_1^?$ object can reach a $\class_{p_2}$ object)};
      \draw[dashed] (3.2,1.1) -- (5.8,1.1) -- (5.8,-1.1) -- (3.2,-1.1) -- (3.2,1.1);
      
      \node[objnode] (l2) at (6.6,0) {$o'_2:\class_{p_2}$};
      \node[objnode] (l2a) at (9.1,0) {$o''_2:\class_2^?$};
      \draw[->] (l2) -- node[above] {$\class_{p_2}.\field_{p_2}$} (l2a);
      
      \node at (10.4,0) {\dots};
      \draw[dashed] (11,1.1) -- (9.8,1.1) -- (9.8,-1.1) -- (11,-1.1);
    \end{tikzpicture}
  \end{center}

  \noindent
  If there is a permutation of $\varphi$ such that all the gaps can be
  filled, then such a $\hpath$ exists, which only traverses fields in
  $\varphi$, and traverses all of them; in this case,
  $\finterpretation$ is guaranteed to be viable.  On the other hand,
  if all the permutations have been considered but it was not possible
  to build $\hpath$ for any of them, then $\finterpretation$ is not
  viable.  It is easy to see that the whole process is computable.
\end{proof}

\mylemma{lemma:decideEquivalence}{The equivalence of two
  path-formul\ae~is decidable.}

\begin{proof}
  Truth assignments are finite since $\fields$ is also finite.
  Therefore, it is enough to find a viable one which is a model of one
  path-formula and a counter-model of the other.  To decide if a truth
  assignment is a model of a formula is easy; moreover, its viability
  is decidable by Lemma \ref{lemma:decideViability}, so that the whole
  problem is decidable.
\end{proof}

\begin{lemma}
  \label{lemma:gammaMonotonicityUnicity}
  (This lemma is only presented in the appendix) $\alphafr$ and
  $\gammafr$ are monotone.  Moreover, for every $\abselemfr',
  \abselemfr''\in \domfr$, $\abselemfr' \neq \abselemfr''$ implies
  $\Rightarrow \gammafr(\abselemfr') \neq \gammafr(\abselemfr'')$.
\end{lemma}

\begin{proof}
  Suppose that there exists a pair of variables $(v,w)$ such that
  $\abselemfr'(v,w) \not\equiv \abselemfr''(v,w)$; this means that the
  corresponding path-formul\ae~do not belong to the same equivalence
  class, i.e., that there is a \emph{viable} $\finterpretation$ which
  is a model of one and only one of them.  Suppose $\finterpretation$
  is a model of $\abselemfr'(v,w)$ but not of $\abselemfr''(v,w)$ (the
  dual case is similar).  By the definition of the equivalence on
  $\pformulae_{\equiv^{v,w}}$, there must be a state $\statesym$ where
  there is a path $\hpath$ from $v$ to $w$ traversing all and only the
  fields belonging to $\{ \field | \fproposition{\field} \in
  \finterpretation \}$, so that
  $\traverses{\hpath}{\abselemfr'(v,w)}$, but
  $\ntraverses{\hpath}{\abselemfr''(v,w)}$.  Obviously, there can be
  many such states; w.l.o.g., $\statesym$ can be chosen among them,
  such that $\hpath$ is the only path in the heap (note that this is
  always possible since the other variables are not relevant, and
  there is no need for any other path from $v$ to $w$).
  Now, $\statesym$ clearly belongs to $\gammafr(\abselemfr')$, but not
  to $\gammafr(\abselemfr'')$ (because there is no path-formula
  $\pformula \leq \abselemfr''(v,w)$ such that $v$ $\pformula$-reaches
  $w$ in $\statesym$), and this concludes the proof.
\end{proof}

\mylemma{lemma:gi-reach}{$\alphafr$ and $\gammafr$ define a
  \emph{Galois insertion} between $\domfr$ and $\condom$.}

\begin{proof}
  We first prove $\concelem \subseteq \gammafr(\alphafr(\concelem))$,
  i.e., that $\statesym \in \concelem$ implies $\statesym \in
  \gammafr(\alphafr(\concelem))$.  Given a concrete state $\statesym
  \in \concelem$, the following holds for every $v$ and $w$:
  
  \centerline{$\pformulag_{v,w} = \bigwedge \{
    \pformula~|~v~\pformula\mbox{-reaches}~w~\mbox{in}~\statesym \} \leq
    (\alphafr(\concelem))(v,w)$}
  
  \noindent since a part of a disjunction always implies
  the disjunction itself.  Now, due to the behavior of path-formul\ae,
  $v$ $\pformulag_{v,w}$-reaches $w$ in $\statesym$, so that
  $\pformulag_{v,w}$ is the path-formula $\pformula$ required by
  $\gammafr$ in order to guarantee that $\statesym \in
  \gammafr(\alphafr(\concelem))$.

  The second part of the proof demonstrates that
  $\alphafr(\gammafr(\abselemfr)) \sqsubseteqfr \abselemfr$ holds for
  every $\abselemfr$.  This amounts to saying that

  \centerline{$\forall \abselemfr.~\forall
  v,w.~(\alphafr(\gammafr(\abselemfr)))(v,w) \leq \abselemfr(v,w)$}

  \noindent The goal is to prove that any path $\hpath$ p-satisfying
  $(\alphafr(\gammafr(\abselemfr)))(v,w)$ will also p-satisfy
  $\abselemfr(v,w)$, since $\leq$ is logical implication.  By the
  definition of $\alphafr$ and $\gammafr$:
  
  \centerline{
    $\begin{array}{rl}
      & \traverses{\hpath}{(\alphafr(\gammafr(\abselemfr)))(v,w)} \\
      \Leftrightarrow & \traverses{\hpath}{\bigvee \{ \bigwedge \{
        \pformula | v~\pformula\mbox{-reaches}~w~\mbox{in}~\statesym
        \} | \statesym
        \in \gammafr(\abselemfr) \}} \\
      \Leftrightarrow & \traverses{\hpath}{\bigvee \{ \bigwedge \{
        \pformula
        | v~\pformula\mbox{-reaches}~w~\mbox{in}~\statesym \} | \\
        & \qquad \qquad \statesym \in \states{\typenv} \wedge \forall
        v',w'. \exists \pformulag \leq
        \abselemfr(v',w').~v'~\pformulag\mbox{-reaches}~w'~\mbox{in}~\statesym
        \}}
    \end{array}$
  }
  \noindent For every $\statesym$ satisfying $(\forall v',w'. \exists
  \pformulag \leq
  \abselemfr(v',w').~v'~\pformulag\mbox{-reaches}~w'~\mbox{in}~\statesym)$,
  the instance $(\exists \pformulag \leq
  \abselemfr(v,w).~v~\pformulag\mbox{-reaches}~w~\mbox{in}~\statesym)$
  also holds, where $v$ and $w$ are the variables of interest.  This
  implies $\bigwedge \{ \pformula |
  v~\pformula\mbox{-reaches}~w~\mbox{in}~\statesym \} \leq
  \abselemfr(v,w)$ since the mentioned $\pformulag$ is actually one of
  the $\pformula$'s in the conjunction, so that $\bigwedge \{
  \pformula | v~\pformula\mbox{-reaches}~w~\mbox{in}~\statesym \} \leq
  \pformulag$, and $\pformulag \leq \abselemfr(v,w)$ by hypothesis.
  Since this holds for every $\statesym$ satisfying the property above
  (i.e., for all $\statesym \in \gammafr(\abselemfr)$), the
  disjunction on all such states required by $\alphafr$ is still less
  than or equal to $\abselemfr(v,w)$.  Therefore, since a path-formula
  
  \centerline{$\bigvee~\left\{~ \bigwedge~\left\{
        \pformula~|~\mbox{$v$ $\pformula$-reaches $w$ in}~\statesym
      \right\}~|~\statesym~\mbox{satisfies the condition above}
    \right\} $} 

  \noindent has been found such that $\hpath$ p-satisfies it, and is
  less than or equal to $\abselemfr(v,w)$, the statement
  $\traverses{\hpath}{\abselemfr(v,w)}$ holds.  Thus, the implication
  $\traverses{\hpath}{(\alphafr(\gammafr(\abselemfr)))(v,w)}
  \Rightarrow \traverses{\hpath}{\abselemfr(v,w)}$ has been proved,
  and this completes the second part of the proof.

  Joining both parts demonstrates that $\alphafr$ and $\gammafr$ build
  a Galois connection.  Lemma \ref{lemma:gammaMonotonicityUnicity}
  guarantees that the correspondence is indeed a Galois insertion.
\end{proof}

\mylemma{lemma:gi-cyc}{$\alphafc$ and $\gammafc$ define a \emph{Galois
    insertion} between $\domfc$ and $\condom$.}

\begin{proof}
  Very similar to Lemma \ref{lemma:gi-reach}.  Note that this proof
  requires to prove monotonicity (easy), and a unicity result for
  cyclicity similar to Lemma \ref{lemma:gammaMonotonicityUnicity},
  i.e., that $\abselemfc' \neq \abselemfc''$ implies
  $\gammafc(\abselemfc') \neq \gammafc(\abselemfc'')$; such a result
  is also easy to prove.
\end{proof}

\mylemma{lemma:reducedProduct}{
  The lattice based on
  $\{~(\abselemfr,\abselemfc)~|~\abselemfr{\in}\reachsetf{\typenv},
  \abselemfc{\in}\cycsetf{\typenv}, (\abselemfr,\abselemfc)~\mbox{is
    in normal form}~\}$, with concretization function
  $\gammafrc(\abselemfr,\abselemfc) = \gammafr(\abselemfr) \cap
  \gammafc(\abselemfc)$, is the reduced product of $\domfr$ and
  $\domfc$.}

\begin{proof}
  The goal is to prove that $\gammafrc$ is injective, i.e., that
  $\abselemfrc' \neq \abselemfrc''$ implies $\gammafrc(\abselemfrc')
  \neq \gammafrc(\abselemfrc'')$.
  If $\abselemfrc' \neq \abselemfrc''$, then either (a) $\abselemfr'
  \neq \abselemfr''$ or (b) $\abselemfc' \neq \abselemfc''$.
  \begin{itemize}
  \item[(a)] Suppose the abstract values differ on the reachability
    from $v$ to $w$.  By Lemma \ref{lemma:gammaMonotonicityUnicity},
    $\gammafr(\abselemfr')$ must be different from
    $\gammafr(\abselemfr'')$.  Suppose there are states $\statesym \in
    \gammafr(\abselemfr') \setminus \gammafr(\abselemfr'')$ (the dual
    case is similar).
    Without loss of generality, $\statesym$ can be taken as follows:
    \begin{itemize}
    \item its behavior on $(v,w)$ is as required by Lemma
      \ref{lemma:gammaMonotonicityUnicity}: there is a path from $v$
      to $w$ which p-satisfies $\abselemfr'(v,w)$ but not
      $\abselemfr''(v,w)$;
    \item no other variables reach each other;
    \item if $v$ and $w$ are different variables, then there are no
      cyclic variables;
    \item if $v$ and $w$ are the same variable, then the path from $v$
      to $w$ is actually a cyclic path, and it is the only cycle in
      the heap (note that this is allowed since $\abselemfc'(v) \geq
      \abselemfr'(v,v)$ is required in normal forms).
    \end{itemize}
    With this definition, $\statesym \in \gammafc(\abselemfc')$ and
    $\statesym \in \gammafr(\abselemfr')$, so that it belongs to
    $\gammafrc(\abselemfrc')$.  On the other hand, $\statesym \notin
    \gammafrc(\abselemfrc'')$ because it does not belong to
    $\gammafr(\abselemfr'')$.
  \item[(b)] Suppose the abstract values differ on the cyclicity of
    $v$.  There are results for cyclicity which are similar to their
    reachability counterpart, and imply $\gammafc(\abselemfc') \neq
    \gammafc(\abselemfc'')$.
    A state $\statesym$ belonging to the set difference
    $\gammafc(\abselemfc') \setminus \gammafc(\abselemfc'')$ can be
    taken as follows:
    \begin{itemize}
    \item its behavior on $v$ is such that there is one and only one
      cycle on $v$, and such a cycle p-satisfies $\abselemfc'(v)$ but
      not $\abselemfc''(v)$;
    \item no variables reach each other, not even $v$ reaches itself
      (this is possible since the unique cycle on $v$ can be supposed
      to start from some location reachable from $v$, not from
      $\statef{\statesym}(v)$ itself).
    \end{itemize}
    With this definition, $\statesym$ trivially belongs to both
    $\gammafr(\abselemfr')$ and $\gammafr(\abselemfr'')$; therefore,
    it does belong to $\gammafrc(\abselemfrc')$ although not to
    $\gammafrc(\abselemfrc'')$.
  \end{itemize}
\end{proof}

\mylemma{lemma:gi-withoutFields}{The abstract domain $\domr$ is an
  abstraction of $\domfr$.}

\begin{proof}
  In order to prove this result, an abstraction function $\alpha$ and
  a concretization function $\gamma$ have to be given, which satisfy
  the definition of Galois connection.  Let $\pformula_\vee$ be
  $\bigvee_{\field \in \fields} \fproposition{\field}$, i.e., the
  formula having all and only non-empty models.
  \[ \begin{array}{rcl} \alpha(\abselemfr) & = & \{~\REACHES{v}{w} \in
    X^\rightsquigarrow~|~\abselemfr(v,w) \wedge \pformula_\vee \neq
    \false~\} \\
    \gamma(\abselemr) & = & \lambda
    (v,w).~\left\{ \begin{array}{l@{\qquad}l}
        \true & \mbox{if}~\REACHES{v}{w} \in \abselemr \\
        \lnot \pformula_\vee & \mbox{otherwise}
      \end{array} \right.
  \end{array} \]
  $\alpha$ means that $\REACHES{v}{w}$ will be included in the
  abstraction whenever the path-formula $\abselemfr(v,w)$ has some
  non-empty model; this is because, unlike $\domfr$, $\domr$ only
  considers paths whose length is at least 1.
  On the other hand, $\gamma$ assigns $\true$ whenever $\REACHES{v}{w}
  \in \abselemr$ because $\domr$ does not track conditions on paths;
  on the other hand, the path-formula returned when $\REACHES{v}{w}
  \notin \abselemr$ still admits paths of length 0, as expected.
  
  The first thing to prove is that $\gamma(\alpha(\abselemfr))
  \sqsupseteqfr \abselemfr$; this follows easily since (1) if
  $\abselemfr(v,w)$ has no non-empty models, then it is either
  $\false$ or equivalent to $\lnot \pformula_\vee$; therefore,
  $(\gamma(\alpha(\abselemfr)))(v,w)$ is $\lnot \pformula_\vee$ which
  is $\geq \abselemfr(v,w)$.  On the other hand, if (2)
  $\abselemfr(v,w)$ has non-empty models, then
  $(\gamma(\alpha(\abselemfr)))(v,w)$ is $\true \geq \abselemfr(v,w)$.
  In both cases, $(\gamma(\alpha(\abselemfr)))(v,w) \geq
  \abselemfr(v,w)$, which proves the result.  Note that, at a first
  sight, if $\abselemfr(v,w) = \pformula$ has non-empty models but
  $\REACHES{v}{w} \notin X^\rightsquigarrow$, then the statement will
  not be added to $\alpha(\abselemfr)$, so that
  $(\gamma(\alpha(\abselemfr)))(v,w)$ will be $\lnot \pformula_\vee$,
  which seems to be smaller than $\pformula$.  However, such case
  never happens since $\REACHES{v}{w} \notin X^\rightsquigarrow$
  implies that there is no path from $v$ to $w$ of length $\geq 1$, so
  that either $\pformula \equiv^{v,w} \false$ or $\pformula
  \equiv^{v,w} \lnot \pformula_\vee$.
  
  The second part of the proof is to demonstrate that
  $\alpha(\gamma(\abselemr))$ is smaller than or equal to $\abselemr$.
  It is straightforward to see that equality holds, so that the
  correspondence is actually a Galois insertion.
\end{proof}

\mylemma{lemma:gi-aliasing}{ The abstract domain $\domfr$ is a
  refinement of the aliasing domain.}

\begin{proof}
  The aliasing domain $\domalias$ is defined as the lattice of sets of
  pairs $\ALIAS{v}{w}$ of variables, ordered by $\subseteq$, and by
  the following abstraction and concretization functions w.r.t.~the
  concrete domain:
  \[
  \begin{array}{rcl}
    \alpha_a(\concelem) & = & \{ \ALIAS{v}{w}~|~\exists \statesym \in
    \concelem.~\statef{\statesym}(v) = \statef{\statesym}(w) \} \\
    \gamma_a(\abselema) & = & \{ \statesym~|~\forall
    v,w.~\statef{\statesym}(v) = \statef{\statesym}(w) \Rightarrow
    \ALIAS{v}{w} \in \abselema \}
  \end{array}
  \]
  In order to prove the result, the definition of suitable abstraction
  and concretization functions between $\domalias$ and $\domfr$ is
  needed.
  \[
  \begin{array}{rcl}
    \alpha(\abselemfr) & = & \{ \ALIAS{v}{w}~|~\abselemfr(v,w) \wedge
    \pformulaempty \neq \false \} \\
    \gamma(\abselema) & = & \lambda v,w.~\left\{ \begin{array}{l@{\qquad}l}
        \true & \mbox{if}~\ALIAS{v}{w} \in \abselema \\
        \false & \mbox{otherwise}
      \end{array} \right.
  \end{array}
  \]
  As usual, it must me proved that $\gamma(\alpha(\abselemfr))
  \sqsupseteqfr \abselemfr$ and $\alpha(\gamma(\abselema))$ is smaller
  than or equal to $\abselema$.  The first part follows from observing
  that $\ALIAS{v}{w} \in \alpha(\abselemfr)$ whenever $\emptyset \in
  \fmodels{\abselemfr(v,w)}$.  In this case,
  $(\gamma(\alpha(\abselemfr)))(v,w)$ will be $\true \geq
  \abselemfr(v,w)$.  On the other hand, the equality between
  $\alpha(\gamma(\abselema))$ and $\abselema$ is straightforward.
\end{proof}

\mylemma{lemma:gc-onlyClasses}{$\domkr$ is an abstraction of $\domr$.}

\begin{proof}
  The abstraction and concretization function identifying a Galois
  insertion between $\domkr$ and $\domr$ are as follows (here,
  $\dectype(v)$ is the \emph{declared type} of $v$):
  \[ \begin{array}{rcl} \alpha(\abselemr) & = & \{
    (\dectype(v),\dectype(w))~|~\REACHES{v}{w} \in \abselemr \} \\
    \gamma(\abselemkr) & = & \{ \REACHES{v}{w}~|~(\class_1,\class_2)
    \in \abselemkr, \dectype(v) \subclasseq \class_1, \dectype(w)
    \subclasseq \class_2 \}
  \end{array} \]
  It is straightforward to see that $\gamma(\alpha(\abselemr))
  \supseteq \abselemr$ since (1) $(\dectype(v),\dectype(w)) \in
  \alpha(\abselemr)$ for every statement $\REACHES{v}{w} \in
  \abselemr$; (2) $\REACHES{v}{w} \in \gamma(\alpha(\abselemr))$ since
  the declared types of $v$ and $w$ are equal to the $\class_1$ and
  $\class_2$ required in the definition of $\gamma$.

  On the other hand, $\alpha(\gamma(\abselemkr)) = \abselemkr$ holds:
  in fact, applying $\gamma$ adds all statements $\REACHES{v}{w}$
  where $v$ and $w$ are \emph{any} variables with compatible types.
  Afterward, $\alpha(\gamma(\abselemkr))$ is computed as the set of
  pairs of declared types of all such variables.  The last set can be
  larger, as it could include more subclasses; however, it is
  equivalent to $\abselemkr$ according to the equivalence relation on
  $\wp(X^{\classes \times \classes})$.
\end{proof}

\mylemma{lemma:gc-monotonic}{
  The following abstraction and concretization functions define a
  Galois connection between $\domfr$ and $\dommr$: the latter strictly
  abstracts the former.
  \[ \begin{array}{rcl} (\alpha(\abselemfr))(v,w) & = &
    \left\{ \begin{array}{ll}
        \false & \mbox{if}~\abselemfr(v,w){\models} \false \\
        \true & \mbox{if}~\true{\models}\abselemfr(v,w) \\
        \bigwedge \{
        \fproposition{\field_1}\!{\vee}..{\vee}\fproposition{\field_k}|
        \abselemfr(v,w){\models}
        \fproposition{\field_1}\!{\vee}..{\vee}\fproposition{\field_k} \}
        & \mbox{otherwise} \end{array} \right.  \\ \gamma(\abselemmr)
    & = & \abselemmr
  \end{array} \]}

\begin{proof}
  It is straightforward to see that the path-formula returned by
  $\alpha$ for every $v$ and $w$ is greater than or equal to
  $\abselemfr(v,w)$, so that $\gamma(\alpha(\abselemfr)) \sqsupseteqfr
  \abselemfr$.

  The second part of the proof, that $\alpha(\gamma(\abselemmr)) =
  \abselemmr$, follows from the fact that $\abselemmr(v,w)$ is a
  monotone formula, and, if it is neither $\true$ nor $\false$, can be
  transformed into an equivalent conjunctive normal form:
  ${\abselemmr}[*] = (\fproposition{\field_{11}} \vee .. \vee
  \fproposition{\field_{1k^1}}) \wedge .. \wedge
  (\fproposition{\field_{i1}} \vee .. \vee
  \fproposition{\field_{ik^i}})$ for some numbers $i$ and $k^i$.  Each
  of the $(\fproposition{\field_{jj}} \vee .. \vee
  \fproposition{\field_{jk^j}})$ is clearly implied by
  $(\gamma({\abselemmr}))(v,w)$, since the latter is logically
  equivalent to ${\abselemmr}[*](v,w)$.  Therefore,
  $(\alpha(\gamma(\abselemmr)))(v,w)$ is a disjunction of (at least)
  all the disjunctions $(\fproposition{\field_{jj}} \vee .. \vee
  \fproposition{\field_{jk^j}})$, so that it is less than or equal to
  (i.e., implies) $\abselemmr(v,w)$.  If $\abselemmr(v,w)$ is either
  $\true$ or $\false$, the result is easy.
  
  This proves that there is a Galois insertion between $\domfr$ and
  $\dommr$.
\end{proof}

\mylemma{lemma:gc-negation}{
  The following functions define a Galois insertion between $\domfr$
  and $\domnr$: the latter is a strict abstraction of the former.
  \begin{eqnarray*}
    \alpha(\abselemfr) & = & \left\{ v
    \not\rightsquigarrow^{\mathsf{B}} w~|~\forall \field \in
    \mathsf{B}.~\abselemfr(v,w) \models \lnot \fproposition{\field}
    \right\} \\ \gamma(\abselemnr) & = & \lambda
    v,w.~\bigwedge_{\field \in \mathsf{B}}^{} \lnot
    \fproposition{\field}\qquad \mbox{where $\mathsf{B}$ is the
      maximal set s.t.}~v \not\rightsquigarrow^{\mathsf{B}} w \in
    \abselemnr
  \end{eqnarray*}
}

\begin{proof}
  The first part demonstrates that $\gamma(\alpha(\abselemfr))
  \sqsupseteqfr \abselemfr$, i.e., that $\abselemfr(v,w)$ implies
  $(\gamma(\alpha(\abselemfr)))(v,w)$ for every $v$ and $w$.  This
  follows from observing that

  \centerline{$(\gamma(\alpha(\abselemfr)))(v,w) = \bigwedge \{ \lnot
    \fproposition{\field} ~|~ \abselemfr(v,w) \models \lnot
    \fproposition{\field} \}$}
  
  \noindent
  which is clearly implied by $\abselemfr(v,w)$ since it is a
  conjunction made of f-propositions which are all implied by
  $\abselemfr(v,w)$.

  The second part proves that $\alpha(\gamma(\abselemnr)) =
  \abselemnr$.  Let $v \not\rightsquigarrow^{\mathsf{B}} w \in
  \abselemnr$ for some $v$ and $w$.  Then, $(\gamma(\abselemnr))(v,w)$
  will be the conjunction $\pformula$ containing all the negative
  f-propositions for fields in $\mathsf{B}$ (note that $\pformula$
  could contain more literals since $\gamma$ takes the maximal
  $\mathsf{B_0}$ s.t.~$v \not\rightsquigarrow^{\mathsf{B_0}} w \in
  \abselemnr$).  It is easy to see that $\alpha(\gamma(\abselemnr))$
  will contain all statements $v \not\rightsquigarrow^{\mathsf{B'}} w$
  such that $\mathsf{B}' \subseteq \mathsf{B}_0$, thus including $v
  \not\rightsquigarrow^{\mathsf{B}} w$.  On the other hand, let $v
  \not\rightsquigarrow^{\mathsf{B}} w \in \alpha(\gamma(\abselemnr))$:
  this means that $\forall \field \in
  \mathsf{B}.~(\gamma(\abselemnr))(v,w) \models \lnot
  \fproposition{\field}$, i.e., that the conjunction
  $\gamma(\abselemnr)(v,w)$ contains all negated literals $\lnot
  \fproposition{\field}$ for fields in $\mathsf{B}$.  Therefore, there
  must be a statement $v \not\rightsquigarrow^{\mathsf{B}'} w \in
  \abselemnr$, with $\mathsf{B} \subseteq \mathsf{B}'$.  Since
  abstract values in $\domnr$ are closed under $\subseteq$ \cite[Lemma
    4.7]{ScapinSpotoMScThesis}, $v \not\rightsquigarrow^{\mathsf{B}}
  w$ also belongs to $ \abselemnr$, and the proof is complete.
\end{proof}

\mylemma{lemma:gc-requirement}{
  The following functions define a Galois insertion between $\domfc$
  and $\domqc$: the latter is a strict abstraction of the former.
  \begin{eqnarray*}
    \alpha(\abselemfc) & = & \abselemqc~\mbox{with domain}~D =
    \{~v~|~\abselemfc(v) \neq \false \} \\
    & & \abselemqc(v) = \{~\field~|~\abselemfc(v) \models
    \fproposition{\field}~\} \\
    \gamma(\abselemqc) & = & \lambda
    v.\left\{ \begin{array}{l@{\qquad}l}
        \bigwedge_{\field \in \abselemqc(v)}
        \fproposition{\field} & \mbox{if}~v \in
        \dom(\abselemqc) \\
        \false & \mbox{otherwise}
      \end{array} \right.
  \end{eqnarray*}}

\begin{proof}
  By definition, $(\gamma(\alpha(\abselemfc)))(v)$ is $\false$ if and
  only if $\abselemfc(v)$ was $\false$.  On the other hand, if
  $\abselemfc(v) \neq \false$, then $(\gamma(\alpha(\abselemfc)))(v)$
  is the formula $\bigwedge \{
  \fproposition{\field}~|~\abselemfc(v) \models
  \fproposition{\field} \}$, which is clearly implied by
  $\abselemfc(v)$.

  The second part of the proof is also easy:
  $\alpha(\gamma(\abselemqc))$ is equal to $\abselemqc$ since $v \in
  \dom(\alpha(\gamma(\abselemqc)))$ if and only if $\gamma(\abselemqc)
  \neq \false$, which holds if and only if $v \in \dom(\abselemqc)$.
  Moreover, in this case, $(\alpha(\gamma(\abselemqc)))(v)$ is exactly
  the set of fields whose corresponding f-proposition is entailed by
  $(\gamma(\abselemqc))(v)$.  But $(\gamma(\abselemqc))(v)$ comes to
  be $\bigwedge_{\field \in \abselemqc(v)} \fproposition{\field}$, and
  the set of f-propositions entailed by such a formula is exactly
  $\abselemqc(v)$.
\end{proof}

\mylemma{lemma:odotAndPathConcatenation}{ Let $\hpath'$ and $\hpath''$
  be two paths such that the last location of $\hpath'$ is the first
  of $\hpath''$.  Then, $\traverses{\hpath'}{\pformula}$ and
  $\traverses{\hpath''}{\pformulag}$ imply $\traverses{\hpath'
    \hpconcat \hpath''}{\pformula \odot \pformulag}$.  }

\begin{proof}
  The set of fields traversed by $\hpath' \hpconcat \hpath''$ is the
  union of the fields traversed by both sub-paths.  Since, by
  hypothesis, the fields traversed by $\hpath'$ and $\hpath''$
  correspond, respectively, to a model $\finterpretation'$ of
  $\pformula$ and a model $\finterpretation''$ of $\pformulag$, the
  union $\finterpretation = \finterpretation' \cup
  \finterpretation''$, which is a model of $\pformula \odot
  \pformulag$ by definition of $\odot$, is exactly the set of fields
  traversed by $\hpath' \hpconcat \hpath''$, so that the result
  clearly holds.
\end{proof}

\mylemma{lemma:ominus2}{ Let $\hpath$ be $\hpath' \hpconcat \hpath''$;
  let $\traverses{\hpath}{\pformula}$ and
  $\traverses{\hpath'}{\pformulag}$.  Then,
  $\traverses{\hpath''}{\pformula \ominus \pformulag}$.}

\begin{proof}
  Let $\varphi$ be the set of fields which are traversed by $\hpath$,
  and $\varphi'$ and $\varphi''$ be the corresponding field sets for
  $\hpath'$ and $\hpath''$.  Clearly, $\varphi = \varphi' \cup
  \varphi''$, so that $\varphi''$ comes to be the result of removing
  from $\varphi$ \emph{some} of the fields belonging to $\varphi'$.
  Such fields are exactly one of the $X$ mentioned in the definition,
  so that the truth assignment corresponding to $\varphi''$ is
  guaranteed to be a model of $\pformula \ominus \pformulag$.
  Therefore, $\traverses{\hpath''}{\pformula \ominus \pformulag}$.
\end{proof}

\mylemma{lemma:ominus3}{ Let $\hpath$ be
  $\tuple{\ell_0,\ell_1..,\ell_k}$ and $\hpath'$ be
  $\tuple{\ell_1,..,\ell_k}$.  Let the path from $\ell_0$ to $\ell_1$
  traverse $\field$, and $\hpath$ p-satisfy $\pformula$.  Then,
  $\traverses{\hpath'}{\pformula \ominus \onlyfieldsp{\field}}$.}
 
\begin{proof}
  If $\hpath'$ traverses $\field$ (i.e., this field is traversed at
  least twice by $\hpath$), then the set of fields traversed by
  $\hpath'$ is the same as $\hpath$.  Such a set is a model of
  $\pformula \ominus \onlyfieldsp{\field}$ since every model of
  $\pformula$ is also a model of $\pformula \ominus
  \onlyfieldsp{\field}$, so that $\traverses{\hpath'}{\pformula
    \ominus \onlyfieldsp{\field}}$.
   
  On the other hand, if $\hpath'$ does not traverse $\field$, then the
  set of fields traversed by $\hpath'$ is a model of $\pformula$ from
  which $\field$ has been removed.  Due to the definition of
  $\ominus$, such a set is a model of $\pformula \ominus
  \onlyfieldsp{\field}$, so that, again,
  $\traverses{\hpath'}{\pformula \ominus \onlyfieldsp{\field}}$.
\end{proof}

\begin{acks}
  \ouracks
\end{acks}

\def\FMOODS{{\em Int.~Conf.~on Formal Methods for Open Object-Based
    Distributed Systems (FMOODS)}}

\def\VMCAI{{\em Int.~Conf.~on Verification, Model Checking, and
    Abstract Interpretation (VMCAI)}}

\def\POPL{{\em ACM Symposium on Principles of Programming Languages
    (POPL)}}

\def\PLDI{{\em ACM Conf.~on Programming Language Design and
    Implementation (PLDI)}}

\def\CAV{{\em Int.~Conf.~on Computer Aided Verification (CAV)}}

\def\SAS{{\em Static Analysis Symposium (SAS)}}

\def\IJCAR{{\em Int.~Joint Conf.~on Automated Reasoning (IJCAR)}}

\def\TOPLAS{{\em ACM Transactions on Programming Languages and
    Systems\/}}

\def\LICS{{\em IEEE Symposium on Logic in Computer Science (LICS)}}

\def\CC{{\em Int.~Conf.~on Compiler Construction (CC)}}

\def\PASTE{{\em ACM Workshop on Program Analysis For Software Tools
    and Engineering (PASTE)}}

\def\TOCL{{\em ACM Transactions on Computational Logic\/}}

\bibliographystyle{ACM-Reference-Format-Journals}

\received{Month Year}{Month Year}{Month Year}

\end{document}